\newcommand{\doctitle}{The Mass Shell of the Nelson Model without Cut-Offs}

\def\arxiv{1}		
\def\todos{0}		
\newcommand{\ifarxiv}[2]{\if\arxiv 1 {#1}\else{#2}\fi}

\if\arxiv 1 
  \documentclass[12pt]{article}
  \usepackage[letterpaper, left=1in,
                           rmargin=1in,
                           top=1in,
                           bottom=1in]{geometry}
  \usepackage{ifpdf}
  \usepackage{amsmath, amsfonts, amssymb, amsthm,
              bbm, enumerate, longtable, xrf, txfonts}
  \usepackage[usenames,dvipsnames]{color}
  
  \usepackage{fancyhdr}
  \fancypagestyle{myheadsfoots}{
    \fancyhead[L]{\small\textit{\doctitle}}
    \fancyhead[R]{\small\thepage}
    \fancyfoot{}
  }
  
  \if\todos 1
    \usepackage[colorinlistoftodos, textwidth=.9in]{todonotes}
  \else
    \usepackage[disable]{todonotes}
  \fi  
      
\else 
\fi

\ifpdf
  \usepackage[
              pdfpagemode=UseOutlines,
              pdfstartview=FitH,
              bookmarks=true,
              bookmarksopen=true,
              bookmarksnumbered=true,
              bookmarkstype=toc,
              colorlinks=true,
              linkcolor=blue,
              citecolor=blue,
              urlcolor=blue]{hyperref}
\fi

\newtheorem{theorem}{Theorem}[section]
\newtheorem{lemma}[theorem]{Lemma}
\newtheorem{corollary}[theorem]{Corollary}
\newtheorem{proposition}[theorem]{Proposition}
\newtheorem{definition}[theorem]{Definition}

\newtheorem{remark}[theorem]{Remark}

\renewcommand{\cal}[1]{{\mathcal{#1}}}                    
\newcommand{\bb}[1]{{\mathbb{#1}}}                      
\newcommand{\id}[1]{\mathbbm{1}_{#1}}                   
\newcommand{\charf}[1]{\mathbbm{1}_{#1}}                

\newcommand{\braket}[1]{\left\langle #1\right\rangle}  

\newcount\diffpages
\newcommand{\ppref}[1]                                  
{%
  \ifthenelse{\getpagerefnumber{#1}=\thepage}%
  {}{{\tiny{p.\ifpdf\pageref*{#1}\else\pageref{#1}\fi}}}%
}
\newcommand{\pref}[1]{\ref{#1}
}               
\newcommand{\eqn}[1]{(\pref{#1})}                       
\newcommand{\thm}[1]{Theorem \pref{#1}}                 
\newcommand{\lem}[1]{Lemma \pref{#1}}                   
\newcommand{\cor}[1]{Corollary \pref{#1}}               
\newcommand{\dfn}[1]{Definition \pref{#1}}				
\newcommand{\sct}[1]{Section \pref{#1}}                 

\newcommand{\slice}[3]{{#1|_{#3}^{#2}}}					
\newcommand{\spec}[1]{\operatorname{Spec}\left(#1\right)}
\newcommand{\gap}[1]{\operatorname{Gap}\left(#1\right)} 
\newcommand{\restrict}[1]{\upharpoonright{#1}}			

\renewcommand{\Im}{\operatorname{Im}}
\renewcommand{\Re}{\operatorname{Re}}

\begin{document}

\title{\doctitle}

\author{S. Bachmann\thanks{svenbac@math.ucdavis.edu}, D.-A. Deckert\thanks{deckert@math.ucdavis.edu}, A. Pizzo\thanks{pizzo@math.ucdavis.edu}\\
  \small
  Department of Mathematics,
  University of California,\\
  \small
  One Shields Avenue,
  Davis, CA 95616, USA}

\date{\small\today}

\maketitle

\begin{abstract}
The massless Nelson model describes non-relativistic, spinless quantum particles interacting with a relativistic, massless, scalar quantum field. The interaction is linear in the field. We analyze its one particle sector. First, we construct the renormalized mass shell of the non-relativistic particle for an arbitrarily small infrared cut-off that turns off the interaction with the low energy modes of the field. No ultraviolet cut-off is imposed. Second, we implement a suitable  Bogolyubov  transformation of the Hamiltonian in the infrared regime. This transformation depends on the total momentum of the system and is non-unitary as the infrared cut-off is removed. For the transformed Hamiltonian we construct the mass shell in the limit where both the ultraviolet and the infrared cut-off are removed. Our approach is constructive and leads to explicit expansion formulae which are amenable to rigorously control the S-matrix elements.
\begin{center}
\begin{tabular}{rp{4in}}
 \textbf{Keywords:} & Multiscale Perturbation Theory, Nelson Model, Renormalization, Ultraviolet Divergence, Infrared Catastrophe.\\
 \textbf{Grants:} & D.-A.D. gratefully acknowledges funding by the DAAD. S.B. is supported by the NFS Grant \#DMS-0757581. A.P. is supported by the NSF grant \#DMS-0905988.
\end{tabular}
\end{center}
\end{abstract}

\tableofcontents

\makeatletter
\providecommand\@dotsep{5}
\makeatother
\listoftodos\relax

\if\arxiv 1
  \pagestyle{myheadsfoots}
\fi

\section{Introduction and Definition of the Model}

We study the mass shell of a non-relativistic spinless quantum particle interacting with the quantized field of relativistic, massless, scalar bosons, where the interaction is linear in the field. This model originated as an effective description of the  interaction between non-relativistic nucleons and mesons. It is usually referred to as `Nelson model' since E. Nelson (see \cite{nelson_interaction_1964}) showed how to remove the ultraviolet cut-off that turns off the interaction with the high frequency modes of the field. The limiting Hamiltonian is  defined  starting from the quadratic form associated with the so-called Gross transformed Hamiltonian. The latter is obtained from the Nelson Hamiltonian through a unitary dressing transformation \cite{gross_particle-like_1962} after subtracting a constant which is divergent in the ultraviolet (UV) limit. This means that only a ground state energy renormalization is necessary in order to define the local interaction.  This model for only one nucleon is known as the  one particle  sector of the translation invariant Nelson model.

In recent years this model has been extensively studied with regard to quantum electrodynamics (QED). In fact, when the bosons are massless particles (i.e. `scalar photons') the model can be seen as a scalar version of the effective theory (non-relativistic QED) that describes a non-relativistic electron interacting with the quantized radiation field. In the study of the translation invariant, massless Nelson model an ultraviolet cut-off of the order of the rest mass energy of the electron is usually imposed. Otherwise relativistic corrections to the electron dynamics and electron-positron pair creation should be taken into account. In spite of these simplifications, the massless Nelson model gives non-perturbative insights on the infrared properties of QED.  

It is an interesting mathematical problem to clarify whether the results concerning the infrared region, which have been obtained in presence of an ultraviolet cut-off, can be extended to the `renormalized'  Nelson model (i.e. without an ultraviolet cut-off). As presented in \cite{hirokawa_ground_2005} these questions do not in general have a straightforward answer.

For the one particle sector of the renormalized Nelson model the study of the mass shell was carried out by Cannon few years after the appearance of Nelson's paper. In \cite{cannon_quantum_1971} it is proven that a perturbed mass shell exists for sufficiently small values of the coupling constant $g$ and in the spectral region $(E,P)$ for $|P|<1$. Here, $E$ and $P$ are the spectral variables of the Hamiltonian and of the total momentum operator, respectively. In fact, starting from translation invariance, one considers the natural decomposition of the Hilbert space on the spectrum of the total momentum operator and studies the existence of the ground state of the fiber Hamiltonians $H_P$ for $|P|<1$. In his paper, Cannon relies on the spectral gap of the fiber Hamiltonians induced by a meson mass.  The mass shell of the nucleon is then defined by analytic perturbation theory of the ground state eigenvector fiber by fiber  for $|P|<1$ and sufficiently small $g$. The interaction is in fact a small perturbation of type B -- i.e. in the form sense -- with respect to the free Hamiltonian. For this type of perturbation it is in principle possible to control the perturbed spectral projection and to give a meaning to the formal expansion of the ground state vector of the perturbed Hamiltonian. The price for this is a very cumbersome formula (see \cite{kato_perturbation_1995}) making his result almost intractable for applications to scattering theory. As a matter of fact, no explicit expression for the perturbed mass shell is provided in \cite{cannon_quantum_1971}.

Finally, for the massless Nelson model, the result concerning the existence of the mass shell was extended by Fr\"ohlich to arbitrarily small infrared cut-off with no restriction on the coupling constant. The method used in \cite{fraehlich_infrared_1973} is based on a lattice approximation of the boson momentum space which is eventually removed, a technique inspired by earlier works of Glimm and Jaffe. However, Fr\"ohlich's expression for the fiber eigenvectors is only implicit. In recent years the $P$-dependence of the ground state energy in the  massless Nelson model and in non-relativistic QED has been studied in presence of an ultraviolet regularization. \cite{bach_renormalized_2007} and \cite{chen_infrared_2008} use the isospectral renormalization group whereas \cite{abdesselam_analyticity_2010} relies on statistical mechanics methods. \\

\paragraph{We accomplish three main goals:}
(1) By using a multiscale technique for small values of the coupling constant and for a fixed infrared cut-off $\kappa>1$ (in units where the electron mass $m$, the Planck's constant $\hbar$, and the speed of light $c$ all equal one) we first derive the results by Cannon for the massless Nelson model.  Rather than using regular perturbation theory for quadratic forms we employ a multiscale technique for operators inspired by \cite{pizzo_one-particle_2003}. Our construction  yields  more explicit expressions for the `renormalized' mass shell. In particular, they  are amenable to rigorously control the S-matrix elements under the removal of the UV cut-off and to compare them with physicists' perturbation formulae.

(2) We then show how to construct the mass shell for the renormalized model when the interaction is extended to frequency ranges down to an arbitrarily small infrared cut-off. This result at a small but fixed value of the coupling constant $g$ is beyond the reach of the method employed by Cannon \cite{cannon_quantum_1971} because the spectral gap shrinks to zero as the infrared cut-off is removed.

(3) The final part of our analysis concerns the properties of the mass shell in the infrared limit where it is well-known that no \emph{proper} mass shell is present, a fact usually referred to as the \emph{infrared catastrophe}. Following the strategy developed in \cite{pizzo_one-particle_2003}, we implement a suitable Bogolyubov transformation for the field variables corresponding to frequencies below the threshold $\kappa>1$. In contrast to Gross' dressing  this transformation depends on the $P$-fiber and is not unitary in the infrared limit. Fiber by fiber,  we obtain a transformed Hamiltonian where the interaction is not linear in the field anymore both because of the Gross transformation in the UV region (frequencies larger than $\kappa$)  and because of the infrared dressing transformation  (frequencies smaller than $\kappa$). Each  transformed  Hamiltonian has a ground state in the infrared limit, the construction of which requires a delicate control of the interplay between high and low frequency modes. The control of the mass shell associated with these \emph{unphysical} fiber Hamiltonians is crucial to analyze the infraparticle behavior of the renormalized electron in the massless Nelson model and to provide an asymptotic expansion for the scattering amplitudes in `Compton scattering', free from both ultraviolet and infrared divergences.

\paragraph{Definition of the model.}   The Hilbert space of the model is  
\begin{align*}
  \cal H:=L^2(\bb R^3,\bb C;dx)\otimes \cal F(h),
\end{align*}
  where $\cal F(h)$ is the Fock space of scalar bosons
\begin{align*}
  \cal F(h):=\bigoplus_{j=0}^{\infty} \cal F^{(j)},
  && \cal F^{(0)}:=\bb C,
  && \cal F^{j\geq 1}:={\bigodot}_{l=1}^j h, && h:=L^2(\bb R^3,\bb C;dk),
 \end{align*}
 where $\odot$ denotes the symmetric tensor product. 
  Let $a(k),a^*(k)$ be the usual Fock space annihilation and creation operators satisfying the canonical commutation relations (CCR)
\begin{align*}
  [a(k),a^*(l)]=\delta(k-l), \qquad [a(k),a(l)]=[a^*(k),a^*(l)]=0.
\end{align*}
The kinematics of the system is described by: (a)  The position $x$ and the momentum $p$ of the non-relativistic particle that satisfy the Heisenberg commutation relations. (b) The  scalar field $\Phi$ and its conjugate momentum where  
 \begin{align*}
 \Phi(y):= \int  dk\;\rho(k)\left(a(k)e^{iky}+a^*(k)e^{-iky}\right), \qquad   \rho(k):=\frac{1}{(2\pi)^{3/2}}\frac{1}{\sqrt{2\omega(k)}}, \qquad \omega(k):=|k|.
 \end{align*}
 The dynamics is generated  by the Hamiltonian of the Nelson model, 
\begin{align*}
 \slice{H}{\Lambda}{\tau} := \frac{p^2}{2} + H^f +g\slice{\Phi}{\Lambda}{\tau} (x)
\end{align*}
 where
\begin{align*}
 H^f:=\int dk \;\omega(k)a^*(k)a(k)
\end{align*}
is the free field Hamiltonian,  and
\begin{align}\label{eqn:interaction}
  g\slice{\Phi}{\Lambda}{\tau}(x):=g \int_{\cal B_{\Lambda}  \setminus \cal B_{\tau}} dk\;\rho(k)\left(a(k)e^{ikx}+a^*(k)e^{-ikx}\right), \qquad  \rho(k):=\frac{1}{(2\pi)^{3/2}}\frac{1}{\sqrt{2\omega(k)}},
\end{align}
is  the interaction term for $0\leq \tau<\Lambda<\infty$;
here $g\in \mathbb{R}$ is the coupling  constant and for the domain of integration we use the notation $\cal B_{\sigma}:=\{k\in\bb R^3\,|\, |k|< \sigma\}$ for any $\sigma>0$.  Note that for $\Lambda=\infty$ the formal expression of the interaction $\slice{\Phi}{\Lambda}{\kappa}$ is not a well-defined operator on $\cal H$ because the form factor $\rho(k)$ is not square integrable.  

We briefly  recall  some well-known facts about this model. For $0\leq\tau<\Lambda<\infty$  the operator $\slice{H}{\Lambda}{\tau}$ is self-adjoint and its domain coincides with the one of $H_{0}:=\frac{p^2}{2} + H^f$ (see also Proposition 1.1 below). The total momentum operator of the system is
\begin{align*}
  P:=p+ P^f:= p+\int dk\;k\,a^*(k)a(k)
\end{align*}
where $P^f$ is the field momentum. Due to translational invariance of the system the Hamiltonian and the total momentum operator commute. Hence, the Hilbert space $\cal H$ can be decomposed on the joint spectrum of the three components of the total momentum operator, i.e.
\begin{align*}
 \cal H=\int^{\oplus}dP\;\cal H_{P}
\end{align*}
where $\cal H_{P}$ is a copy of the Fock space $\cal F$ carrying the (Fock) representation corresponding to annihilation and creation operators
\begin{align*}
  b(k):=a(k)e^{ikx}, && b^*(k):=a^*(k)e^{-ikx}.
\end{align*}
We will use the same symbol $\cal F$ for all Fock spaces. The fiber Hamiltonian can be expressed as
 \begin{align*}
 \slice{H_{P}}{\Lambda}{\tau}:=\frac{1}{2}\left(P-P^f\right)^2 + H^f + g\int_{\cal B_{\Lambda} \setminus \cal B_{\tau}} dk\;\rho(k)\left(b(k)+b^*(k)\right).
\end{align*}

By construction, the fiber Hamiltonian maps its domain in $\cal H_{P}$ into $\cal H_{P}$. Finally, for later use we define 
\begin{align}\label{eqn:free and diff hamiltonian}
 H_{P,0}:=\frac{(P-P^f)^2}{2}+H^f, && \slice{\Delta H_{P}}{\Lambda}{\tau}:=\slice{H_{P}}{\Lambda}{\tau}-H_{P,0}.
\end{align}

\paragraph{The Gross transformation.}   We use a frequency
\[
1<\kappa<2
\]
to separate  the ultraviolet and the infrared regimes. The renormalization of the Hamiltonian  must cure the divergence which appears in the second order correction to the ground state energy as $\Lambda \to \infty$.  
This logarithmically divergent term 
\begin{align}\label{eqn:selfenergy}
\slice{V_{\mathrm{self}}}{\Lambda}{\kappa}:=-\frac{g^2}{[2(2\pi)^3]}\,\int_{\cal B_{\Lambda}\setminus \cal B_{\kappa}}\,dk\, \frac{1}{|k|\left[\frac{|k|^2}{2}+|k|\right]}
\end{align}
can be separated from the rest of the Hamiltonian by a Bogolyubov transformation $e^{-\slice{T}{\Lambda}{\kappa}}$, acting on all frequencies above $\kappa$, whose skew-adjoint generator is given by
\begin{align}\label{def-beta(k)}
\slice{T}{\Lambda}{\kappa}:=\int_{\cal B_{\Lambda} \setminus \cal B_{\kappa}} dk\; \beta(k)\left(b(k)-b^*(k)\right),\qquad \beta(k):=-g\frac{\rho(k)}{\frac{|k|^2}{2}+\omega(k)}.
\end{align}
Note that for any $1<\kappa<\Lambda\leq\infty$, the operators $\slice{T}{\Lambda}{\kappa}$, $\slice{T^*}{\Lambda}{\kappa}$ are well-defined on $D(H_{P,0})$.
For $1<\kappa<\Lambda<\infty$ the Hamiltonian $\slice{H_{P}}{\Lambda}{\kappa}$ transforms as follows:
\begin{align}
  \slice{H'_{P}}{\Lambda}{\kappa}&:=e^{\slice{T}{\Lambda}{\kappa}}\slice{H_{P}}{\Lambda}{\kappa}e^{-\slice{T}{\Lambda}{\kappa}}-\slice{V_{\mathrm{self}}}{\Lambda}{\kappa}\label{eqn:Vself renorm}\\
\begin{split}
 &=\frac{1}{2}\left(P-P^f\right)^2 + H^f+\frac{1}{2}[(\slice{B}{\Lambda}{\kappa})^2+(\slice{B^*}{\Lambda}{\kappa})^2]+\slice{B^*}{\Lambda}{\kappa}\cdot \slice{B}{\Lambda}{\kappa}\\
 &\quad-(P-P^f)\cdot \slice{B}{\Lambda}{\kappa}-\slice{B^*}{\Lambda}{\kappa}\cdot (P-P^f)
\end{split}\label{eqn:G trafo Hamiltonain}
\end{align}
where
\begin{align}\label{eqn:B}
\slice{B}{\Lambda}{\kappa}:= \int_{\cal B_{\Lambda}\setminus \cal B_{\kappa}} dk\;k\beta(k)b(k).
\end{align}
It is important to note that the operator equality \eqn{eqn:G trafo Hamiltonain} holds on $D(H_{P,0})$ as proven in \cite[Lemma 3]{nelson_interaction_1964}.
 In the following sections we will study the renormalized Hamiltonian
\begin{equation}
  \slice{H'_{P}}{\Lambda}{\kappa}+g\slice{\Phi}{\kappa}{\tau}
  \end{equation}
   
 The proofs of \cite[Lemma 2 and 3]{nelson_interaction_1964} imply:
\begin{proposition}
For $0\leq\tau<\Lambda<\infty$, the operators $\slice{H_{P}}{\Lambda}{\tau}$ and $\slice{H'_{P}}{\Lambda}{\kappa}+g\slice{\Phi}{\kappa}{\tau}$ are self-adjoint and their domain coincide with the one of $H_{P,0}$.
\end{proposition}
By \cite[Main Theorem]{nelson_interaction_1964} there exists an ultraviolet renormalized Hamiltonian:
\begin{theorem}\label{thm:nelson}
 For all $\tau\geq 0$, there is a unique self-adjoint operator $\slice{H_P}{\infty}{\tau}$ on $\cal F$ that generates the unitary group defined by
\begin{align*}
 e^{-it\slice{H_P}{\infty}{\tau}}:=\operatorname{s-lim}\limits_{\Lambda\to\infty}e^{-it(\slice{H_{P}}{\Lambda}{\tau}-\slice{V_{\mathrm{self}}}{\Lambda}{\kappa})}, \quad t\in\bb R.
\end{align*} 
 The domain of $\slice{H_P}{\infty}{\tau}$ is a dense subset of the domain of $H_{P,0}^{1/2}$, and $\slice{H_P}{\infty}{\tau}$ is bounded from below.
\end{theorem}

However, we will not make use of \thm{thm:nelson}. In the case of $|P|<P_{max}$ defined in (\ref{Pmax})   and for sufficiently small $|g|$ this result will follow from our multiscale analysis. 

\section{Main Results}\label{sec:main results}

 Since the particle is non-relativistic we restrict the total momentum to the ball
\begin{equation}
|P|\leq P_{\mathrm{max}}:=\frac{1}{4}.\label{Pmax}
\end{equation} 
  
\paragraph{The ultraviolet and infrared scaling.} We shall introduce a scaling that divides the interaction term into slices of boson momenta for which, step by step, we apply analytic perturbation theory. In the ultraviolet regime, this scaling is defined by the sequence
\begin{align*}
\sigma_n:=\kappa\beta^n, && 1<\beta, && n\in\bb N ,
\end{align*}
while in the infrared regime we use
\begin{align*}
\tau_m:=\kappa\gamma^m, && 0<\gamma<\frac{1}{2}, && m\in\bb N.
\end{align*}

With respect to these scalings we shall use the following notation for Hamiltonians and Fock spaces:

\begin{center}
\renewcommand{\arraystretch}{1.75}
\begin{tabular}{c c c c}
IR& UV& Hamiltonian & Fock space\\
\hline
\hline
$\kappa$ & $\sigma_n$ & $\slice{H'_{P}}{n}{0}:=\slice{H'_P}{\sigma_n}{\kappa}$ & $\slice{\cal F}{n}{0}:=\cal F(L^2(\cal B_{\sigma_n}\setminus\cal B_{\kappa}))$\\
\hline
$\tau_m$ & $\sigma_n$ & $\slice{H'_{P}}{n}{m}:=\slice{H'_{P}}{n}{0}+g\slice{\Phi}{\kappa}{\tau_m}$ & $\slice{\cal F}{n}{m}:=\cal F(L^2(\cal B_{\sigma_n}\setminus\cal B_{\tau_m}))$\\
\hline
\end{tabular}
\end{center}
\vskip.3cm

\noindent  The normalized vacuum vector in each of these Fock spaces is denoted by the same symbol $\Omega$. We shall exclusively use the index $n$ to denote the ultraviolet cut-off $\sigma_n$ and the index $m$ to denote the infrared cut-off $\tau_m$, e.g. 
\begin{align*}
 \slice{\cal F}{n}{n-1}:=\cal F(L^2(\cal B_{\sigma_n}\setminus\cal B_{\sigma_{n-1}})), && \slice{\cal F}{m-1}{m}:=\cal F(L^2(\cal B_{\tau_{m-1}}\setminus\cal B_{\tau_m})).
\end{align*}
For a vector $\psi$ in $\slice{\cal F}{n-1}{0}$ and an operator $O$ on $\slice{\cal F}{n-1}{0}$ we shall use the same symbol to denote the vector $\psi\otimes\Omega$ in $\slice{\cal F}{n}{0}$ and the operator $O\otimes\id{\slice{\cal F}{n}{n-1}}$ on $\slice{\cal F}{n}{0}$, respectively. 

Moreover, the Fock space slices and the related interaction terms are given by
\begin{center}
\renewcommand{\arraystretch}{1.75}
\begin{tabular}{c | c c c}
& Slice & Interaction & Fock space\\
\hline
\hline
UV & $[\sigma_{n-1},\sigma_{n})$ & $\slice{\Delta H'_P}{n}{n-1}:=\slice{H'_{P}}{n}{0}-\slice{H'_{P}}{n-1}{0}$ & $\slice{\cal F}{n}{n-1}$\\
\hline
IR & $(\tau_{m},\tau_{m-1}]$ & $g\slice{\Phi}{m-1}{{m}}:=g\slice{\Phi}{\tau_{m-1}}{\tau_{m}}$ & $\slice{\cal F}{m-1}{m}$\\
\hline
\end{tabular}
\end{center}
\vskip.3cm
Similarly we shall use $\slice{}{n}{m}$, $\slice{}{n}{n-1}$, $\slice{}{m-1}{m}$ instead of $\slice{}{\sigma_n}{\tau_m}$, $\slice{}{\sigma_n}{\sigma_{n-1}}$, $\slice{}{\tau_{m-1}}{\tau_m}$, respectively, as short-hand notation to denote the range of boson momenta associated with the interaction.

For a self-adjoint operator $A$ which is bounded from below we define the spectral gap as
\begin{align*}
 \gap{A}:=\inf\{\spec A\setminus\{\inf\spec A\}\}-\inf\spec A.
\end{align*}
 Moreover, we denote 
\begin{align}\label{eqn:energies}
 \slice{E_P}{n}{m}:=\inf\spec{ \slice{H_P}{n}{m} \restrict{\slice{\cal F}{n}{m}}}, && \slice{E'_P}{n}{m}:=\inf\spec{ \slice{H'_P}{n}{m} \restrict{\slice{\cal F}{n}{m}}}=\slice{E_P}{n}{m}-\slice{V_{\mathrm{self}}}{n}{0}
\end{align}
where $\spec{A\restrict{X}}$ denotes the spectrum of the linear operator $A$ restricted to the subspace $X$. If $\slice{E'_P}{n}{m}$ is a non-degenerate eigenvalue of the Hamiltonian $\slice{H'_P}{n}{m}$ we shall denote a (possibly unnormalized) corresponding eigenvector by $\slice{\Psi'_P}{n}{m}$. In this situation we have
\begin{align*}
 \gap{\slice{H'_P}{n}{m}\restrict{\slice{\cal F}{n}{m}}} = \inf_{\psi\perp\slice{\Psi'_P}{n}{m}}\braket{\slice{H'_P}{n}{m}-\slice{E'_P}{n}{m}}_\psi
\end{align*}
where the infimum is taken over the vectors $\psi$ in the domain of $\slice{H'_P}{n}{m}\restrict{\slice{\cal F}{n}{m}}$, and we have used the notation
\[
  \braket{A}_\psi=\frac{\braket{\psi,A\psi}}{\braket{\psi,\psi}}
\]
for any operator $A$ and $\psi\in D(A)$.

\paragraph{The Mass Shell of $\slice{H'_P}{\infty}{0}$.} The multiscale perturbation theory that we use here relies on the control of the spectral gap as more and more slices of the interaction Hamiltonian are added.  In the construction of  the mass shell eigenvectors   one observes a major difference between removing the ultraviolet and the infrared cut-off. In the infrared limit the main problem is that the gap closes and the infimum of the spectrum is not an eigenvalue anymore (see \cite{pizzo_one-particle_2003}). In the ultraviolet limit the main problem is that the whole spectrum moves towards $-\infty$. The latter is caused by the well-known logarithmic divergence  in  \eqn{eqn:selfenergy}. In order to gain control on the gap it is necessary to extract this divergent term which, as it is also well-known, can be accomplished via the Gross transformation. At first, we shall therefore apply the multiscale perturbation theory to the Gross transformed Hamiltonians $\slice{H'_{P}}{n}{0}$, and then use unitarity to inherit all results for the back-transformed Nelson Hamiltonians
\begin{align*}
 \slice{H_{P}}{n}{0}:=e^{-\slice{T}{n}{0}}\slice{H'_{P}}{n}{0}e^{\slice{T}{n}{0}}+\slice{V_{\mathrm{self}}}{n}{0}, \quad n\in\bb N.
\end{align*}

The iterative analytic perturbation theory, which was successfully applied for the infrared regime \cite{pizzo_one-particle_2003}, can be adapted to the ultraviolet regime using the following induction:
 
Suppose that, for a given and appropriately chosen real sequence $(\xi_n)_{n\in\bb N}$ bounded from below by a positive constant, we know that the following holds for the $(n-1)$-th step of the induction:
\begin{enumerate}[(i)]
 \item $\slice{\Psi'_{P}}{n-1}{0}$ is the unique ground state of $\slice{H'_{P}}{n-1}{0}$ with energy $\slice{E'_P}{n-1}{0}$.
 \item $\gap{\slice{H'_{P}}{n-1}{0}\restrict{\slice{\cal F}{n-1}{0}}}\geq \xi_{n-1}$.
\end{enumerate}
In order to show the induction step $(n-1)\Rightarrow n$, we first estimate the new spectral gap while adding the slice $\slice{\cal F}{n}{n-1}$ of boson Fock space without modifying the Hamiltonian. An a priori variational argument yields $\gap{\slice{H'_{P}}{n-1}{0}\restrict{\slice{\cal F}{n}{0}}}\geq \xi_{n-1}$. With this at hand we apply analytic perturbation theory \`a la Kato to construct the ground state of $\slice{H'_{P}}{n}{0}\restrict{\slice{\cal F}{n}{0}}$.   More precisely, we show that the Neumann series of the resolvent 
\begin{align}\label{eqn:res}
\frac{1}{\slice{H'_{P}}{n}{0}-z}=\frac{1}{\slice{H'_{P}}{n-1}{0}-z}\sum_{j=0}^{\infty}[-\Delta H'_{P}|^{n}_{n-1}\frac{1}{\slice{H'_{P}}{n-1}{0}-z}]^j
\end{align}
is well-defined for all $z$ in the domain
\begin{align*}
\frac{1}{2}\xi_{n}\leq |\slice{E'_P}{n-1}{0}-z|\leq \xi_{n}<\xi_{n-1}.
\end{align*}
Step by step we show the convergence of the Neumann series  for a sufficiently small $|g|$ (and $\beta$ sufficiently close to one) but  uniformly in $n$. In the control of the resolvent  in \eqn{eqn:res} a convenient definition of $(\xi_n)_{n\in\bb N}$ turns out to be crucial.   Kato's perturbation theory ensures the existence of a projection $\slice{\cal Q_{P}'}{n}{0}$ onto the unique ground state $\slice{\Psi'_P}{n}{0}$ with eigenvalue $\slice{E'_P}{n}{0}$. Since an a priori variational argument yields $\slice{E'_P}{n}{0}\leq \slice{E'_P}{n-1}{0}$, we conclude that $\gap{\slice{H'_{P}}{n}{0}\restrict{\slice{\cal F}{n}{0}}}\geq \xi_{n}$.

This way we construct a convergent sequence of ground states corresponding to $\slice{H'_{P}}{n}{0}$, $n\in\bb N$,
\begin{align*}
 \slice{\Psi'_P}{n}{0} := \slice{\cal Q_{P}'}{n}{0}\slice{\cal Q_{P}'}{n-1}{0}\cdots \slice{\cal Q_{P}'}{1}{0}\Omega
\end{align*}
where $\Omega$ is the ground state of $H'_{P,0}$. The projections $\slice{\cal Q'_P}{n}{0}$ will be given explicitly in \eqn{eqn:projection}. Finally, the unitarity of the Gross transformation implies that
\begin{align*}
 \slice{\Psi_P}{n}{0}:=e^{-\slice{T}{n}{0}}\slice{\Psi'_P}{n}{0},\;n\in\bb N,
\end{align*}
is a sequence of ground states of $\slice{H_{P}}{n}{0}$ that also converges, say to a $\slice{\Psi_P}{\infty}{0}\in\cal F$. Furthermore, we prove the convergence of $\slice{H'_{P}}{n}{0}$ in the norm resolvent sense to a limiting Hamiltonian $\slice{H'_{P}}{\infty}{0}$, the unique ground state of which is $\slice{\Psi'_P}{\infty}{0}$. Precisely, we prove:
\begin{theorem}\label{thm:main uv}
Let $|P|\leq P_{\mathrm{max}}$. There is a constant $g_{\mathrm{max}}>0$  such  such that for all $|g|<g_{\mathrm{max}}$ the following holds true:
\begin{enumerate}[(i)]
 \item  The sequence of operators $(\slice{H_{P}}{n}{0}-\slice{V_{\mathrm{self}}}{n}{0})_{n\in\bb N}$ converges in the norm resolvent sense to a self-adjoint operator $\slice{H_{P}}{\infty}{0}$ acting on $\cal F$.
 \item The limit $\slice{\Psi_P}{\infty}{0}:=\lim_{n\to\infty}\slice{\Psi_P}{n}{0}$ exists in $\cal F$ and is non-zero.
 \item $\slice{E_P}{\infty}{0}:=\lim_{n\to\infty}(\slice{E_P}{n}{0}-\slice{V_{\mathrm{self}}}{n}{0})$ exists.
 \item $\slice{E_P}{\infty}{0}$ is the non-degenerate ground state energy of the Hamiltonian $\slice{H_{P}}{\infty}{0}$ with corresponding ground state $\slice{\Psi_P}{\infty}{0}$. Moreover,   the spectral gap of $\slice{H_{P}}{\infty}{0}\restrict{\slice{\cal F}{\infty}{0}}$ is  bounded from below by $\frac{1}{16}\kappa$
.
\end{enumerate}
\end{theorem}

\paragraph{The Mass Shell of $\slice{H'_P}{\infty}{m}$ for $m\in\bb N$.} Starting from the ground states $\slice{\Psi'_P}{n}{0}$ of the Hamiltonian $\slice{H'_{P}}{n}{0}$, we continue to add interaction slices $g\slice{\Phi}{\tau_{m-1}}{\tau_m}$, $m\in\bb N$, now below the frequency $\kappa$ and construct the family of 
ground states $\slice{\Psi'_P}{n}{m}$ of the Hamiltonians $\slice{H'_P}{n}{m}$ with eigenvalue $\slice{E'_P}{n}{m}$, i.e.
\[
  \slice{H'_P}{n}{m}\slice{\Psi'_P}{n}{m}=\slice{E'_P}{n}{m}\slice{\Psi'_P}{n}{m}.
\]
For arbitrarily large but fixed $m\in\bb N$, we prove results analogous to \thm{thm:main uv}: Norm resolvent convergence of $(\slice{H'_P}{n}{m})_{n\in\bb N}$ is shown in \lem{lem:ir res conv}. The existence of $\slice{\Psi'_P}{\infty}{m}$, $m\in\bb N$, is shown in \thm{thm:ir induction}. In particular, the spectral gap of $\slice{H'_P}{n}{m}$ is bounded from below by a constant times $\tau_m$ uniformly for all $n\in\bb N\cup\{\infty\}$. This is proven in \lem{lem:ir gap}.

\paragraph{The Mass Shell of $\slice{H^{W'}_P}{\infty}{\infty}$.} As it is well-known (see \cite{fraehlich_infrared_1973,pizzo_one-particle_2003}), for every $n\in\bb N\cup\{\infty\}$ the ground state $\frac{\slice{\Psi'_P}{n}{m}}{\|\slice{\Psi'_P}{n}{m}\|}$ weakly converge to zero as $m\to\infty$. This is linked to the infamous infrared catastrophe problem in QED.  In fact, in the infrared limit the interaction turns out to be \emph{marginal} according to renormalization group terminology.  On the other hand it was proven in \cite{fraehlich_infrared_1973} that
\begin{align}\label{eqn:froehlich}
 b(k)\slice{\Psi'_P}{n}{m}=g\;\rho(k)\frac{1}{\slice{E'_P}{n}{m}-|k|-\slice{H'_{P-k}}{n}{m}}\slice{\Psi'_P}{n}{m}
\end{align}
which implies that
\begin{equation}\label{eqn:coherent}
 b(k)\slice{\Psi'_P}{n}{m}\approx \alpha_m(\nabla\slice{E'_P}{n}{m},k)\slice{\Psi'_P}{n}{m},\qquad
\alpha_m(Q,k):= -g\;\frac{\rho(k)}{\omega(k)}\frac{\charf{\cal B_\kappa\setminus\cal B_{\tau_m}}(k)}{1-\widehat k\cdot Q}
\end{equation}
up to higher order terms as $k\to 0$. This motivates a strategy to analyze the infrared limit by using the Bogolyubov transformation $W_m(\nabla\slice{E'_P}{n}{m})$ defined as follows:
for $Q\in\bb R^3$, $|Q|<1$,
\begin{align}\label{eqn:W trafo}
 W_m(Q)\;b^{\#}(k)\;W_m(Q)^* &:= b^{\#}(k) + \alpha_m(Q,k)\quad b^{\#}(k)=b(k),b^{*}(k).
\end{align}
Instead of studying $\slice{H'_P}{n}{m}$ directly one considers the transformed Hamiltonian
\begin{equation}
\slice{H_P^{W'}}{n}{m} := W_m(\nabla\slice{E'_P}{n}{m})\;\slice{H'_P}{n}{m}\;W_m(\nabla\slice{E'_P}{n}{m})^*.
\end{equation}
Note that the transformation acts non-trivially only on boson momenta below $\kappa$. For any finite $m$, the operator $W_m(Q)$ is unitary but this property does not hold anymore in the limit $m\to\infty$.  Furthermore, for  $Q\neq Q'$ the function $\alpha_m(Q,k)-\alpha_m(Q',k)$ is not square integrable as $m\to\infty$. 

Most importantly, the interaction term
\begin{equation}
\slice{H_P^{W'}}{n}{m}-H_{P,0}
\end{equation}
of the transformed Hamiltonian is now \emph{superficially marginal} in the infrared limit, in contrast to the  interaction $\slice{H'_{P}}{n}{m}-H_{P,0}$.  At a fixed ultraviolet cut-off and at a small coupling constant $g$, it has been proven in \cite{pizzo_one-particle_2003} that the sequence of ground states $(\slice{\phi_P}{n}{m})_{m\in\bb N}$, i.e.
\begin{equation}
\slice{H_{P}^{W'}}{n}{m}\slice{\phi_P}{n}{m}=\slice{E'_P}{n}{m}\slice{\phi_P}{n}{m},
\end{equation}
converges in the limit $m\to \infty$ while the spectral gap closes. Consequently, infrared asymptotic freedom holds. This result requires a sophisticated proof by induction. In the present paper we prove the same result while simultaneously removing the ultraviolet cut-off. Before sketching the main technical difficulties in dealing with the construction of the states $\slice{\phi_P}{\infty}{\infty}$ let us briefly explain their physical relevance.

With the states $\slice{\phi_P}{n}{m}$ and the Bogolyubov transformation $W_m(\nabla\slice{E'_P}{n}{m})$ at hand it is possible to control the properties of the physical mass shell given by the states $\slice{\Psi'_P}{n}{m}$ in the infrared limit, i.e. $m\to \infty$, namely the dependence on the total momentum $P$. This spectral  information represents the key ingredient to construct the scattering states for the so-called \emph{infraparticles}   (see \cite{pizzo_one-particle_2003} and  \cite{chen_infraparticle_2009}). The QED analogue of the transformation of  the field variables in  \eqn{eqn:W trafo} is related to the Li\'enard-Wiechert fields carried by the charged particle and to the infrared radiation emitted in Compton scattering; see \cite{chen_infraparticle_2009} for precise mathematical statements.

More technically, while simultaneously removing the infrared and the ultraviolet cut-off in $\slice{\phi_P}{n}{m}$ a major difficulty arises in the induction mentioned above. In fact, the proof of the limit of $\slice{\phi_P}{n}{m}$ as $m\to\infty$ relies on a suitable rearrangement  of the terms in the Hamiltonian $ \slice{H_P^{W'}}{n}{m}$ given by
\begin{align}
  \slice{H^{W'}_P}{n}{m}&=\frac{1}{2}\slice{\Gamma_P}{n}{m}^2+ H^f-\nabla \slice{E'_P}{n}{m} \cdot P^f + C_{P,m}^{(n)} + \slice{R_P}{n}{m}\,,
\end{align}
see \eqn{eqn:Hw} in Section \ref{sec:W ground states}, where the vector operator $\slice{\Gamma_P}{n}{m}$ has the crucial property
\begin{equation}
\langle \slice{\phi_P}{n}{m},\slice{\Gamma_P}{n}{m}\slice{\phi_P}{n}{m}\,\rangle=0 .
\end{equation}
However, the operator $\slice{\Gamma_P}{n}{m}$ is ill-defined in the limit $n\to \infty$. This suggests the following strategy for the simultaneous removal of the cut-offs, for sufficiently small $g$ but uniform in $n$ and $m$:
\begin{enumerate}[(i)]
\item First show that  $(\slice{\phi_P}{n}{m})_{m\in \mathbb{N}}$ is a Cauchy sequence uniformly in $n$;
\item then provide bounds of the form
\begin{equation}
\|\slice{\phi_P}{n}{m}-\slice{\phi_P}{n-1}{m}\|\leq f_1(n,m),
\end{equation}
and
\begin{equation}
|\nabla\slice{E'_P}{n}{m}-\nabla\slice{E'_P}{n-1}{m}| \leq f_2(n,m),
\end{equation}
where $ f_1(n,m)$ and $f_2(n,m)$ are such that for the scaling $n(m):=\alpha m$ with $\alpha$ sufficiently large both $(\slice{\phi_P}{n(m)}{m})_{n\in\mathbb{N}}$ and $(\nabla\slice{E'_P}{n(m)}{m})_{n\in \mathbb{N}}$ are Cauchy sequences.
\end{enumerate}
This program will be carried out in Sections \ref{sec:W ground states} and \ref{sec:Winfinity}. It will yield the second main result:
 \begin{theorem}\label{thm:ir}
 Let $|P|\leq P_{\mathrm{max}}$. For $|g|$ sufficiently small the following holds true:
\begin{enumerate}[(i)]
 \item There exists an $\alpha_{\mathrm{min}}> 0$ such that for any integer $\alpha'> \alpha_{\mathrm{min}}$ and $n(m) = \alpha' m$, the limit
\begin{align*}
 \slice{\phi_{P}}{\infty}{\infty}:=\lim_{m\to\infty}\slice{\phi_P}{n(m)}{m}
\end{align*}
 exists in $\cal F$ and is non-zero.
 \item $E'_{P,\infty}:=\lim_{m\to\infty} \slice{E'_P}{\infty}{m}$ exists and is the ground state energy corresponding to the eigenvector $\slice{\phi_{P}}{\infty}{\infty}$
 of the self-adjoint operator
\begin{align*}
 \slice{H^{W'}_{P}}{\infty}{\infty}:=\lim_{m\to \infty}\slice{H^{W'}_{P}}{n(m)}{m},
\end{align*}
 where the limit is understood in the norm resolvent sense.
\end{enumerate}
\end{theorem}
At this point, we emphasize that at least within the scope of the presented multiscale technique, the given scaling to remove both UV and IR cut-offs simultaneously is natural. The method indeed relies on the control of the spectral gap, and as the gap closes in the IR limit, the UV limit must be taken at a comparatively fast enough rate.\\

For the notation throughout this paper, the reader is advised to consult the list below.
\paragraph{Notation.}
\begin{enumerate}
 \item  By convention $0\notin\bb N$.
 \item  The symbol $C$ denotes any universal constant. Any appearing $C$ is independent of the indices $m$ and $n$ and of all parameters in the paper, i.e. $g, \beta,\gamma$ and $\zeta$, at least in prescribed neighborhoods.
 \item The bars $|\cdot|$, $\|\cdot\|$ denote the euclidean and the Fock space norm, respectively. The brackets $\langle \cdot,\cdot \rangle$ denote the scalar product of vectors in $\cal F$. Given a subspace $\cal K\subseteq \cal F$ and an operator $A$ on $\cal F$ we use the notation
 \[
   \|A\|_{\cal K}=\|A\restrict{\cal K}\|.
  \]
 \item For a vector operator $A=(A^{(1)},A^{(2)},A^{(3)})$ with components $A^{(i)}:D(A^{(i)})\to\cal F$, $1\leq i\leq 3$,  we use the notation
 \[
   \|A\psi\|^2=\sum_{i=1}^3\|A^{(i)}\psi\|^2.
 \]
\end{enumerate}

\section{Tools} 

We recall some standard operator inequalities which are frequently used. For every square integrable function $f$ the estimates
\begin{align}
\begin{split}\label{eqn:standard ineq}
  \left\|\int_{\cal B_{\Lambda}\setminus\cal B_{\tau}} dk\; f(k) b(k)\psi\right\|&\leq \left(\int_{\cal B_{\Lambda}\setminus\cal B_{\tau}} dk\left|\frac{f(k)}{\sqrt{|k|}}\right|^2\right)^{1/2}\|(\slice{H^f}{\Lambda}{\tau})^{1/2}\psi\|,\\
    \left\|\int_{\cal B_{\Lambda}\setminus\cal B_{\tau}} dk\; f(k) b^*(k)\psi\right\|&\leq \left(\int_{\cal B_{\Lambda}\setminus\cal B_{\tau}} dk\left|\frac{f(k)}{\sqrt{|k|}}\right|^2\right)^{1/2}\|(\slice{H^f}{\Lambda}{\tau})^{1/2}\psi\|\\
    &\quad+\left(\int_{\cal B_{\Lambda}\setminus\cal B_{\tau}} dk\left|f(k)\right|^2\right)^{1/2}\|\psi\|
\end{split}
\end{align}
hold true for all $0\leq \tau<\Lambda\leq \infty$ and $\psi$ in the domain of $H_{P,0}^{1/2}$ whenever the integrals on the right-hand side of (\ref{eqn:standard ineq})  are well defined. 

 The following two results are crucial ingredients in the proofs presented in the next sections. The first one, \thm{thm:ground state energies}, is a classical result by L. Gross that turns out to be the main non-perturbative ingredient for the gap estimates that we obtain by iterative analytic perturbation theory; see Sections \ref{sec:main proof} and \ref{sec:ground states with ir cutoff}.

\begin{theorem}\label{thm:ground state energies}
 For $0\leq \tau<\Lambda<\infty$ and all $P\in\bb R^3$ the ground state energies $\slice{E_P}{\Lambda}{\tau}:=\inf \spec{\slice{H_P}{\Lambda}{\tau}}$ fulfill $\slice{E_0}{\Lambda}{\tau}\leq \slice{E_P}{\Lambda}{\tau}$.
\end{theorem}

\begin{proof}
See   \cite[Theorem 8]{gross_existence_1972}.
\end{proof}

The second one,  
\lem{lem:a priori},  plays  a role in Sections \ref{sec:ground states with ir cutoff}, \ref{sec:W ground states}, \ref{sec:Winfinity} where we consider the interaction both in the ultraviolet and in the infrared regime. It is a crucial ingredient  to prove statements (i), (ii) in \cor{cor:ir ground state energies}  .  We stress that the multiscale technique which we apply in \sct{sec:main proof} to remove the ultraviolet cut-off at $m=0$ does not refer to \cor{cor:ir ground state energies}  (i),(ii), and only relies on \thm{thm:ground state energies} and on  a weaker estimate given in \eqn{eqn:remaining resolvent} that follows from \eqn{eqn:standard ineq}.

\begin{lemma}\label{lem:a priori}\label{lem:ir a priori}
There exist finite constants $\namel{a}{c_a},\namel{b}{c_b}>0$ such that  
\begin{align}\label{eqn:a priori}
 \braket{\psi,H_{P,0}\psi}\leq \frac{1}{1-|g|\namer{a}}\bigg[\braket{\psi,\slice{H'_P}{n}{m}\psi}+|g|\namer{b}\braket{\psi,\psi}\bigg]
\end{align}
 for $|g|\leq 1,\frac{1}{\namer{a}}$ and $\psi\in D(H_{P,0}^{1/2})$ with $m,n\in\bb N$.
\end{lemma}

\begin{proof}
See Appendix A.
\end{proof}


\section{Ground States of the Gross Transformed Hamiltonians $\slice{H'_{P}}{\infty}{0}$}\label{sec:main proof}

This section provides the proof of \thm{thm:main uv} in \sct{sec:main results}. We start by introducing a sequence of gap bounds.
\begin{definition}\label{def:gap bounds}
We define the sequence of gap bounds
\begin{align}\label{eqn:def gap bounds}
 \xi_n:=\frac{1}{8}
 \kappa\left(1-\sum_{j=1}^n\Delta\xi_j\right), && \quad \Delta\xi_n:=\frac{(\beta-1)^2}{2\beta}\frac{n}{\beta^{n}}
\end{align}
for $n\in\bb N$ with the scaling parameter $\beta>1$. Furthermore, we impose the constraint
\begin{align}\label{eqn:uv constraint}
 |g|\leq (\beta -1), \qquad 1<\beta<2.
\end{align}
\end{definition}

The definition of the sequence of gap bounds $(\xi_n)_{n\in\bb N}$ in (\ref{eqn:def gap bounds}) will be motivated in  \lem{lem:neumann}. Note that $\sum_{j=1}^\infty\Delta\xi_j=\frac{1}{2}$ implies
\begin{align}\label{gap bounds-2}
\frac{1}{16}
\kappa\leq\xi_n\leq\frac{1}{8}
\kappa.
\end{align}

 \begin{remark}
In this section  the constraints  $|P|<P_{\mathrm{max}}$ and $1<\kappa<2$ are implicitly assumed. 
 \end{remark}
\begin{lemma}\label{lem:gap}
For an integer $n>1$ assume:
\begin{enumerate}[(i)]
 \item $\slice{E'_P}{n-1}{0}$ is the non-degenerate eigenvalue of $\slice{H'_{P}}{n-1}{0}\restrict{\slice{\cal F}{n-1}{0}}$ with eigenvector $\slice{\Psi'_{P}}{n-1}{0}$.
\item $\gap{\slice{H'_{P}}{n-1}{0}\restrict{\slice{\cal F}{n-1}{0}}}\geq\xi_{n-1}$.
\item $\slice{E'_P}{n-1}{0}$ is differentiable in $P$ and  $| \nabla \slice{E'_P}{n-1}{0}|\leq \namel{c:cp}{C_{\nabla E}}\equiv\frac{3}{4}$.
\end{enumerate}
This implies that $\slice{E'_P}{n-1}{0}$ is also the non-degenerate ground state energy of $\slice{H'_{P}}{n-1}{0}\restrict{\slice{\cal F}{n}{0}}$ with eigenvector $\slice{\Psi'_{P}}{n-1}{0}\otimes \Omega$. Furthermore,
\begin{align}\label{estimate-gap}
\begin{split}
\gap{\slice{H'_{P}}{n-1}{0}\restrict{\slice{\cal F}{n}{0}}}
&\geq \inf_{\slice{\cal F}{n}{0}\ni\psi\perp\slice{\Psi'_{P}}{n-1}{0}\otimes\Omega}\braket{\slice{H_{P}'}{n-1}{0}-\theta\slice{H^f}{n}{n-1}-\slice{E'_P}{n-1}{0}}_\psi\geq \xi_{n-1}
\end{split}
\end{align}
 where $0<\theta<\frac{1}{8}$ and the infimum is taken over $\psi\in D(H_{P,0})$.
\end{lemma}

\begin{proof}
Using (i), a direct computation yields
\begin{align*}
\slice{H'_{P}}{n-1}{0}(\slice{\Psi'_{P}}{n-1}{0}\otimes\Omega)=\slice{E'_P}{n-1}{0}(\slice{\Psi'_{P}}{n-1}{0}\otimes\Omega)
\end{align*}
 as the interaction is cut off at $\sigma_{n-1}$. Hence, $\slice{E'_P}{n-1}{0}$ is an eigenvalue of $\slice{H'_{P}}{n-1}{0}\restrict{\slice{\cal F}{n}{0}}$ with eigenvector $\slice{\Psi'_{P}}{n-1}{0}\otimes\Omega$. Let us consider
\begin{align}\label{eqn:gap n-1 to n}
 \inf_{\slice{\cal F}{n}{0}\ni\psi\perp\slice{\Psi'_{P}}{n-1}{0}\otimes\Omega}\braket{\slice{H'_{P}}{n-1}{0}-\slice{E'_P}{n-1}{0}}_\psi.
\end{align}
As the Gross transformation is unitary and does not affect $\slice{\cal F}{n}{n-1}$, and since $\slice{H^f}{n}{n-1}$ is positive,  we have
\begin{align}\label{eqn:inf}
 \eqn{eqn:gap n-1 to n}\geq \inf_{\slice{\cal F}{n}{0}\ni\psi\perp\slice{\Psi_P}{n-1}{0}\otimes\Omega}\braket{\slice{H_{P}}{n-1}{0}-\theta\slice{H^f}{n}{n-1}-\slice{E_{P}}{n-1}{0}}_\psi.
\end{align}
We subtract the term $\theta\slice{H^f}{n}{n-1}$ for a technical reason which will become clear in \lem{lem:neumann}. 

Now, the right-hand side of \eqn{eqn:inf} is bounded from below by
\begin{align*}
 \min\left\{\gap{\slice{H'_{P}}{n-1}{0}\restrict{\slice{\cal F}{n-1}{0}}}, \inf_{\psi=\varphi\otimes\eta}\braket{\slice{H_{P}}{n-1}{0}-\theta\slice{H^f}{n}{n-1}-\slice{E_{P}}{n-1}{0}}_\psi\right\},
\end{align*}
where $\varphi\in\slice{\cal F}{n-1}{0}$, $\eta\in\slice{\cal F}{n}{n-1}$, $\varphi\otimes\eta$ belongs to $D(H_{P,0})$ and $\eta$ is a vector with a definite, strictly positive number of bosons. For $m\geq 1$ bosons in the vector $\eta$ we estimate
\begin{align}
 &\inf_{\psi=\varphi\otimes\eta}\braket{\slice{H_{P}}{n-1}{0}-\theta\slice{H^f}{n}{n-1}-\slice{E_{P}}{n-1}{0}}_\psi \nonumber\\
&\geq\inf_{\varphi, k_j\in[\sigma_{n-1},\sigma_n)}\braket{\frac{1}{2}\left(P-P^f-\sum_{j=1}^m k_j\right)^2+H^f+g \slice{\Phi}{n-1}{0}+(1-\theta)\sum_{j=1}^m|k_j|-\slice{E_{P}}{n-1}{0}}_\varphi \nonumber\\
  &\geq\inf_{k_j\in[\sigma_{n-1},\sigma_n)}\left[(1-\theta)\sum_{j=1}^m|k_j|+\slice{E_{P-\sum_{j=1}^m k_j}}{n-1}{0}-\slice{E_{P}}{n-1}{0}\right] \label{eq: gradnograd1}\\
  &\geq(1-\theta-\namer{c:cp})\sigma_{n-1}\geq\frac{1}{8}\kappa \label{eq: gradnograd2}
\end{align}
 where the steps \eqn{eq: gradnograd1} and \eqn{eq: gradnograd2} follow from:
\begin{enumerate}
 \item $\sigma_{n-1}\geq\kappa$, $0<\theta<\frac{1}{8}$ and $\namer{c:cp}=\frac{3}{4}$.
 \item The estimate 
\[
 \slice{E_{P-\sum_{j=1}^m k_j}}{n-1}{0}-\slice{E_{P}}{n-1}{0}=\slice{E_{P-\sum_{j=1}^m k_j}}{n-1}{0}-\slice{E_{0}}{n-1}{0}+\slice{E_{0}}{n-1}{0}-\slice{E_{P}}{n-1}{0}\geq \slice{E_{0}}{n-1}{0}-\slice{E_{P}}{n-1}{0}
\]
which holds by \thm{thm:ground state energies}.
 \item The estimate
\[
 \slice{E'_0}{n-1}{0}-\slice{E_{P}}{n-1}{0}\geq- \sup_{|Q|\leq P_{\mathrm{max}}}|\nabla \slice{E'_Q}{n}{0}|\geq -\namer{c:cp}
\]
since $\slice{E'_P}{n-1}{0}$ is differentiable in $P$ and  $|P|<1$.
\end{enumerate}

First, this implies that \eqn{eqn:gap n-1 to n} is bounded from below by $\min\left\{\xi_{n-1},\frac{\kappa}{8}\right\}=\xi_{n-1}$; see  (\ref{gap bounds-2}). Second, it turns out that $\slice{\Psi'_{P}}{n-1}{0}$ is the non-degenerate ground state of $\slice{H'_{P}}{n-1}{0}\restrict{\slice{\cal F}{n}{0}}$ with
\begin{align*}
 \gap{\slice{H'_{P}}{n-1}{0}\restrict{\slice{\cal F}{n}{0}}}\geq \xi_{n-1}.
\end{align*}

\end{proof}


%
%

\begin{remark}
Under the assumptions of \lem{lem:gap} it follows that  for $j,n \in \mathbb{N}$
\begin{align*}
 \slice{E'_P}{n}{0}=\inf\spec{\slice{H'_{P}}{n}{0}\restrict{\slice{\cal F}{n}{0}}}=\inf\spec{\slice{H'_{P}}{n}{0}\restrict{\cal F_{n+j}}}\,.
\end{align*}
\end{remark}

\begin{lemma}\label{lem:neumann}
 Let  $n\geq 1$. For $n=1$, set $\slice{H'_{P}}{n-1}{0}:=H'_{P,0}$, $\slice{E'_P}{n-1}{0}:=P^2/2$,  and $\xi_{n-1}:=\kappa/2$. Assume that for some universal constant $\namel{c:ce}{C_{E'}}$ the bound $|\slice{E'_P}{n-1}{0}|<\namer{c:ce}$ holds true. Then there exist $\beta_{\mathrm{max}}>1$ and $g_{\mathrm{max}}>0$ such that, for all $1<\beta\leq\beta_{\mathrm{max}}$ and $|g|\leq g_{\mathrm{max}}$, the assumptions (i), (ii) in   \lem{lem:gap}
imply that
\begin{align}\label{eqn:domain}
 \frac{1}{\slice{H'_{P}}{n}{0}-z}\restrict{\slice{\cal F}{n}{0}}, \quad\frac{\xi_{n}}{2}\leq|\slice{E'_P}{n-1}{0}-z|\leq \xi_{n},
\end{align} 
is well-defined.
\end{lemma}

\begin{proof}
 Let $z$ be in the domain given in \eqn{eqn:domain}. In order to control the expansion of the resolvent $(\slice{H'_{P}}{n}{0}-z)^{-1}$, i.e. 
\begin{align*}
\frac{1}{\slice{H'_{P}}{n-1}{0}-z}\sum_{j=0}^\infty \left[-\slice{\Delta H_P'}{n}{n-1} \frac{1}{\slice{H'_{P}}{n-1}{0}-z}\right]^j\restrict{\cal F|^{0}_{n}},
\end{align*}
it is sufficient to prove that
\begin{align}\label{eqn:neumann term}
 \left\|\left(\frac{1}{\slice{H'_{P}}{n-1}{0}-z}\right)^{1/2}\slice{\Delta H_P'}{n}{n-1}\left(\frac{1}{\slice{H'_{P}}{n-1}{0}-z}\right)^{1/2}\right\|_{\slice{\cal F}{n}{0}}<1.
\end{align}
As we shall show now, this can be achieved by a convenient choice of $\beta$ and $g$ (uniformly in $n$) using the gap bounds $(\xi_n)_{n\in\bb N}$ from \dfn{def:gap bounds}. We can express the interaction term by
\begin{align}\label{eqn:delta H}
\begin{split}
  \slice{\Delta H_P'}{n}{n-1}=&\frac{1}{2}\left((\slice{B}{n}{n-1})^2+(\slice{B^*}{n}{n-1})^2\right)+\slice{B}{n-1}{0}\cdot\slice{B}{n}{n-1}+\slice{B^*}{n}{n-1}\cdot\slice{B^*}{n-1}{0}\\
 &-(P-P^f)\cdot\slice{B}{n}{n-1}-\slice{B^*}{n}{n-1}\cdot(P-P^f)\\
 &+\slice{B^*}{n}{n-1}\cdot\slice{B}{n}{n-1}+\slice{B^*}{n-1}{0}\cdot\slice{B}{n}{n-1}+\slice{B^*}{n}{n-1}\cdot\slice{B}{n-1}{0}.
\end{split}
\end{align}
Hence, the left-hand side of \eqn{eqn:neumann term} is bounded by
\begin{align}\label{eqn:neumann bounds}
 &\left\|\slice{B}{n}{n-1}\left(\frac{1}{\slice{H'_{P}}{n-1}{0}-z}\right)^{1/2}\right\|_{\slice{\cal F}{n}{0}}\times\\
 &\quad\times\Bigg[\left\|\slice{B^*}{n}{n-1}\left(\frac{1}{\slice{H'_{P}}{n-1}{0}-z}\right)^{1/2}\right\|_{\slice{\cal F}{n}{0}}+\left\|\slice{B}{n}{n-1}\left(\frac{1}{\slice{H'_{P}}{n-1}{0}-z}\right)^{1/2}\right\|_{\slice{\cal F}{n}{0}}\\
 &\quad\quad+2\left\|\slice{B^*}{n-1}{0}\left(\frac{1}{\slice{H'_{P}}{n-1}{0}-z}\right)^{1/2}\right\|_{\slice{\cal F}{n}{0}}+2\left\|\slice{B}{n-1}{0}\left(\frac{1}{\slice{H'_{P}}{n-1}{0}-z}\right)^{1/2}\right\|_{\slice{\cal F}{n}{0}}+\label{eqn:neumann bounds1}\\
 &\quad\quad+2\left\|(P-P^f)\left(\frac{1}{\slice{H'_{P}}{n-1}{0}-z}\right)^{1/2}\right\|_{\slice{\cal F}{n}{0}}\Bigg]\label{eqn:neumann bounds2}.
\end{align}
Notice that  the standard inequalities in \eqn{eqn:standard ineq} yield
\begin{align}
\begin{split}\label{eqn:uv standard ineq}
  \|\slice{B}{n}{m}\psi\|&\leq |g|\;C\left(\frac{1}{\sigma_m}-\frac{1}{\sigma_n}\right)^{1/2}\|(\slice{H^f}{n}{m})^{1/2}\psi\|, \\
 \|\slice{B^*}{n}{m}\psi\|&\leq |g|\;C\left(\left(\frac{1}{\sigma_m}-\frac{1}{\sigma_n}\right)^{1/2}\|(\slice{H^f}{n}{m})^{1/2}\psi\|+(\ln\sigma_n-\ln\sigma_m)^{1/2}\|\psi\|\right)
\end{split}
\end{align}
for all $\psi$ in the domain of $H_{P,0}^{1/2}$. 
Then expression \eqn{eqn:neumann bounds}-\eqn{eqn:neumann bounds2} can be controlled as follows:
\begin{enumerate}[1.]
 \item We estimate
\begin{align}\label{eqn:b n n-1 term}
 \left\|\slice{B}{n}{n-1}\left(\frac{1}{\slice{H'_{P}}{n-1}{0}-z}\right)^{1/2}\right\|_{\slice{\cal F}{n}{0}}\leq |g|C\left(\frac{\beta-1}{\sigma_{n}}\right)^{1/2}\left\|(\slice{H^f}{n}{n-1})^{1/2}\left(\frac{1}{\slice{H'_{P}}{n-1}{0}-z}\right)^{1/2}\right\|_{\slice{\cal F}{n}{0}}.
\end{align}
Furthermore, since $\slice{H^f}{n}{n-1}$ and $\slice{H'_{P}}{n-1}{0}$ commute, we have that
 \begin{align}
 &\left\|(\slice{H^f}{n}{n-1})^{1/2}\left(\frac{1}{\slice{H'_{P}}{n-1}{0}-z}\right)^{1/2}\right\|_{\slice{\cal F}{n}{0}} \label{eqn-teo-spectr}\\
 &\leq \theta^{-1/2}\;\left\|\left(\frac{\theta\slice{H^f}{n}{n-1}}{\slice{H'_{P}}{n-1}{0}-\theta\slice{H^f}{n}{n-1}-\slice{E'_P}{n-1}{0}+\theta\slice{H^f}{n}{n-1}+\slice{E'_P}{n-1}{0}-z)}\right)^{1/2}\right\|_{\slice{\cal F}{n}{0}}\nonumber \\
 &\leq \theta^{-1/2}\;\left\|\left(\frac{\theta\slice{H^f}{n}{n-1}}{\xi_{n-1}-\xi_{n}+\theta\slice{H^f}{n}{n-1})}\right)^{1/2}\right\|_{\slice{\cal F}{n}{0}} \leq  \theta^{-1/2}\nonumber
\end{align}
for, e.g. $\theta=\frac{1}{16}$. This is true  because of \lem{lem:gap}, the constraints on $z$ given in \eqn{eqn:domain}, and the bound $\Delta\xi_n=\xi_{n-1}-\xi_{n}>0$ (see \dfn{def:gap bounds}).
 \item Next we consider the bounds
 \begin{align}\label{eqn:B res}
 \left\|\slice{B}{n-1}{0}\left(\frac{1}{\slice{H'_{P}}{n-1}{0}-z}\right)^{1/2}\right\|_{\slice{\cal F}{n}{0}}&\leq |g|C \left\|(\slice{H^f}{n-1}{0})^{1/2}\left(\frac{1}{\slice{H'_{P}}{n-1}{0}-z}\right)^{1/2}\right\|_{\slice{\cal F}{n}{0}},
\end{align}
and
\begin{align}\label{eqn:Bstar res}
\begin{split}
 \left\|\slice{B^*}{n-1}{0}\left(\frac{1}{\slice{H'_{P}}{n-1}{0}-z}\right)^{1/2}\right\|_{\slice{\cal F}{n}{0}}&\leq |g|C\Bigg(\left\|(\slice{H^f}{n-1}{0})^{1/2}\left(\frac{1}{\slice{H'_{P}}{n-1}{0}-z}\right)^{1/2}\right\|_{\slice{\cal F}{n}{0}}
\\
&\hskip2cm+\left(\ln\beta^{n-1}\right)^{1/2}\left\|\left(\frac{1}{\slice{H'_{P}}{n-1}{0}-z}\right)^{1/2}\right\|_{\slice{\cal F}{n}{0}}\Bigg). 
\end{split}
\end{align}
 
 Terms including $\slice{H^f}{n-1}{0}$ or $(P-P^f)$ can be estimated as follows:
 \begin{align}\label{first}
 \left\|(\slice{H^f}{n-1}{0})^{1/2}\left(\frac{1}{\slice{H'_{P}}{n-1}{0}-z}\right)^{1/2}\right\|_{\slice{\cal F}{n}{0}}&\leq\left\|H_{P,0}^{1/2}\left(\frac{1}{\slice{H'_{P}}{n-1}{0}-z}\right)^{1/2}\right\|_{\slice{\cal F}{n}{0}},\\
 \left\|(P-P^f)\left(\frac{1}{\slice{H'_{P}}{n-1}{0}-z}\right)^{1/2}\right\|_{\slice{\cal F}{n}{0}}&\leq \sqrt 2 \left\|H_{P,0}^{1/2}\left(\frac{1}{\slice{H'_{P}}{n-1}{0}-z}\right)^{1/2}\right\|_{\slice{\cal F}{n}{0}}.\label{second}
 \end{align}
In order to estimate the right-hand side in (\ref{first}) and (\ref{second}), we observe that the standard inequalities \eqn{eqn:uv standard ineq} readily imply  that
 there exists a n-indepedent finite constant $c_{uv}$ such that,  for $|g|\leq1$ and   $|g|<\frac{1}{c_{uv}}$, $\psi\in D(H_{P,0}^{1/2})$ and $n\in\bb N$, it holds 
\begin{align}\label{eqn:a priori-2}
 \braket{\psi,H_{P,0}\psi}\leq \frac{1}{1-|g|c_{uv}}\bigg[\braket{\psi,\slice{H'_P}{n}{0}\psi}+g^2c_{uv}^2\ln \sigma_n\braket{\psi,\psi}\bigg].
\end{align}
Consequently, for $|g|$ sufficiently small, we can estimate
 \begin{align}
 \left\|H_{P,0}^{1/2}\left(\frac{1}{\slice{H'_{P}}{n-1}{0}-z}\right)^{1/2}\right\|^2_{\slice{\cal F}{n}{0}}
 &
 \leq C\sup_{\|\psi\|=1}\braket{\psi,\left[1+\left(|z|+|g|\ln \sigma_n\right)\left|\frac{1}{\slice{H'_{P}}{n-1}{0}-z}\right|\right]\psi}\label{eqn:H0 est 1}
\end{align}
where $\psi\in\slice{\cal F}{n}{0}$.
Moreover, the right-hand side of
\begin{align*}
 |z|\leq|\slice{E'_P}{n-1}{0}-z|+|\slice{E'_P}{n-1}{0}|
\end{align*}
is uniformly bounded  because, first, $|\slice{E'_P}{n-1}{0}-z|\leq\xi_{n-1}\leq \frac{1}{2}\kappa$, and,  second, $|\slice{E'_P}{n-1}{0}|\leq \namer{c:ce}$ by assumption.
Hence, we get
\begin{align}
  \left\|H_{P,0}^{1/2}\left(\frac{1}{\slice{H'_{P}}{n-1}{0}-z}\right)^{1/2}\right\|^2_{\slice{\cal F}{n}{0}}\leq C\left(1+(1+|g|\ln \sigma_n)\left\|\left(\frac{1}{\slice{H'_{P}}{n-1}{0}-z}\right)^{1/2}\right\|^2_{\slice{\cal F}{n}{0}}\right)\label{eqn:remaining resolvent}.
\end{align}

Finally, the remaining norm in \eqn{eqn:remaining resolvent} can be controlled by
\begin{align}
 \left\|\left(\frac{1}{\slice{H'_{P}}{n-1}{0}-z}\right)^{1/2}\right\|^2_{\slice{\cal F}{n}{0}}&\leq \max\left\{\frac{1}{|\slice{E'_P}{n-1}{0}-z|},\frac{1}{\gap{\slice{H'_{P}}{n-1}{0}\restrict{\slice{\cal F}{n}{0}}}-|\slice{E'_P}{n-1}{0}-z|}\right\}\leq \frac{C}{\Delta\xi_n}\label{eqn:uv res est}
\end{align}
which is due to \lem{lem:gap} and the domain of $z$ given in \eqn{eqn:domain}.
\end{enumerate}

We recall that by \dfn{def:gap bounds} the sequence $(\Delta\xi_n)_{n\in\bb N}$ tends to zero, which is a necessary ingredient in the induction scheme in the proof of \thm{thm:main uv}. Hence, the terms proportional to $(\Delta\xi_n)^{-1/2}$ must be treated cautiously. 
It turns out that the sum of the terms in \eqn{eqn:neumann bounds}-\eqn{eqn:neumann bounds2} is bounded by
\begin{align}
 \cal O \left(|g|\left(\frac{(\beta-1)}{\sigma_n\Delta\xi_n}\right)^{1/2}\right)+\cal O \left(|g|\left(\frac{(\beta-1)\ln\beta^{n-1}}{\sigma_n\Delta\xi_n}\right)^{1/2}\right)\leq |g|^{1/2}  C\left(\frac{(\beta-1)^2n}{\beta^n\Delta\xi_n}\right)^{1/2}\label{eqn:uv int est}
\end{align}
 for $|g|\leq (\beta -1)$; see \eqn{eqn:uv constraint}. This  dictates the choice $\Delta\xi_n:=\frac{(\beta-1)^2}{2\beta}\frac{n}{\beta^{n}}$ made in \dfn{def:gap bounds}. Hence, for all $n\in\bb N$ we get
\begin{align}\label{main-inequality-lemma4.4}
 \left\|\left(\frac{1}{\slice{H'_{P}}{n-1}{0}-z}\right)^{1/2}\slice{\Delta H_P'}{n}{n-1}\left(\frac{1}{\slice{H'_{P}}{n-1}{0}-z}\right)^{1/2}\right\|_{\slice{\cal F}{n}{0}}
 &\leq |g|^{1/2}C \left(\frac{(\beta-1)^2n}{\beta^n\;\Delta\xi_n}\right)^{1/2}\leq |g|^{1/2} C.
\end{align}
Therefore, \eqn{eqn:neumann term} holds for $|g|$  sufficiently small which proves the claim.
\end{proof}

\begin{definition}\label{definition-4.3}
 For $n\in\bb N$ we define the contour
\begin{align*}
 \Gamma_n := \left\{z\in\bb C\;\bigg|\;|\slice{E'_P}{n-1}{0}-z|=\frac{1}{2}\xi_n\right\}.
\end{align*}
\end{definition}

The bound in (\ref{eqn:uv int est}) was delicate because the outer boundary of the domain of $z$ might be close to the spectrum. However, when considering $z$ being further away from the spectrum we get a much better estimate:
\begin{corollary}\label{lem:resolvent diff}
Let $g$, $\beta$ fulfill the conditions of \lem{lem:neumann} and $z\in\Gamma_n$ or $z=\slice{E'_P}{n}{0}+i\lambda$ with $\lambda\in \mathbb{R}, |\lambda|=1$  for $n\in\bb N$. The following estimates  hold true
 \begin{align}\label{eqn:uv res est}
 \left\|\left(\frac{1}{\slice{H'_{P}}{n-1}{0}-z}\right)^{1/2}\slice{\Delta H_P'}{n}{n-1}\left(\frac{1}{\slice{H'_{P}}{n-1}{0}-z}\right)^{1/2}\right\|_{\slice{\cal F}{n}{0}}&\leq C|g| \left(\frac{(\beta-1)n}{\beta^n}\right)^{1/2},\\
\label{eqn:res diff}
\left\|\frac{1}{\slice{H'_{P}}{n}{0}-z}-\frac{1}{\slice{H'_{P}}{n-1}{0}-z}\right\|_{\slice{\cal F}{n}{0}}&\leq C|g| \left(\frac{(\beta-1)n}{\beta^n}\right)^{1/2}.  
\end{align}
\end{corollary}

\begin{proof}
It is enough to notice that in the estimate of the left-hand side of (\ref{eqn:uv res est}) one can just replace $\Delta\xi_n$ in (\ref{eqn:uv int est}) by a constant and use that $1< \beta< 2$, see (\ref{eqn:uv constraint}). For $|g|$ small enough, the inequality in \eqn{eqn:res diff} follows from (\ref{eqn:uv res est}).
\end{proof}

With these lemmas at hand we prove the induction step for the removal of the ultraviolet cut-off.

\begin{theorem}\label{thm:uv induction}
 Let $g$, $\beta$ fulfill the assumptions of  \lem{lem:neumann}. Then for $|g|$ sufficiently small the following holds true for all $n\in\bb N$:
\begin{enumerate}[(i)]
 \item $\slice{E'_P}{n}{0}:=\inf\spec{\slice{H'_{P}}{n}{0}\restrict{\slice{\cal F}{n}{0}}}$ is a non-degenerate eigenvalue of $\slice{H'_{P}}{n}{0}\restrict{\slice{\cal F}{n}{0}}$.
 \item $\gap{\slice{H'_{P}}{n}{0}\restrict{\slice{\cal F}{n}{0}}}\geq\xi_{n}$.
 \item The vectors
\begin{align}
 \slice{\Psi'_P}{0}{0}&:=\Omega,\nonumber\\
 \slice{\Psi'_P}{j}{0}&:=\slice{\cal Q'_{P}}{j}{0}\slice{\Psi'_P}{j-1}{0}, \qquad \slice{\cal Q'_{P}}{j}{0}:=-\frac{1}{2\pi i}\oint_{\Gamma_{j}}\frac{dz}{H'_{P,j}-z}, \qquad j\geq 1,\label{eqn:projection}
\end{align}
  are well-defined and $\slice{\Psi'_P}{n}{0}$ is the unique ground state of $\slice{H_{P}'}{n}{0}\restrict{\slice{\cal F}{n}{0}}$.
  \item The following holds:
\begin{align}\label{diff-uv}
  \left\|\slice{\Psi'_P}{n}{0}-\slice{\Psi'_{P}}{n-1}{0}\right\|&\leq C|g|\left(\frac{(\beta-1)n}{\beta^n}\right)^{1/2},\\
 \|\slice{\Psi'_P}{n}{0}\|&\geq C_{\Psi'} \label{bound from below}
 \end{align} 
  where $0<C_{\Psi'}<1$.
  \item
$\slice{E'_P}{n}{0}$ is analytic in $P$ for all $n\in\bb N$ and the following bounds hold true
  \begin{align}
|\slice{E'_P}{n}{0}-\slice{E'_P}{n-1}{0}|&\leq C|g|^{2}\frac{(\beta-1)n}{\beta^n}\label{eqn:EnShift},\qquad |\slice{E'_P}{n}{0}| < \namer{c:ce}\left(>\frac{P^2}{2}\right), \\
|\nabla \slice{E'_P}{n}{0}-\nabla \slice{E'_P}{n-1}{0}|&\leq C|g|^{2}\frac{(\beta-1)n}{\beta^n},\qquad  |\nabla\slice{E'_P}{n}{0}|\leq \namer{c:cp}\left(=\frac{3}{4}\right), \label{diff-grad-uv}
\end{align}
where $\slice{E_{P}'}{0}{0}\equiv \frac{P^{2}}{2}$ and $\nabla\slice{E'_P}{0}{0}\equiv P$.
\end{enumerate}
\end{theorem}

\begin{proof}
  We prove this by induction: Statements (i)-(v) for $(n-1)$ will be referred to as assumptions A(i)-A(v) while the same statements for $n$ are claims C(i)-C(v). For $n=1$ the claims can be verified by direct computation and by using \lem{lem:neumann}. Let $n> 1$ and suppose A(i)-A(v) hold. 

\begin{enumerate}[1.]
 \item Because of A(i), A(ii), and A(v) \lem{lem:gap} states that
\begin{align*}
 \gap{\slice{H'_{P}}{n-1}{0}\restrict{\slice{\cal F}{n}{0}}}\geq \xi_{n-1}. 
\end{align*}
 \lem{lem:neumann} ensures that the resolvent $(\slice{H'_{P}}{n}{0}-z)^{-1}$ is well-defined for $\frac{1}{2}\xi_{n}\leq|\slice{E'_P}{n-1}{0}-z|\leq \xi_n$. 
 \item Hence, Kato's theorem yields claims C(i) and C(iii). As a consequence, the spectrum of $\slice{H'_{P}}{n}{0}\restrict{\slice{\cal F}{n}{0}}$ is contained in $\{\slice{E'_P}{n}{0}\}\cup(\slice{E'_P}{n-1}{0}+\xi_n,\infty)$ because $\slice{E'_P}{n}{0}\leq \slice{E'_P}{n-1}{0}$ by (iii) of \cor{cor:ground state energies} for $m=0$, which proves claim C(ii).   
 \item Next, we prove C(iv). By A(iii)  we have
\begin{align}
 \|\slice{\Psi'_P}{n}{0}-\slice{\Psi'_{P}}{n-1}{0}\| &\leq \|(\cal Q'_{n}-\cal Q'_{n-1})\slice{\Psi'_{P}}{n-1}{0}\|=\cal O\left(|g|\left(\frac{(\beta-1)n}{\beta^n}\right)^{1/2}\right)\label{eqn:order of psi diff}
\end{align}
where we have used \lem{lem:resolvent diff} and that $\|\slice{\Psi'_{P}}{n-1}{0}\|\leq 1$ holds by construction. Furthermore, starting from the identity
\begin{align}\label{eqn:norm psi}
  \|\slice{\Psi'_P}{n}{0}\|^2=\|\slice{\Psi'_{P}}{n-1}{0}\|^2+\|\slice{\Psi'_P}{n}{0}-\slice{\Psi'_{P}}{n-1}{0}\|^2+2\Re\braket{\slice{\Psi'_{P}}{n-1}{0},\slice{\Psi'_P}{n}{0}-\slice{\Psi'_{P}}{n-1}{0}}
\end{align}
we conclude that
\begin{align}
  \|\slice{\Psi'_P}{n}{0}\|^2-\|\slice{\Psi'_{P}}{n-1}{0}\|^2=\cal O\left(|g|^{2}\frac{(\beta-1)n}{\beta^n}\right).
\end{align}
Finally, since $\|\slice{\Psi'_P}{0}{0}\|=1$ by definition,
\begin{align*}
 \|\slice{\Psi'_P}{n}{0}\|^2\geq 1-\sum_{j=1}^n\left|\|\Psi'_{P}|_{0}^{j}\|^2-\|\Psi'_{P}|_{0}^{j-1}\|^2\right|\geq 1-C|g|^{2}\sum_{j=0}^{n}\frac{(\beta-1) j}{\beta^j}\geq 1-\cal O(|g|)\geq C_{\Psi'}>0
\end{align*}
for some positive constant $C_{\Psi'}$,  and $|g|$ sufficiently small and subject to the constraint $|g|\leq (\beta-1)$; see \eqn{eqn:uv constraint}.

\item In order to prove C(v), first by using (\ref{eqn:uv res est}) and  (\ref{bound from below}) we can estimate the energy shift  as follows
\begin{equation}
 |\slice{E'_P}{n}{0}-\slice{E'_P}{n-1}{0}|=\left|\frac{\braket{\slice{\Psi'_P}{n}{0},\slice{\Delta H_P'}{n}{n-1}\slice{\Psi'_{P}}{n-1}{0}}}{\braket{\slice{\Psi'_P}{n}{0},\slice{\Psi'_{P}}{n-1}{0}}} \right|\label{eqn:en diff0}=\cal O\left(|g|^{2}\frac{(\beta-1) n}{\beta^n}\right)\nonumber
\end{equation}
This readily implies
\begin{equation}\label{energy-bound-section-4}
|\slice{E'_P}{n}{0}|\leq \frac{P^2}{2}+C|g|^{2}\sum_{j=0}^{n}\frac{ (\beta-1) j}{\beta^j}\leq \namer{c:ce}
\end{equation}
for some constant $\namer{c:ce}$.

Since $(\slice{H'_P}{n}{0})_{|P|\leq P_{\mathrm{max}}}$ is an analytic family of type A and $\slice{E'_P}{n}{0}$ is an isolated eigenvalue, $\slice{E'_P}{n}{0}$ is an analytic function of $P$ and
\begin{align}\label{gradient-at-m=0}
\nabla \slice{E'_P}{n}{0}=P-\braket{[P^f+\slice{B}{n}{0}+\slice{B^*}{n}{0}]}_{\slice{\Psi'_{P}}{n}{0}}.
\end{align}

By using equations (\ref{eqn:b n n-1 term}), (\ref{eqn-teo-spectr}),  (\ref{eqn:B res}), (\ref{second}), (\ref{eqn:a priori-2}) for $z\in \Gamma_n$ (see Definition \ref{definition-4.3}), and (\ref{eqn:order of psi diff}), for $|g|$ sufficiently small
one can easily prove that
\begin{align*}
\nabla \slice{E'_P}{n}{0}-\nabla \slice{E'_P}{n-1}{0}
&= -\braket{[\slice{B}{n}{n-1}+\slice{B^*}{n}{n-1}]}_{\slice{\Psi'_{P}}{n}{0}}\\
& \qquad+\braket{[P-P^f+\slice{B}{n-1}{0}+\slice{B^*}{n-1}{0}]}_{\slice{\Psi'_{P}}{n}{0}}-\braket{[P-P^f+\slice{B}{n-1}{0}+\slice{B^*}{n-1}{0}]}_{\slice{\Psi'_{P}}{n-1}{0}}\\
&= \cal O\left(|g|^{2}\frac{ (\beta-1) n}{\beta^n}\right)
\end{align*}
\end{enumerate}
 and finally the bound $|\nabla \slice{E'_P}{n}{0}|\leq\frac{3}{4}=\namer{c:cp}$.
\end{proof}

We can now prove the first main result.
\begin{proof}[Proof of \thm{thm:main uv} in \sct{sec:main results}]
\mbox{}\\
\begin{enumerate}[(i)]
 \item Recall that $\slice{\Psi_P}{n}{0}:=e^{-\slice{T}{n}{0}}\slice{\Psi'_P}{n}{0}$. By unitarity of the Gross transformation
\begin{align*}
 \|\slice{\Psi_P}{n}{0}-\slice{\Psi_P}{n-1}{0}\|&=\|\slice{\Psi'_P}{n}{0}-e^{\slice{T}{n}{n-1}}\slice{\Psi'_{P}}{n-1}{0}\|\\
 &\leq\|(e^{\slice{T}{n}{n-1}}-1)\slice{\Psi'_P}{n-1}{0}\|+\|\slice{\Psi'_P}{n}{0}-\slice{\Psi'_{P}}{n-1}{0}\|
\end{align*}
holds. The convergence of  $(\slice{\Psi'_P}{n}{0})_{n\in\bb N}$ to a non-zero vector (see \thm{thm:uv induction}) and 
\begin{align*}
 \|(e^{\slice{T}{n}{n-1}}-1)\slice{\Psi'_P}{n-1}{0}\|&\leq \int_0^1d\lambda\;\| e^{\lambda\slice{T}{n}{n-1}}\slice{T}{n}{n-1}\slice{\Psi'_P}{n-1}{0}\|\\
 &\leq \|\slice{T}{n}{n-1}\slice{\Psi'_P}{n-1}{0}\|\xrightarrow[n\to\infty]{}0
\end{align*}
imply the claim.
 \item  Again the unitarity of the Gross transformation and \eqn{eqn:Vself renorm} implies 
\begin{align}
 \slice{E_P}{n}{0}-\slice{V_{\mathrm{self}}}{n}{0}:=\inf \spec{\slice{H_{P}}{n}{0}\restrict{\slice{\cal F}{n}{0}}}-\slice{V_{\mathrm{self}}}{n}{0}=\slice{E'_P}{n}{0}.
\end{align}
Since the right-hand side of (\ref{eqn:EnShift}) in \thm{thm:uv induction} is summable, the sequence $(\slice{E'_P}{n}{0})$ is convergent.
\item By \cor{lem:resolvent diff} the resolvent $(\slice{H'_{P}}{n}{0}-z)^{-1}$, for $z= \slice{E'_P}{n}{0}+i\lambda$, $\lambda\in \mathbb{R}$ and $|\Im \lambda|=1$, converges as $n\to\infty$. Furthermore, for every $n$ the range of $(\slice{H'_{P}}{n}{0}-z)^{-1}$ is given by $D(H_{P,0})$ which is dense in $\cal F$. Hence, the Trotter-Kato Theorem \cite[Theorem VIII.22]{reed_methods_1981} ensures the existence of a limiting self-adjoint Hamiltonian $\slice{H'_P}{\infty}{0}$ on $\cal F$. Because of the unitarity of the Gross transformation, the family of Hamiltonians $\slice{H_{P}}{n}{0}-\slice{V_{\mathrm{self}}}{n}{0}$, $n\in\bb N$, converges to $\slice{H_{P}}{\infty}{0}:=e^{\slice{-T}{\infty}{\kappa}}\slice{H'_P}{\infty}{0} e^{\slice{T}{\infty}{\kappa}}$ in the norm resolvent sense as $n\to\infty$.
 \item By (iii), $\slice{\Psi_P}{\infty}{0}$ is  a ground state of $\slice{H_{P}}{\infty}{0}$. Moreover, \thm{thm:uv induction} ensures 
\begin{align*}
\spec{(\slice{H'_{P}}{n}{0}-\slice{E'_P}{n}{0})\restrict{\slice{\cal F}{n}{0}}}\subset\{0\}\cup(\xi_n,\infty).
\end{align*}
 Since $\xi_n\geq \frac{1}{16}\kappa$ the set $(-\infty,0)\cup(0,\frac{1}{16}\kappa)$ is not part of the spectrum of $(\slice{H'_{P}}{n}{0}-\slice{E'_P}{n}{0})\restrict{\slice{\cal F}{n}{0}}$ for any $n\in\bb N$. As the spectrum cannot suddenly expand in the limit \cite[Theorem VIII.24]{reed_methods_1981}, this proves the claimed gap bound. The gap bound and the resolvent convergence imply that $\slice{E'_P}{\infty}{0}$ is a non-degenerate eigenvalue.
\end{enumerate}
\end{proof}

\section{Ground States of the Gross Transformed Hamiltonians $\slice{H'_P}{\infty}{m}$ for $m\in\bb N$}\label{sec:expansions}\label{sec:ground states with ir cutoff}

So far we have studied the Gross transformed Hamiltonian $\slice{H'_{P}}{n}{0}$  for an arbitrary large $n$. In the following we want to add interaction slices below the frequency $\kappa$. As a preparation for this we state some important properties of the Hamiltonian
\begin{align*}
 \slice{H'_P}{n}{m}:=\slice{H'_{P}}{n}{0}+g\slice{\Phi}{0}{m}
\end{align*}
for any $m\in\bb N\cup\{\infty\}$ and $n\in\bb N$. Note that for all such cut-offs the operator $\slice{H'_P}{n}{m}$ is a Kato small perturbation of $H_{P,0}$ and therefore self-adjoint on $D(H_{P,0})$. We collect these facts including the limiting case $n\to\infty$ in the next lemma.
 \begin{remark}
In this section  we implicitly assume  the constraints  $|P|<P_{\mathrm{max}}$ and $1<\kappa<2$. Furthermore, $g$ and $\beta$ are such that all  the results of \sct{sec:main proof} hold true. 
 \end{remark}
\begin{lemma}\label{lem:ir res conv}
 Let  $|g|$ be sufficiently small. For $n\in\bb N$, $m\in\bb N\cup\{\infty\}$ there exists $\lambda\in \mathbb{R}$ such  the operator
\begin{align*}
\frac{1}{\slice{H'_P}{n}{m}-\slice{E'_P}{n}{0}\pm i\lambda}
\end{align*}
has range $D(H_{P,0})$ and converges in norm as $n\to\infty$. Therefore, the sequence of operators $\slice{H'_P}{n}{m}$, $n\in\bb N$, converges to a self-adjoint operator acting on $\cal F$ in the norm resolvent sense.
\end{lemma}

\begin{proof}
Let $m\in\bb N\cup\{\infty\}$. The only non-straightforward case is $n\to\infty$. First, we show the validity of the Neumann expansion
\begin{align}\label{eqn:ir neumann expansion}
 \frac{1}{\slice{H'_P}{n}{m}-\slice{E'_P}{n}{0}\pm i\lambda}=\frac{1}{\slice{H'_{P}}{n}{0}+g\slice{\Phi}{0}{m}-\slice{E'_P}{n}{0}\pm i\lambda}&=R_n\sum_{j=0}^\infty (S R_n)^j
\end{align}
for
\begin{align*}
 R_n:=\frac{1}{\slice{H'_{P}}{n}{0}-\slice{E'_P}{n}{0} \pm i\lambda} && \text{and} && S=-g\slice{\Phi}{0}{m}.
\end{align*}
With the standard inequalities \eqn{eqn:standard ineq} 
%
%
we estimate
\begin{align}
 \|SR_n\|\leq C|g|\,\left\|(\slice{H^f}{0}{m})^{1/2}\left(\frac{1}{\slice{H'_{P}}{n}{0}-\slice{E'_P}{n}{0} \pm i\lambda}\right)^{1/2}\right \| \left(\frac{1}{|\lambda|}\right)^{1/2}
+C|g|\frac{1}{|\lambda|}\label{eqn:res conv bound}.
\end{align}
Fix a $\theta'$ such that $1-\theta'-\namer{c:cp}>0$. From an analogous computation as conducted in the proof of \lem{lem:gap}  one finds
\begin{align*}
 \inf_{\|\psi\|=1}\braket{\psi,(\slice{H'_{P}}{n}{0} - \theta'\slice{H^f}{0}{m} - \slice{E'_P}{n}{0})\psi}\geq 0
\end{align*}
where the infimum is taken over $\psi\in D(H_{P,0})$. Consequently, we get that 
\begin{align*}
 \left\|\left(\frac{\theta'\slice{H^f}{0}{m}}{\slice{H'_{P}}{n}{0} - \theta'\slice{H^f}{0}{m} - \slice{E'_P}{n}{0} + \theta'\slice{H^f}{0}{m} \pm i\lambda}\right)^{1/2}\right\|^2\leq \frac{1}{|\lambda|}
\end{align*}
holds because $\slice{H^f}{0}{m}$ and $\slice{H'_{P}}{n}{0}$ commute.
For $| \lambda|$ sufficiently large this gives
\begin{align}\label{eqn:RSRbound}
 \eqn{eqn:res conv bound}\leq\frac{|g |C \theta'^{-1/2}+|g|C}{| \lambda|}<1
\end{align}
so that the Neumann expansion in  \eqn{eqn:ir neumann expansion} is well-defined for all $n\in\bb N$. Moreover, the limit of \eqn{eqn:ir neumann expansion} for $n\to\infty$ exists because:
\begin{enumerate}
  \item The sequence $(R_n)_{n\in\bb N}$, converges in norm; see \thm{thm:main uv}
  \item $\|R_l S\|,\| SR_l\| < 1$ for all $l\in\bb N$, see \eqn{eqn:RSRbound}
  \item For any $j\geq 1$ we have
\begin{equation*}
 \|R_l(SR_l)^{j+1}-R_n(SR_n)^{j+1}\|\leq \|SR_l\|\;\|R_l(SR_l)^j-R_n(SR_n)^j\|+\|R_nS\|^{j+1}\;\|R_l-R_n\|.
\end{equation*}

\end{enumerate}
For all $n\in\bb N$ the range of the resolvent $(\slice{H'_P}{n}{m}- \slice{E'_P}{n}{0} \pm i\lambda)^{-1}$
 equals $D(H_{P,0})$ and therefore it is dense. Finally the Trotter-Kato Theorem \cite[Theorem VIII.22]{reed_methods_1981} ensures the existence of a self-adjoint limiting operator $\slice{H'_P}{\infty}{m}$ bounded from below.
\end{proof}

For the Hamiltonian $\slice{H'_P}{n}{m}$, where the infrared cut-off $\tau_m$ is arbitrarily small but strictly larger than zero, we construct the corresponding ground state $\slice{\Psi'_P}{n}{m}$.  For this construction we introduce a new parameter $\zeta$ and provide necessary constraints on the infrared scaling parameter $\gamma$ depending on the coupling constant $g$. 
\begin{definition}\label{def:gap params}
We consider an infrared scaling parameter $\gamma$ that obeys
\begin{align}\label{eqn:constraint}
 0<\gamma<\frac{1}{2}, && |g|\leq \gamma^2, && \sum_{j=1}^\infty \gamma^{\frac{j}{4}}(1+j)\leq \frac{1}{2}.
\end{align}
Furthermore, we fix the auxiliary constant $0<\zeta<\frac{1}{16}$ such that
\begin{align*}
 1-\theta-\namer{c:cp}\geq 2\zeta
\end{align*}
where  $0<\theta<\frac{1}{8}$ and $\namer{c:cp}=\frac{3}{4}$.
\end{definition}

As we shall see later, the upper bound on $\zeta$ is constrained by the ultraviolet gap estimate; see (iv) in \thm{thm:main uv}.

In the iterative construction of the ground state we use \cor{cor:ir ground state energies} below that relies on \lem{lem:a priori} and on \thm{thm:ground state energies} for statements  (i),(ii).  The estimate in (iii) is  based on a simple variational argument.

\begin{corollary}\label{cor:ground state energies}\label{cor:ir ground state energies}
 Let $|g|$ be sufficiently small. For all  $n,m\in\bb N$  the following holds true:
\begin{enumerate}[(i)]
 \item $-|g|\namer{b} \leq \slice{E'_P}{n}{m} \leq \frac{1}{2}P^2$, where $c_b$ is the constant introduced in \lem{lem:a priori}.
 \item There is a $g_{\mathrm{max}}>0$ such that for $0\leq |g|<g_{\mathrm{max}}$  and all $k\in\bb R^3$ 
\begin{equation}
\label{eq:Grad E}
\slice{E'_{P-k}}{n}{m}- \slice{E'_{P}}{n}{m}\geq -\namer{c:cp}|k|.
\end{equation}
 \item Assume that $\slice{E'_P}{n+1}{m}, \slice{E'_P}{n}{m+1}$, and $\slice{E'_P}{n}{m}$ are eigenvalues of $\slice{H'_P}{n+1}{m}\restrict{\slice{\cal F}{ n+1 }{m}}$, $\slice{H'_P}{n}{m+1}\restrict{\slice{\cal F}{n}{m+1}}$, and $\slice{H'_P}{n}{m}\restrict{\slice{\cal F}{n}{m}}$, respectively; then $\slice{E'_P}{n+1}{m},\slice{E'_P}{n}{m+1}\leq \slice{E'_P}{n}{m}$.

\end{enumerate}
\end{corollary} 

\begin{proof}
See Appendix A.
\end{proof}

\begin{lemma}\label{lem:ir gap}
Let $|g|$ be sufficiently small and $n\in\bb N\cup\{\infty\}$. For an integer $m\geq 1$, assume:
\begin{enumerate}[(i)]
 \item $\slice{E'_P}{n}{m-1}$ is the non-degenerate eigenvalue of $\slice{H'_P}{n}{m-1}\restrict{\slice{\cal F}{n}{m-1}}$ with eigenvector $\slice{\Psi'_P}{n}{m-1}$.
\item $\gap{\slice{H'_P}{n}{m-1}\restrict{\slice{\cal F}{n}{m-1}}}\geq\zeta\tau_{m-1}$.
\end{enumerate}
This implies that $\slice{E'_P}{n}{m-1}$ is also the non-degenerate ground state energy of $\slice{H'_P}{n}{m-1}\restrict{\slice{\cal F}{n}{m}}$ with eigenvector $\slice{\Psi'_P}{n}{m-1}\otimes \Omega$. Furthermore, it holds:
\begin{align}
\gap{\slice{H'_P}{n}{m-1}\restrict{\slice{\cal F}{n}{m}}}
&\geq \inf_{\slice{\cal F}{n}{m}\ni \psi\perp\slice{\Psi'_{P}}{n}{m-1}\otimes\Omega}\braket{\slice{H'_P}{n}{m-1}-\theta
\slice{H^f}{m-1}{m}-\slice{E'_P}{n}{m-1}}_\psi \nonumber \\
&\geq2\zeta\tau_m 
\label{ir estimate-gap}
\end{align}
where the infimum is taken over $\psi\in D(H_{P,0})$.
\end{lemma}

\begin{proof}
  Mimicking the steps in the proof \lem{lem:gap} and the inequality in (\ref{eq:Grad E}) we get the bound 
\begin{align*}
\inf_{\slice{\cal F}{n}{m}\ni\psi\perp\slice{\Psi'_P}{n}{m-1}\otimes\Omega}\braket{\slice{H'_{P}}{n}{0} + g\slice{\Phi}{0}{{m-1}} -\theta
\slice{H^f}{m-1}{{m}}-\slice{E'_P}{n}{m-1}}_\psi\geq (1-\theta-\namer{c:cp})\tau_m\geq 2\zeta\tau_m
.
\end{align*}
This gives the estimate
\begin{equation*}
  \gap{\slice{H'_P}{n}{m-1}\restrict{\slice{\cal F}{n}{m}}}=\gap{\left(\slice{H'_{P}}{n}{0}+g\slice{\Phi}{0}{{m-1}}\right)\restrict{\slice{\cal F}{n}{m}}} \geq \min\left\{\zeta\tau_{m-1},2\zeta\tau_{m}\right\}= 2\zeta\tau_{m}
\end{equation*}
where in the last step we have used that $\gamma<\frac{1}{2}$; see \eqn{eqn:constraint}. This proves the claim for any finite $n,m$. But the resolvent convergence proved in \lem{lem:ir res conv} ensures that the statements remain true in the limit $n\to\infty$ as the spectrum cannot suddenly expand in the limit \cite[Theorem VIII.24]{reed_methods_1981}.
\end{proof}

\begin{lemma}\label{lem:ir neumann}
For $n\in\bb N\cup\{\infty\}$ and $m\geq 1$ there is a $g_{\mathrm{max}}>0$ such that, for $|g|<g_{\mathrm{max}}$ and $\gamma$ fulfilling the constraints in \eqn{eqn:constraint}, the assumptions of \lem{lem:ir gap}  imply that the resolvent
\begin{equation*}
 \frac{1}{\slice{H'_P}{n}{m}-z}, 
\end{equation*}
restricted to $\slice{\cal F}{n}{m}$, is well-defined in the domain
\begin{equation}\label{eqn:ir domain}
\frac{1}{4}\zeta\tau_{m}\leq|\slice{E'_P}{n}{m-1}-z|\leq \zeta\tau_{m}.
\end{equation}

\end{lemma}

\begin{proof}
It is sufficient to show that
 \begin{align}\label{eqn:ir res in contour}
  \left\|\left(\frac{1}{\slice{ H'_P}{n}{m-1}-z}\right)^{1/2}\;g\slice{\Phi}{m-1}{m}\;\left(\frac{1}{\slice{ H'_P}{n}{m-1}-z}\right)^{1/2}\right\|_{\slice{\cal F}{n}{m}}
\end{align}
is less than one for all $z$ in the given domain. For $g$ sufficiently small this is true because:
\begin{enumerate}[1.]
 \item By standard inequalities in \eqn{eqn:standard ineq} the estimate
\begin{align}\label{eqn:ir phi}
 &\left\|g\slice{\Phi}{{m-1}}{m}\left(\frac{1}{\slice{ H'_P}{n}{m-1}-z}\right)^{1/2}\right\|_{\slice{\cal F}{n}{m}}\leq |g|C\left((1-\gamma)\tau_{m-1}\right)^{1/2}\left\|(\slice{H^f}{m-1}{m})^{1/2}\left(\frac{1}{\slice{ H'_P}{n}{m-1}-z}\right)^{1/2}\right\|_{\slice{\cal F}{n}{m}}
\end{align}
holds true. Since $\slice{H^f}{m-1}{m}$ commutes with $\slice{ H'_P}{n}{m-1}$ and using \eqn{ir estimate-gap}, the spectral theorem yields 
\begin{align}\label{eqn:eqn:hf ir}
 \left\|(\slice{H^f}{m-1}{m})^{1/2}\left(\frac{1}{\slice{ H'_P}{n}{m-1}-z}\right)^{1/2}\right\|_{\slice{\cal F}{n}{m}}\leq C.
\end{align}
 \item Using \lem{lem:ir gap} we get 
\begin{equation}
\left\|\left(\frac{1}{\slice{ H'_P}{n}{m-1}-z}\right)^{1/2}\right\|^2_{\slice{\cal F}{n}{m}}\leq \max\left\{\frac{1}{\frac{1}{4}\zeta\tau_{m}},\frac{1}{\zeta\tau_m}\right\}\leq \frac{4}{\zeta\tau_{m}}\label{eqn:gap core}.
\end{equation}

\end{enumerate}
Combining (\ref{eqn:ir phi}), (\ref{eqn:eqn:hf ir}),  and (\ref{eqn:gap core}) we find
\begin{align*}
\eqn{eqn:ir res in contour}\leq C |g| \left(\frac{\tau_{m-1}}{\tau_{m}}\right)^{1/2}=C |g| \gamma^{-1/2}\leq C |g|^{3/4}. 
\end{align*}
where we have used the constraints in \eqn{eqn:constraint}. This proves the claim.
\end{proof}
 Inside the domain where the resolvent is well-defined, let us now introduce the integration contour that is used to iteratively construct the ground state vectors in \thm{thm:ir induction} below.
\begin{definition}\label{def:IR contour}
 For $m\in\bb N$ we define the contour
\begin{align*}
 \Delta_m := \left\{z\in\bb C\;\bigg|\;|\slice{E'_{P}}{n}{m-1}-z|=\frac{1}{2}\zeta\tau_m\right\}.
\end{align*}
\end{definition}

\begin{theorem}\label{thm:ir induction}
 Let $n\in\bb N\cup\{\infty\}$ and $g,\gamma$ sufficiently small such that the constraints in \eqn{eqn:constraint} are fulfilled. Then for all $m\geq 0$ the following holds true:
\begin{enumerate}[(i)]
 \item $\slice{E'_P}{n}{m}:=\inf\spec{\slice{ H'_P}{n}{m}\restrict{\slice{\cal F}{n}{m}}}$ is the non-degenerate ground state energy of $\slice{H'_P}{n}{m}\restrict{\slice{\cal F}{n}{m}}$.
 \item $\gap{\slice{H'_P}{n}{m}\restrict{\slice{\cal F}{n}{m}}}\geq\zeta\tau_{m}$.
 \item The vectors
\begin{align}
 \slice{\Psi'_P}{n}{0}&:=\slice{\Psi'_P}{n}{0},\nonumber\\
 \slice{\Psi'_P}{n}{m}&:=\slice{\cal Q_{P}'}{n}{m}\slice{\Psi'_P}{n}{m-1}, \qquad\slice{\cal Q_{P}'}{n}{m}:=-\frac{1}{2\pi i}\oint_{\Delta_{m}}\frac{dz}{\slice{ H'_P}{n}{m}-z}, \qquad m\geq 1,\label{eqn:projection}
\end{align}
are well-defined and non-zero. The vector $\slice{\Psi'_P}{n}{m}$ is the unique ground state of $\slice{H'_P}{n}{m}\restrict{\slice{\cal F}{n}{m}}$.
\end{enumerate}
\end{theorem}
\begin{proof}
The proof is by induction and it relies on \cor{cor:ground state energies}, \lem{lem:ir gap}, and \lem{lem:ir neumann}. Since the rationale can be inferred from similar steps in the proof of  \thm{thm:uv induction},  we do not provide the details. 

\noindent
The main difference with respect to  \thm{thm:uv induction} is the fact the sequence of vectors does not converge. Moreover, here we only prove that the norm of the vector $ \slice{\Psi'_P}{n}{m}$ is nonzero for all finite $m$ that follows from the bound $\|\slice{\Psi'_P}{n}{m}\|\geq C\|\slice{\Psi'_P}{n}{m-1}\|$. The same type of argument is shown   for the vectors  $ \slice{\phi_P}{n}{m}$ (with $n$ finite) in the next section. We refer the reader to equations (\ref{eqn:tilde diff})--(\ref{bound-norm-below}).
\end{proof}
%

An auxiliary result needed for the next section is:

\begin{lemma}\label{lem:grad E}
 Let $|g|$ be sufficiently small. Then for all $n,m\in\bb N$ 
\begin{enumerate}[(i)]
 \item \begin{equation}\label{energy-shift-section-5}
 \big|\slice{E'_P}{n}{m+1}-\slice{E'_P}{n}{m}\big|\leq Cg^2\gamma^m
 \end{equation}
 \item \begin{equation}\label{gradient-bound-section-5}
 \big|\nabla\slice{E'_P}{n}{m}\big|\leq \namer{c:cp}
 \end{equation}
\end{enumerate}
hold true,  where $\nabla \slice{E'_P}{n}{m}$ is given by
\begin{align}\label{eqn:grad e formuar}
\nabla \slice{E'_P}{n}{m}=P-\braket{[P^f+\slice{B}{n}{0}+\slice{B^*}{n}{0}]}_{\slice{\Psi'_{P}}{n}{m}}.
\end{align}
\end{lemma}

\begin{proof}
\begin{enumerate}[(i)]
 \item The claim can be seen from:
   \begin{enumerate}
 \item The gap estimate \eqn{ir estimate-gap} and (i) in \cor{cor:ground state energies}.
 \item The bound
\begin{align*}
 \theta\slice{H^f}{m}{m+1}+g\slice{\Phi}{m}{m+1}+g^2\int_{\mathcal{B}_{\tau_{m}}\setminus \mathcal{B}_{\tau_{m+1}}} dk \frac{\rho(k)^2}{\theta\omega(k)}\geq 0
\end{align*}
which can be inferred from completion of the square.
\item The inequality \begin{align*}
 \int_{\mathcal{B}_{\tau_{m}}\setminus \mathcal{B}_{\tau_{m+1}}}  dk \frac{\rho(k)^2}{\theta \omega(k)}\leq \frac{C}{\theta}\gamma^m.
\end{align*}

 \end{enumerate}

\item Since $(\slice{H'_P}{n}{m})_{|P|\leq P_{\mathrm{max}}}$, is an analytic family of type A and $\slice{E'_P}{n}{m}$ is an isolated eigenvalue, equation \eqn{eqn:grad e formuar} holds by analytic perturbation theory.  Moreover, \eqn{gradient-bound-section-5} follows immediately from \cor{cor:ground state energies} (ii).


\end{enumerate}
\end{proof}

\section{Ground States of the Transformed Hamiltonians $\slice{H^{W'}_P}{n}{\infty}$ for $n\in\bb N$}\label{sec:W ground states}

 This section provides the key result for \sct{sec:Winfinity} where we remove both limits simultaneously. Here (\sct{sec:W ground states}) we generalize the strategy employed in \cite{pizzo_one-particle_2003} to perform the limit of a vanishing infrared cut-off $\tau_m$ uniformly in the ultraviolet  cut-off $\sigma_n$.

\begin{remark}
In this section  we implicitly assume  the constraints  $|P|<P_{\mathrm{max}}$ and $1<\kappa<2$. Furthermore, $g, \beta$, and  $\gamma$ are such that all  the results of Sections \ref{sec:main proof} and \ref{sec:ground states with ir cutoff} hold true. 
 \end{remark}
  
\paragraph{Preliminaries.} We collect the definitions of the transformed operators and vectors, and we explain some of their properties:
\begin{center}
\renewcommand{\arraystretch}{1.75}
\begin{tabular}{c c c c}
Hamiltonian & Fock space\\
\hline
\hline
$\slice{H_P^{W'}}{n}{m}:=W_{m}(\nabla\slice{E'_P}{n}{m})\;\slice{H'_P}{n}{m}\;W_{m}(\nabla \slice{E'_P}{n}{m})^*$ & $\slice{\cal F}{n}{m}:=\cal F(L^2(\cal B_{\sigma_n}\setminus\cal B_{\tau_m}))$\\
\hline
$\slice{\widetilde H_P^{W'}}{n}{m}:=W_{m}(\nabla \slice{E'_P}{n}{m-1})\;\slice{H'_P}{n}{m}\;W_{m}(\nabla \slice{E'_P}{n}{m-1})^*$ & $\slice{\cal F}{n}{m}$\\
\hline
\end{tabular}
\end{center}
\vskip.3cm
Notice that
\begin{align}\label{eq:Htilde_H}
\slice{\widetilde H_P^{W'}}{n}{m}=W_m(\nabla\slice{E'_P}{n}{m-1})W_m(\nabla\slice{E'_P}{n}{m})^* \slice{H_P^{W'}}{n}{m}W_m(\nabla\slice{E'_P}{n}{m})W_m(\nabla\slice{E'_P}{n}{m-1})^*.
\end{align}
The transformation $W_m(Q)$, $Q\in\bb R^3$ and $|Q|\leq 1$, was defined in \eqn{eqn:W trafo} and it is unitary for all finite $m$. For $n,m\in\bb N$ we iteratively define the vectors
\begin{align}
\begin{split}
 \slice{\phi_{P}}{n}{0}&:=\frac{\slice{\Psi'_P}{n}{0}}{\|\slice{\Psi'_P}{n}{0}\|},\\
 \slice{\widetilde\phi_{P}}{n}{m}&:=\slice{\widetilde{\cal Q'}_P}{n}{m}\slice{\phi_{P}}{n}{m-1}, \qquad \slice{\widetilde{\cal Q'}_P}{n}{m}:=-\frac{1}{2\pi i}\oint_{\Delta_m}\frac{dz}{\slice{\widetilde H^{W'}_{P}}{n}{m}-z}\\
 \slice{\phi_{P}}{n}{m}&:=W_m(\nabla \slice{E'_P}{n}{m})W_m(\nabla \slice{E'_P}{n}{m-1})^*\slice{\widetilde \phi_{P}}{n}{m} 
\end{split}\label{eqn:ir sequence}
\end{align}
where the contour $\Delta_{m}$ was introduced in \dfn{def:IR contour}.  This family of vectors is well-defined because of the unitarity of the transformations $W_m$ and of the results of  \sct{sec:ground states with ir cutoff}. If the vectors $\slice{\phi_{P}}{n}{m}$ and $\slice{\widetilde\phi_{P}}{n}{m}$ are non-zero they are by construction the (unnormalized) ground states of $\slice{H^{W'}_{P}}{n}{m}$ and $\slice{\widetilde H^{W'}_{P}}{n}{m}$, respectively. Assuming that these vectors are non-zero we introduce the following auxiliary definitions:
\begin{gather}
\begin{split}
\label{eqn:def A Ck}
 A_{P,m}^{(n)} := \int dk\;k\alpha_m(\nabla \slice{E'_P}{n}{m},k)[b(k)+b^*(k)], \qquad
 C^{(k,n)}_{P,m} := \int dk\;k\alpha_m(\nabla \slice{E'_P}{n}{m},k)^2,\\
 C^{(\omega,n)}_{P,m} := \int dk\;\omega(k)\alpha_m(\nabla \slice{E'_P}{n}{m},k)^2, \qquad
 C^{(\rho,n)}_{P,m}:= 2g\int dk\;\rho(k)\alpha_m(\nabla \slice{E'_P}{n}{m},k).
\end{split}
\end{gather}
where the function 
\begin{align*}
\alpha_m(\nabla \slice{E'_P}{n}{m},k):= -g\;\frac{\rho(k)}{\omega(k)}\frac{\charf{\cal B_\kappa\setminus\cal B_{\tau_m}}(k)}{1-\widehat k\cdot \nabla \slice{E'_P}{n}{m}}
\end{align*}
was introduced in (\ref{eqn:coherent}). Furthermore, we define
\begin{align}
 \slice{R_P}{n}{m} &:= -\nabla \slice{E'_P}{n}{m}\cdot(\slice{B}{n}{0}+\slice{B^*}{n}{0})-\frac{1}{2}\left([\slice{B}{n}{0},P-P^f]+[P-P^f,\slice{B^*}{n}{0}]+[\slice{B}{n}{0},\slice{B^*}{n}{0}]\right),\nonumber\\
 \slice{\Pi_P}{n}{m} &:= P^f+A_{P,m}^{(n)}+\slice{B}{n}{0}+\slice{B^*}{n}{0}\label{eqn:pi def}\\
 &=W_m(\nabla\slice{E'_P}{n}{m})\left(P^f+\slice{B}{n}{0}+\slice{B^*}{n}{0}\right)W_m(\nabla\slice{E'_P}{n}{m})^*-C^{(k,n)}_{P,m},\nonumber\\
 \slice{\Gamma_P}{n}{m} &:=\slice{\Pi_P}{n}{m}-\braket{\slice{\Pi_P}{n}{m}}_{\slice{\phi_P}{n}{m}},\label{eqn:def Gamma}\\
 C_{P,m}^{(n)} &:= \frac{P^2}{2}-\frac{1}{2}\left(P-\nabla\slice{E'_P}{n}{m}\right)^2-\nabla \slice{E'_P}{n}{m}\cdot C^{(k,n)}_{P,m}+C^{(\omega,n)}_{P,m}+C^{(\rho,n)}_{P,m}.\nonumber
\end{align}
Using these abbreviations and a formal computation carried out in Appendix \ref{sec:formal}, one can prove that the identity 
\begin{align}\label{eqn:Hw}
  \slice{H^{W'}_P}{n}{m}&=\frac{1}{2}\slice{\Gamma_P}{n}{m}^2+
  H^f-\nabla \slice{E'_P}{n}{m} \cdot P^f + C_{P,m}^{(n)} + \slice{R_P}{n}{m}
\end{align}
holds on $D(H_{P,0})$ for all $n,m\in\bb N$. As in \cite{pizzo_one-particle_2003} the `normal ordered' operator $\slice{\Gamma_P}{n}{m}$ will play a crucial role in the next steps.

Analogously, one can verify that on $D(H_{P,0})$ and for $n,m\in\bb N$ the following identity holds true:
\begin{align}
\label{eqn:tilde H}
  \slice{\widetilde H^{W'}_P}{n}{m}&=\frac{1}{2}\left(\slice{\Gamma_P}{n}{m-1}+\widetilde A_{P,m}^{(n)}-A_{P,m-1}^{(n)}+\widetilde C^{(k,n)}_{P,m}-C^{(k,n)}_{P,m-1}\right)^2+H^f-\nabla \slice{E'_P}{n}{m-1} \cdot P^f+\widetilde C_{P,m}^{(n)}+\slice{R_P}{n}{m-1};
\end{align}
here we have similarly introduced, for any fixed $n\in\bb N$,
\begin{gather}
\begin{split}
\label{eqn:def tilde A Ck}
 \widetilde A_{P,m}^{(n)} := \int dk\;k\alpha_m(\slice{\nabla E_P}{n}{m-1},k)[b(k)+b^*(k)], \qquad
 \widetilde C^{(k,n)}_{P,m} := \int dk\;k\alpha_m(\nabla \slice{E'_P}{n}{m-1},k)^2,\\
 \widetilde C^{(\omega,n)}_{P,m} := \int dk\;\omega(k)\alpha_m(\slice{\nabla E_P}{n}{m-1},k)^2, \qquad
 \widetilde C^{(\rho,n)}_{P,m}:= 2g\int dk\;\rho(k)\alpha_m(\slice{\nabla E_P}{n}{m-1},k),
\end{split}
\end{gather}
which differ from those in (\pref{eqn:def A Ck}) only in the argument of $\alpha_m$. We also define
\begin{align}
 \slice{\widetilde\Pi_P}{n}{m} &:= P^f+\widetilde A_{P,m}^{(n)}+\slice{B}{n}{0}+\slice{B^*}{n}{0}\nonumber\\
 &=W_m(\nabla\slice{E'_P}{n}{m-1})\left(P^f+\slice{B}{n}{0}+\slice{B^*}{n}{0}\right)W_m(\nabla\slice{E'_P}{n}{m-1})^*-{\widetilde C}^{(k,n)}_{P,m},\nonumber\\
 \slice{\widetilde \Gamma_P}{n}{m} &:=\slice{\widetilde\Pi_P}{n}{m}-\braket{\slice{\widetilde\Pi_P}{n}{m}}_{\slice{\widetilde\phi_P}{n}{m}},\label{eqn:def tilde Gamma}\\
 \widetilde C_{P,m}^{(n)} &:= \frac{P^2}{2}-\frac{1}{2}\left(P-\nabla\slice{E'_P}{n}{m-1}\right)^2-\nabla \slice{E'_P}{n}{m-1}\cdot C^{(k,n)}_{P,m}+\widetilde C^{(\omega,n)}_{P,m}+\widetilde C^{(\rho,n)}_{P,m}.\nonumber
\end{align}
Notice that using (\ref{eqn:grad e formuar}) we have the following identities
\begin{align}
\braket{\slice{\Pi_P}{n}{m}}_{\slice{\phi_P}{n}{m}}&=P-\nabla \slice{E'_P}{n}{m}-C^{(k,n)}_{P,m}, \\
 \slice{\Gamma_P}{n}{m}&=W_m(\nabla\slice{E'_P}{n}{m})\left(P^f+\slice{B}{n}{0}+\slice{B^*}{n}{0}\right)W_m(\nabla\slice{E'_P}{n}{m})^*-P+\nabla \slice{E'_P}{n}{m}, \label{gamma-identity} \\
 \slice{\widetilde \Gamma_P}{n}{m}&=W_m(\nabla\slice{E'_P}{n}{m-1})W_m(\nabla\slice{E'_P}{n}{m})^* \slice{\Gamma_P}{n}{m}W_m(\nabla\slice{E'_P}{n}{m})W_m(\nabla\slice{E'_P}{n}{m-1})^*,\label{eqn:gamma tilde}
 \\
\label{eqn:gamma diff}
 \slice{\widetilde \Gamma_P}{n}{m}-\slice{\Gamma_{P}}{n}{m-1}&= (\nabla \slice{E'_P}{n}{m}-\nabla \slice{E'_P}{n}{m-1}) +(\widetilde A_{P,m}^{(n)}-A_{P,m-1}^{(n)})+(\widetilde C^{(k,n)}_{P,m}-C^{(k,n)}_{P,m-1}).
\end{align}

To start with, we show that for any finite $m$, the vectors $\slice{\phi_{P}}{n}{m}$ and $\slice{\widetilde\phi_{P}}{n}{m}$ are non-zero. Namely,  by starting from $\slice{\phi_P}{n}{0}$, we estimate the norm difference
\begin{align}\label{eqn:diff of res}
 \|\slice{\widetilde\phi_P}{n}{m}-\slice{\phi_P}{n}{m-1}\|=\left\|-\frac{1}{2\pi i}\oint_{\Delta_m}\frac{dz}{\slice{\widetilde H^{W'}_{P}}{n}{m}-z}\slice{\phi_{P}}{n}{m-1}-\slice{\phi_{P}}{n}{m-1}\right\|.
\end{align}
In \eqn{eqn:diff of res} we expand the resolvent with respect to 
 \begin{align}\label{eqn:tilde h}
 \slice{\Delta \widehat H_P^{W'}}{m-1}{m}&:=\slice{\widetilde H_P^{W'}}{n}{m}-\slice{H_P^{W'}}{n}{m-1}-\widetilde C_{P,m}^{(n)}+C_{P,m-1}^{(n)}\\
 &=\frac{1}{2}\left(\widetilde A_{P,m}^{(n)}-A_{P,m-1}^{(n)}+\widetilde C^{(k,n)}_{P,m}- C^{(k,n)}_{P,m-1}\right)^2+\label{eqn:quad term}\\
 &\quad+\frac{1}{2}\left[\widetilde A_{P,m}^{(n)}-A_{P,m-1}^{(n)},\slice{\Gamma_P}{n}{m-1}\right]+\label{eqn:quad term 2}\\
 &\quad+\left(\widetilde A_{P,m}^{(n)}-A_{P,m-1}^{(n)}\right)\cdot\slice{\Gamma_P}{n}{m-1}+\left(\widetilde C^{(k,n)}_{P,m}- C^{(k,n)}_{P,m-1}\right)\cdot\slice{\Gamma_P}{n}{m-1}.\label{eqn:mixed term 2}
\end{align}

Given the form of $\slice{\Delta \widehat H_P^{W'}}{m-1}{m}$ it is convenient to replace the integration contour $\Delta_m$ with $\widehat\Delta_m$ defined below:

\begin{definition}
 For $m\in\bb N$ define
\begin{align*}
 \widehat\Delta_m:=\left\{z-(C_{P,m-1}^{(n)}-\widetilde C_{P,m}^{(n)})\;\big|\;z\in\Delta_m\right\}.
\end{align*}
\end{definition}

In the same fashion as Theorem \pref{thm:ir induction}  we ensure the bounds
\begin{align}\label{eqn:energy est}
 \frac{1}{4}\zeta\tau_{m}\leq |\slice{E'_P}{n}{m-1}-z+\widetilde C_{P,m}^{(n)}-C_{P,m-1}^{(n)}|\leq \zeta\tau_m.
\end{align}
for $z$ in the original integration contour $\Delta_m$.
For this we observe that
\begin{align}\label{eqn:c est}
 |\widetilde C_{P,m}^{(n)}-C_{P,m-1}^{(n)}|\leq g^2C\tau_{m-1},
\end{align}
and hence, for $|g|$ sufficiently small,
\begin{align*}
 \zeta\tau_{m}\geq \frac{1}{2}\zeta\tau_{m}+g^2 C\tau_{m-1}\geq |\slice{E'_P}{n}{m-1}-z+\widetilde C_{P,m}^{(n)}-C_{P,m-1}^{(n)}|\geq \frac{1}{2}\zeta\tau_m-g^2 C\tau_{m-1}\geq \frac{1}{4}\zeta\tau_{m}
\end{align*}
where in the last step we have used the constraints in \eqn{eqn:constraint}. The upper bound \eqn{eqn:energy est} follows from \eqn{eqn:c est} by a similar argument. Hence, we can use the shifted contours $\widehat\Delta_m$ instead of $\Delta_m$ and
 estimate
\begin{align}
\|\slice{\widetilde\phi_{P}}{n}{m}-\slice{\phi_{P}}{n}{m-1}\| &\leq 
   \Bigg\|\frac{1}{2\pi i}\oint_{\widehat\Delta_m}dz\;\sum_{j=1}^\infty \left(\frac{1}{\slice{E'_P}{n}{m-1}-z}\right)^{1/2}\left(\frac{1}{\slice{H^{W'}_P}{n}{m-1}-z}\right)^{1/2}\times\label{eqn:tilde diff}\\
   &\quad\times\left[\left(\frac{1}{\slice{H^{W'}_P}{n}{m-1}-z}\right)^{1/2}(-\slice{\Delta \widehat H^{W'}_n}{m-1}{m})\left(\frac{1}{\slice{H^{W'}_P}{n}{m-1}-z}\right)^{1/2}\right]^j\slice{\phi_{P}}{n}{m-1}\Bigg\|\nonumber\\
 &\leq 
   C\gamma^m \sup_{z\in\widehat\Delta_m} \left|\frac{1}{\slice{E'_P}{n}{m-1}-z}\right|^{1/2} \; \left\|\left(\frac{1}{\slice{H^{W'}_P}{n}{m-1}-z}\right)^{1/2}\right\|_{\slice{\cal F}{n}{m}} \times\label{eqn:ir diff 1}\\
   &\quad\times\sum_{j=1}^\infty\left\|\left(\frac{1}{\slice{H^{W'}_P}{n}{m-1}-z}\right)^{1/2}\slice{\Delta \widehat H^{W'}_n}{m-1}{m}\left(\frac{1}{\slice{H^{W'}_P}{n}{m-1}-z}\right)^{1/2}\right\|^{j-1}_{\slice{\cal F}{n}{m}}\times\label{eqn:ir diff 2}\\
   &\quad\times\left\|\left(\frac{1}{\slice{H^{W'}_P}{n}{m-1}-z}\right)^{1/2}\slice{\Delta \widehat H^{W'}_n}{m-1}{m}\left(\frac{1}{\slice{H^{W'}_P}{n}{m-1}-z}\right)^{1/2}\slice{\phi_{P}}{n}{m-1}\right\|\label{eqn:sandwich}.
\end{align}
Firstly, the gap estimate in \eqn{eqn:gap core} immediately yields
\begin{align*}
 \sup_{z\in\widehat\Delta_m} \left|\frac{1}{\slice{E'_P}{n}{m-1}-z}\right|^{1/2}\;\left\|\left(\frac{1}{\slice{H^{W'}_P}{n}{m-1}-z}\right)^{1/2}\right\|_{\slice{\cal F}{n}{m}} \leq \frac{C}{\gamma^{m}}
\end{align*}
so that \eqn{eqn:ir diff 1} is bounded by a constant. Secondly, we show that the series in \eqn{eqn:ir diff 2} is convergent. We remark that $(\widetilde A_{P,m}^{(n)}-A_{P,m-1}^{(n)})$ commutes with $W_{m-1}(\nabla\slice{E'_P}{n}{m-1})$ so that
\begin{multline*}
\left\|\left(\frac{1}{\slice{H^{W'}_P}{n}{m-1}-z}\right)^{1/2}\left(\widetilde A_{P,m}^{(n)}-A_{P,m-1}^{(n)}\right)\cdot\slice{\Gamma_P}{n}{m-1}\left(\frac{1}{\slice{H^{W'}_P}{n}{m-1}-z}\right)^{1/2}\right\|_{\slice{\cal F}{n}{m}}\\
=\Bigg\|\left(\frac{1}{\slice{H'_P}{n}{m-1}-z}\right)^{1/2}\left(\widetilde A_{P,m}^{(n)}-A_{P,m-1}^{(n)}\right)\cdot\left(P^f+\slice{B}{n}{0}+\slice{B^*}{n}{0}+\nabla\slice{E'_P}{n}{m-1}-P\right)\left(\frac{1}{\slice{H'_P}{n}{m-1}-z}\right)^{1/2}\Bigg\|_{\slice{\cal F}{n}{m}},
\end{multline*}
where we used again the unitarity of $W_{m-1}$. Since $(\widetilde A_{P,m}^{(n)}-A_{P,m-1}^{(n)})$ commutes with $\slice{B}{n}{0}\,,\,\slice{B^*}{n}{0}$ it is enough to bound
\begin{align}
\left\|\left(\frac{1}{\slice{H'_P}{n}{m-1}-z}\right)^{1/2} \right. & \left. \left(\widetilde A_{P,m}^{(n)}-A_{P,m-1}^{(n)}\right)\cdot[P^f-P+\slice{B}{n}{0}]\left(\frac{1}{\slice{H'_P}{n}{m-1}-z}\right)^{1/2}\right\|_{\slice{\cal F}{n}{m}}\nonumber\\
 &\leq C\left\|\left(\widetilde A_{P,m}^{(n)}-A_{P,m-1}^{(n)}\right)\cdot\left(\frac{1}{\slice{H'_P}{n}{m-1}-z}\right)^{1/2}\right\|_{\slice{\cal F}{n}{m}}\times\label{eqn:e1}\\
 &\quad\times\left[ \left\|H_{P,0}^{1/2}\left(\frac{1}{\slice{H'_P}{n}{m-1}-z}\right)^{1/2}\right\|_{\slice{\cal F}{n}{m}}+ \left\|\slice{B}{n}{0}\left(\frac{1}{\slice{H'_P}{n}{m-1}-z}\right)^{1/2}\right\|_{\slice{\cal F}{n}{m}}\right]\label{eqn:e2}
 \end{align}
 The factor \eqn{eqn:e1} can be bounded by $C|g|\gamma^{(m-1)/2}$, similarly to \eqn{eqn:ir phi}. Both terms in \eqn{eqn:e2} can be estimated as $C |g|\gamma^{-m/2}$ using inequalities (\pref{eqn:standard ineq})-(\ref{eqn:a priori}) and the uniform bound on $|\slice{E'_P}{n}{m}|$ given by \cor{cor:ground state energies}; see an analogous argument in \eqn{eqn:remaining resolvent} that exploits the bound in (\ref{eqn:a priori-2}). All the remaining terms can be controlled in a similar fashion. Hence, for $|g|$ sufficiently small and $\gamma$ satisfying the constraint \eqn{eqn:constraint}, we conclude that
\begin{align}\label{bound-norm-below}
 \|\slice{\widetilde\phi_P}{n}{m}\|\geq C\|\slice{\phi_P}{n}{m-1}\|
\end{align}
for a strictly positive constant $C$. 
 

\paragraph{Key result.} 
\thm{thm:key theorem} below is the key tool needed for proving the second main result of this paper, namely that the ground states $(\slice{\phi_P}{n}{m})_{m\in\bb N}$ converge to a non-zero vector. This theorem relies on several lemmas (\lem{lemma:pi terms}, \lem{lem:leading term}, and \lem{lem:ir trafo diff}) that will be proven later on. 

Recall that the symbol $C$ denotes any universal constant. Throughout the computation it will be important to distinguish the constants $C_i$, $1\leq i\leq 7$. 

\begin{theorem}\label{thm:key theorem}
For $|g|$, $\gamma$, and $\zeta$ sufficiently small and fulfilling the constraints in \dfn{def:gap params} the following holds true for all $n\in\bb N$, $m\geq 1$:
\begin{enumerate}[(i)]
 \item $\|\slice{\phi_{P}}{n}{m}-\slice{\widetilde\phi_{P}}{n}{m}\|\leq m\gamma^{\frac{m}{4}}$ and $\|\slice{\widetilde\phi_P}{n}{m}-\slice{\phi_P}{n}{m-1}\|\leq \gamma^{\frac{m}{4}}$,
 \item $\|\slice{\phi_P}{n}{m}\|\geq 1-\sum_{j=1}^m \gamma^{\frac{j}{4}}(1+j)\qquad (\geq \frac{1}{2})$,
 \item Let $z\in\widehat\Delta_{m+1}$ and $\delta:=\frac{1}{2}$ 
 then
 \begin{align*}
 |g|^{\delta}\left|\braket{\slice{\Gamma^{(i)}_P}{n}{m}\slice{\phi_P}{n}{m},\frac{1}{\slice{H^{W'}_P}{n}{m}-z}\slice{\Gamma^{(i)}_P}{n}{m}\slice{\phi_P}{n}{m}}\right|\leq \gamma^{-\frac{m}{2}}, \qquad i=1,2,3.
\end{align*}
\end{enumerate}
\end{theorem}

\begin{proof}
We prove this by induction: Statements (i)-(iii) for $(m-1)$ shall be referred to as assumptions A(i)-A(iii) while the same statements for $m$ are referred to as claims C(i)-C(iii).

A straightforward computation yields the case $m=1$.

Let $m\geq 2$ and suppose A(i)-A(iii) hold. We start proving claims C(i) and C(ii).
\begin{enumerate}
 \item 
Due to the inequality in \eqn{eqn:ir diff 1}-\eqn{eqn:sandwich}, the estimate
\begin{equation*}
\|\slice{\widetilde\phi_{P}}{n}{m}-\slice{\phi_{P}}{n}{m-1}\| \leq C_1\left\|\left(\frac{1}{\slice{H^{W'}_P}{n}{m-1}-z}\right)^{1/2}\slice{\Delta \widehat H^{W'}_n}{m-1}{m}\left(\frac{1}{\slice{H^{W'}_P}{n}{m-1}-z}\right)^{1/2}\slice{\phi_P}{n}{m-1}\right\|
\end{equation*}
holds true for $|g|$ sufficiently small, uniformly in $n$ and $m$. Furthermore, \lem{lem:leading term} states that
\begin{multline*}
\left\|\left(\frac{1}{\slice{H^{W'}_P}{n}{m-1}-z}\right)^{1/2}\slice{\Delta \widehat H^{W'}_n}{m-1}{m}\left(\frac{1}{\slice{H^{W'}_P}{n}{m-1}-z}\right)^{1/2}\slice{\phi_P}{n}{m-1}\right\|\\
\leq |g|C_2\gamma^{\frac{m-2}{2}}\left(1+
 \sum_{i=1}^3\left|\braket{\slice{\Gamma^{(i)}_P}{n}{m-1}\slice{\phi_P}{n}{m-1},\frac{1}{\slice{H^{W'}_P}{n}{m-1}-z}\slice{\Gamma^{(i)}_P}{n}{m-1}\slice{\phi_P}{n}{m-1}}\right|^{\frac{1}{2}}\right)
 \end{multline*}
which together with the induction assumption A(iii) yields
\begin{align*}
 \|\slice{\widetilde\phi_P}{n}{m}-\slice{\phi_P}{n}{m-1}\|\leq |g|C_1 C_2 \gamma^{\frac{m-2}{2}}\left(1+3|g|^{-\frac{\delta}{2}}\gamma^{-\frac{m-1}{4}}\right).
\end{align*}
For $|g|$ sufficiently small and $\gamma$ satisfying the constraints in \eqn{eqn:constraint} we have
\begin{align}\label{eqn:triangle first term}
 \|\slice{\widetilde\phi_P}{n}{m}-\slice{\phi_P}{n}{m-1}\|\leq \gamma^{\frac{m}{4}}.
\end{align}

Finally, from \eqn{eqn:triangle first term}, A(ii) and \eqn{eqn:constraint} we conclude
\begin{equation}
\label{eqn:norm estimate}
 \|\slice{\widetilde \phi_{P}}{n}{m}\| \geq \|\slice{\phi_P}{n}{m-1}\|-\|\slice{\widetilde\phi_P}{n}{m}-\slice{\phi_P}{n}{m-1}\| \geq 1-\sum_{j=1}^{m-1}\gamma^{\frac{j}{4}}(1+j)-\gamma^{\frac{m}{4}}\geq \frac{1}{2}.
\end{equation}

\item
We observe that
\begin{align}
\|\slice{\phi_{P}}{n}{m}-\slice{\widetilde\phi_{P}}{n}{m}\|
&\leq \|[W_m(\nabla \slice{E'_P}{n}{m})W_m(\nabla \slice{E'_P}{n}{m-1})^*-\id{\slice{\cal F}{n}{m}}]\;\slice{\widetilde\phi_P}{n}{m}\|
\nonumber\\
&\leq \left\|[W_m(\nabla \slice{E'_P}{n}{m})-W_m(\nabla \slice{E'_P}{n}{m-1})]\;\frac{\slice{\Psi'_P}{n}{m}}{\|\slice{\Psi'_P}{n}{m}\|}\right\|\label{eqn:est trafo diff}
\end{align}
holds because the vectors $\slice{\Psi'_P}{n}{m}$ and $W_m(\nabla \slice{E'_P}{n}{m-1})^*\slice{\widetilde\phi_P}{n}{m}$ are parallel and $\|\slice{\widetilde\phi_P}{n}{m}\|\leq 1$.
\lem{lem:ir trafo diff} yields
\begin{align}\label{eqn:ln term}
 \eqn{eqn:est trafo diff}\leq |g|\, C_3 m\; |\ln \gamma|\bigg|\nabla \slice{E'_P}{n}{m}-\nabla \slice{E'_P}{n}{m-1}\bigg|.
\end{align}
The difference of the gradients of the ground state energies in \eqn{eqn:ln term} is estimated in \lem{lem:grad E diff} which states that
\begin{align*}
 |\nabla \slice{E'_P}{n}{m}-\nabla \slice{E'_P}{n}{m-1}|&\leq g^2C_4\gamma^{\frac{m-1}{2}}+\sup_{z\in\widehat\Delta_m}\left\|\left(\frac{1}{\slice{H^{W'}_P}{n}{m}-z}\right)^{1/2}\slice{\Delta \widehat H^{W'}_n}{m-1}{m}\left(\frac{1}{\slice{H^{W'}_P}{n}{m}-z}\right)^{1/2}\slice{\phi_P}{n}{m-1}\right\|\\
&\quad+C\frac{\|\slice{\phi_P}{n}{m-1}-\slice{\widetilde\phi_P}{n}{m}\|}{\|\slice{\phi_P}{n}{m-1}\|^2\|\slice{\widetilde\phi_P}{n}{m}\|^2}.
\end{align*}
Hence, using \lem{lem:leading term}, \eqn{eqn:triangle first term}, \eqn{eqn:norm estimate} as well as assumptions A(ii) and A(iii), one finds that
\begin{align*}
 \|\slice{\phi_{P}}{n}{m}-\slice{\widetilde\phi_{P}}{n}{m}\|\leq |g| C_3 m|\ln \gamma|\left(g^2C_4\gamma^{m-1}+|g| C_2\gamma^{\frac{m-2}{2}}(1+3|g|^{-\frac{\delta}{2}}\gamma^{-\frac{m-1}{4}})+C_5\gamma^{\frac{m}{4}}\right)
\end{align*}
which implies 
\begin{align}\label{eqn:triangle second term}
 \|\slice{\phi_{P}}{n}{m}-\slice{\widetilde\phi_{P}}{n}{m}\|\leq m\gamma^{\frac{m}{4}}
\end{align}
for $|g|$ sufficiently small and $\gamma$  fulfilling the constraints in \eqn{eqn:constraint}.
\end{enumerate}
Estimates \eqn{eqn:triangle first term} and \eqn{eqn:triangle second term} prove
C(i). C(ii) follows along the same lines as \eqn{eqn:norm estimate} using the bound in \eqn{eqn:triangle second term}.

Finally, we prove claim C(iii).
Let $z\in\widehat\Delta_{m+1}$ and $i=1,2,3$. Using the unitarity of the transformations $W_m$ we get
\begin{align*}
 &\left|\braket{\slice{\Gamma^{(i)}_P}{n}{m}\slice{\phi_P}{n}{m},\frac{1}{\slice{H^{W'}_P}{n}{m}-z}\slice{\Gamma^{(i)}_P}{n}{m}\slice{\phi_P}{n}{m}}\right|=\left|\braket{\slice{\widetilde\Gamma^{(i)}_P}{n}{m}\slice{\widetilde\phi_P}{n}{m},\frac{1}{\slice{\widetilde H^{W'}_P}{n}{m}-z}\slice{\widetilde\Gamma^{(i)}_P}{n}{m}\slice{\widetilde\phi_P}{n}{m}}\right|,
\end{align*}
see identities \eqn{eq:Htilde_H}-\eqn{eqn:gamma tilde}.
For $|g|$ sufficiently small,  i.e., $|g|$ of order $\gamma^2$, we can expand the resolvent $(\slice{\widetilde H^{W'}_P}{n}{m}-z)^{-1}$ by the same reasoning as for \eqn{eqn:tilde diff}-\eqn{eqn:sandwich} even for $z\in\widehat\Delta_{m+1}$ because of the bound on the energy shifts
\begin{align}\label{eqn:the eng shift}
 \big|\slice{E'_P}{n}{m+1}-\slice{E'_P}{n}{m}\big|\leq Cg^2\gamma^m,
\end{align}
given by \lem{lem:grad E}, and because of \eqn{eqn:ir domain}. Hence, using (\ref{eqn:tilde h}) we find
\begin{align}
 &\left|\braket{\slice{\widetilde\Gamma^{(i)}_P}{n}{m}\slice{\widetilde\phi_P}{n}{m},\frac{1}{\slice{\widetilde H^{W'}_P}{n}{m}-z}\slice{\widetilde\Gamma^{(i)}_P}{n}{m}\slice{\widetilde\phi_P}{n}{m}}\right| \nonumber \\
 &\qquad\leq \sum_{j=1}^\infty \left\|\left(\frac{1}{\slice{H^{W'}_P}{n}{m-1}-z}\right)^{1/2}\big[\slice{\Delta \widehat H^{W'}_n}{m-1}{m}+\widetilde C_{P,m}^{(n)}-C_{P,m-1}^{(n)}\big]\left(\frac{1}{\slice{H^{W'}_P}{n}{m-1}-z}\right)^{1/2}\right\|^{j-1}_{\slice{\cal F}{n}{m}}\times\label{eqn:neumann sum}\\
 &\qquad\phantom{\leq}\times\left\|\left(\frac{1}{\slice{H^{W'}_P}{n}{m-1}-z}\right)^{1/2}\slice{\widetilde\Gamma^{(i)}_P}{n}{m}\slice{\widetilde\phi_P}{n}{m}\right\|^2\nonumber\\
 &\qquad\leq C\left\|\left(\frac{1}{\slice{H^{W'}_P}{n}{m-1}-z}\right)^{1/2}\slice{\widetilde\Gamma^{(i)}_P}{n}{m}\slice{\widetilde\phi_P}{n}{m}\right\|^2\nonumber
\end{align}
 Furthermore,
\begin{align}
\left\|\left(\frac{1}{\slice{H^{W'}_P}{n}{m-1}-z}\right)^{1/2}\slice{\widetilde\Gamma^{(i)}_P}{n}{m}\slice{\widetilde\phi_P}{n}{m}\right\|^2 &\leq 2\left\|\left(\frac{1}{\slice{H^{W'}_P}{n}{m-1}-z}\right)^{1/2}\slice{\Gamma^{(i)}_P}{n}{m-1}\slice{\phi_P}{n}{m-1}\right\|^2+\label{eqn:almost our induction assumption}\\
 &\quad+2\left\|\left(\frac{1}{\slice{H^{W'}_P}{n}{m-1}-z}\right)^{1/2}(\slice{\widetilde\Gamma^{(i)}_P}{n}{m}\slice{\widetilde\phi_P}{n}{m}-\slice{\Gamma^{(i)}_P}{n}{m-1}\slice{\phi_P}{n}{m-1})\right\|^2.\label{eqn:anyway small term}
\end{align}
Term \eqn{eqn:almost our induction assumption}: Exploiting the property
\begin{align*}
\braket{\slice{\phi_P}{n}{m-1},\slice{\Gamma_P}{n}{m-1}\slice{\phi_P}{n}{m-1}}=0 
\end{align*}
and the spectral theorem, one can show that the term on the right-hand side of \eqn{eqn:almost our induction assumption} fulfills
\begin{align}
 &\left\|\left(\frac{1}{\slice{H^{W'}_P}{n}{m-1}-z}\right)^{1/2}\slice{\Gamma^{(i)}_P}{n}{m-1}\slice{\phi_P}{n}{m-1}\right\|^2 =\braket{\slice{\Gamma^{(i)}_P}{n}{m-1}\slice{\phi_P}{n}{m-1},\left|\frac{1}{\slice{H^{W'}_P}{n}{m-1}-z}\right|\slice{\Gamma^{(i)}_P}{n}{m-1}\slice{\phi_P}{n}{m-1}}\nonumber\\
 &\qquad\leq C \left|\braket{\slice{\Gamma^{(i)}_P}{n}{m-1}\slice{\phi_P}{n}{m-1},\frac{1}{\slice{H^{W'}_P}{n}{m-1}-z}\slice{\Gamma^{(i)}_P}{n}{m-1}\slice{\phi_P}{n}{m-1}}\right|\label{eqn:z res}\\
 &\qquad\leq C \left|\braket{\slice{\Gamma^{(i)}_P}{n}{m-1}\slice{\phi_P}{n}{m-1},\frac{1}{\slice{H^{W'}_P}{n}{m-1}-y}\slice{\Gamma^{(i)}_P}{n}{m-1}\slice{\phi_P}{n}{m-1}}\right|\label{eqn:z prime res}\\
 &\qquad\quad+ C\frac{\sup_{y\in\widehat\Delta_m,z\in\widehat\Delta_{m+1}}|z-y|}{\mathrm{dist}\left(z,\spec{\slice{H'_P}{n}{m-1}\restrict{\slice{\cal F}{n}{m-1}}}\setminus\{\slice{E'_P}{n}{m-1}\}\right)} \left|\braket{\slice{\Gamma_P}{n}{m-1}\slice{\phi}{n}{m-1},\frac{1}{\slice{H'_P}{n}{m-1}-y}\slice{\Gamma_P}{n}{m-1}\slice{\phi}{n}{m-1}}\right|\label{eqn:dist_term}\\
 &\qquad\leq C_7 \left|\braket{\slice{\Gamma_P}{n}{m-1}\slice{\phi}{n}{m-1},\frac{1}{\slice{H'_P}{n}{m-1}-y}\slice{\Gamma_P}{n}{m-1}\slice{\phi}{n}{m-1}}\right|.
\end{align}
for $y\in\widehat\Delta_{m}$ (recall that $z\in\widehat\Delta_{m+1}$). In passing from \eqn{eqn:z res} to \eqn{eqn:z prime res} we have used the property $\braket{\slice{\Gamma_P}{n}{m-1}\slice{\phi}{n}{m-1},\slice{\phi_P}{n}{m-1}}=0$ which implies that the vector $\slice{\Gamma^{(i)}_P}{n}{m-1}\slice{\phi_P}{n}{m-1}$ has spectral support (with respect to $\slice{H^{W'}_P}{n}{m-1}$) contained in the interval $(\slice{E'_P}{n}{m-1}+\zeta\tau_{m-1},\infty)$, and hence:
\begin{enumerate}[a)]
 \item
\begin{align*}
 \braket{\slice{\Gamma_P}{n}{m-1}\slice{\phi}{n}{m-1},\left|\frac{1}{\slice{H'_P}{n}{m-1}-y}\right|\slice{\Gamma_P}{n}{m-1}\slice{\phi}{n}{m-1}}\leq C\left|\braket{\slice{\Gamma_P}{n}{m-1}\slice{\phi}{n}{m-1},\frac{1}{\slice{H'_P}{n}{m-1}-y}\slice{\Gamma_P}{n}{m-1}\slice{\phi}{n}{m-1}}\right|
\end{align*}
 \item
\begin{multline*}
 \left|\braket{\slice{\Gamma_P}{n}{m-1}\slice{\phi}{n}{m-1},\frac{1}{\slice{H'_P}{n}{m-1}-z}\frac{1}{\slice{H'_P}{n}{m-1}-y}\slice{\Gamma_P}{n}{m-1}\slice{\phi}{n}{m-1}}\right|\\
 \leq \frac{1}{\mathrm{dist}\left(z,\spec{\slice{H'_P}{n}{m-1}\restrict{\slice{\cal F}{n}{m-1}}}\setminus\{\slice{E'_P}{n}{m-1}\}\right)}\braket{\slice{\Gamma_P}{n}{m-1}\slice{\phi}{n}{m-1},\left|\frac{1}{\slice{H'_P}{n}{m-1}-y}\right|\slice{\Gamma_P}{n}{m-1}\slice{\phi}{n}{m-1}}.
\end{multline*}
\end{enumerate}
In the step from \eqn{eqn:z res}-\eqn{eqn:dist_term} we used inequality \eqn{eqn:the eng shift}. Therefore, we can conclude that
\begin{align}
 \eqn{eqn:almost our induction assumption}\leq C_7\left|\braket{\slice{\Gamma_P}{n}{m-1}\slice{\phi}{n}{m-1},\frac{1}{\slice{H'_P}{n}{m-1}-y}\slice{\Gamma_P}{n}{m-1}\slice{\phi}{n}{m-1}}\right|.\label{eqn:C7}
\end{align}
Term \eqn{eqn:anyway small term}: We first observe that
\begin{align}
 \eqn{eqn:anyway small term}&\leq 4\left\|\left(\frac{1}{\slice{H^{W'}_P}{n}{m-1}-z}\right)^{1/2}(\slice{\widetilde\Gamma^{(i)}_P}{n}{m}-\slice{\Gamma^{(i)}_P}{n}{m-1})\slice{\widetilde\phi_P}{n}{m}\right\|^2+\label{eqn:first worm}\\
 &\quad+4\left\|\left(\frac{1}{\slice{H^{W'}_P}{n}{m-1}-z}\right)^{1/2}\slice{\Gamma^{(i)}_P}{n}{m-1}(\slice{\widetilde\phi_P}{n}{m}-\slice{\phi_P}{n}{m-1})\right\|^2\label{eqn:second worm}.
\end{align} 
In order to estimate \eqn{eqn:first worm} we use the identity in \eqn{eqn:gamma diff} and the following ingredients:
\begin{enumerate}[a)]
\setcounter{enumi}{2}
\item  The bound on $|\nabla \slice{E'_P}{n}{m}-\nabla \slice{E'_P}{n}{m-1}|$ from \lem{lem:grad E diff}
 \item The estimate in \eqn{eqn:c est}, i.e. $|\widetilde C^{(k,n)}_{P,m}-C^{(k,n)}_{P,m-1}|\leq g^2C\gamma^{m-1}$
 \item The bound
\begin{equation*}
\left\|\left(\frac{1}{\slice{H^{W'}_P}{n}{m-1}-z}\right)^{1/2}\int dk\;k[\alpha_m(\nabla \slice{E'_P}{n}{m-1},k)-\alpha_{m-1}(\nabla \slice{E'_P}{n}{m-1},k)](b(k)+b^*(k))\right\|_{\slice{\cal F}{n}{m}}^2 \leq g^2C\gamma^{m-3}.
\end{equation*}

\end{enumerate}

\noindent Hence, we obtain
\begin{align}
 \eqn{eqn:first worm}
 &\leq \frac{C}{\tau_{m+1}}\Bigg[g^2\tau_{m-1}^{1/2}+\sup_{y\in\widehat\Delta_m}\bigg\|\left(\frac{1}{\slice{H^{W'}_P}{n}{m-1}-y}\right)^{1/2}\slice{\Delta \widehat H^{W'}_n}{m-1}{m}\left(\frac{1}{\slice{H^{W'}_P}{n}{m-1}-y}\right)^{1/2}\slice{\phi_P}{n}{m-1}\bigg\|\label{eqn:w1}+ \\
&\qquad\qquad+\frac{\|\slice{\phi_P}{n}{m-1}-\slice{\widetilde\phi_P}{n}{m}\|}{\|\slice{\phi_P}{n}{m-1}\|^2\|\slice{\widetilde\phi_P}{n}{m}\|^2}\Bigg]^2+ \label{eqn:w2}\\
&\quad + \frac{C}{\tau_{m+1}}\left[g^2 C\gamma^{m-1}\right]^{2}\label{eqn:w5}\\
&\quad + g^2 C\gamma^{m-3}\label{eqn:w4}
\end{align}
where \eqn{eqn:w1}-\eqn{eqn:w2}, \eqn{eqn:w5} and \eqn{eqn:w4} are related to ingredients c), d) and e) respectively. 

For the remaining term \eqn{eqn:second worm} we use analytic perturbation theory to find
\begin{align*}
 \sqrt{\eqn{eqn:second worm}}
 &\leq C\tau_{m} \sup_{y\in\widehat\Delta_m}\left\|\left(\frac{1}{\slice{H^{W'}_P}{n}{m-1}-z}\right)^{1/2}\slice{\Gamma^{(i)}_P}{n}{m-1}\left(\frac{1}{\slice{H^{W'}_P}{n}{m-1}-y}\right)^{1/2}\right\|_{\slice{\cal F}{n}{m}}\times\\
 &\quad\times \sum_{j=1}^\infty \left\|\left(\frac{1}{\slice{H^{W'}_P}{n}{m-1}-y}\right)^{1/2}\slice{\Delta\widehat H^{W'}_n}{m-1}{m}\left(\frac{1}{\slice{H^{W'}_P}{n}{m-1}-y}\right)^{1/2}\right\|_{\slice{\cal F}{n}{m}}^{j-1}\times\\
 &\quad\times\left\|\left(\frac{1}{\slice{H^{W'}_P}{n}{m-1}-y}\right)^{1/2}\slice{\Delta\widehat H^{W'}_n}{m-1}{m}\left(\frac{1}{\slice{H^{W'}_P}{n}{m-1}-y}\right)^{1/2}\slice{\phi_P}{n}{m-1}\right\|\left\vert \frac{1}{\slice{E'_P}{n}{m-1}-y} \right\vert^{1/2}\\
 &\leq \frac{C}{\gamma^{\frac{1}{2}}}\frac{1}{\gamma^{\frac{m}{2}}}\sup_{y\in\widehat\Delta_m}\left\|\left(\frac{1}{\slice{H^{W'}_P}{n}{m-1}-y}\right)^{1/2}\slice{\Delta\widehat H^{W'}_n}{m-1}{m}\left(\frac{1}{\slice{H^{W'}_P}{n}{m-1}-y}\right)^{1/2}\slice{\phi_P}{n}{m-1}\right\|,
\end{align*}
where we have used the estimates in (\ref{eqn:ir diff 1})-(\ref{eqn:sandwich}) for $y\in\widehat\Delta_m$, and, using the identity in (\ref{gamma-identity}) 
\begin{align}\label{gamma-term}
 &\left\|\left(\frac{1}{\slice{H^{W'}_P}{n}{m-1}-z}\right)^{1/2}\slice{\Gamma^{(i)}_P}{n}{m-1}\left(\frac{1}{\slice{H^{W'}_P}{n}{m-1}-y}\right)^{1/2}\right\|_{\slice{\cal F}{n}{m}}\\
 &\qquad = \left\|\left(\frac{1}{\slice{H'_P}{n}{m-1}-z}\right)^{1/2}[P^f-P+\nabla \slice{E'_P}{n}{m-1}+\slice{B}{n}{0}+\slice{B^*}{n}{0}]\left(\frac{1}{\slice{H'_P}{n}{m-1}-y}\right)^{1/2}\right\|_{\slice{\cal F}{n}{m}}\nonumber \\
& \qquad \leq C\tau_m^{-1}\gamma^{-1/2}.\label{final-eq-a}
\end{align}
The inequality in (\ref{final-eq-a}) can be derived by combining the first inequality in (\ref{eqn:uv standard ineq}) with \lem{lem:ir a priori}.

Using \lem{lem:leading term}, Assumption A(iii), the estimates \eqn{eqn:triangle first term}, \eqn{eqn:norm estimate} and the constraints \eqn{eqn:constraint}
we get
\begin{align*}
 \eqn{eqn:first worm}&\leq C\left[g^4\gamma^{-2}+g^2\gamma^{-3}\left(1+\gamma^{-\frac{m-1}{2}}g^{-\delta}\right)+\gamma^{-\frac{m+2}{2}}+g^4\gamma^{m-3}+g^2\gamma^{m-3}\right]\leq \frac{C}{\gamma^{\frac{m+2}{2}}},\\
 \eqn{eqn:second worm}&\leq C g^2\gamma^{-3}(1+\gamma^{-\frac{m-1}{2}}|g|^{-\delta})\leq \frac{C}{\gamma^{\frac{m+2}{2}}},
\end{align*}
and hence,
\begin{align}\label{eqn:C6}
 \eqn{eqn:anyway small term}&\leq \frac{C_{6}}{\gamma^{\frac{m+2}{2}}}.
\end{align}

Finally, we collect inequalities \eqn{eqn:C7}, \eqn{eqn:C6} and make use of assumption A(iii) to derive
\begin{align*}
 &|g|^{\delta}\left|\braket{\slice{\Gamma^{(i)}_P}{n}{m}\slice{\phi_P}{n}{m},\frac{1}{\slice{H^{W'}_P}{n}{m}-z}\slice{\Gamma^{(i)}_P}{n}{m}\slice{\phi_P}{n}{m}}\right|\leq  C_7\gamma^{-\frac{m-1}{2}} + |g|^{\delta}\frac{C_{6}}{\gamma^{\frac{m+2}{2}}}\leq \gamma^{-\frac{m}{2}}
\end{align*}
for $\gamma$ and $|g|$ sufficiently small and fulfilling the constraints in (\ref{eqn:constraint}). This proves claim C(iii).
\end{proof}

We shall now provide the lemmas we have used.

\begin{lemma}\label{lemma:pi terms}
Let $|g|$ be sufficiently small. For $n,m\in\bb N$ the following expectation values are uniformly bounded:
\begin{align*}
 \left|\braket{\slice{\phi_P}{n}{m},\slice{\Pi_P}{n}{m} \slice{\phi_P}{n}{m}}\right|,\left|\braket{\slice{\widetilde\phi_P}{n}{m},\slice{\widetilde\Pi_P}{n}{m} \slice{\widetilde\phi_P}{n}{m}}\right|\leq C,
\end{align*}
\end{lemma}

\begin{proof}
We only prove the bound for the first term. The second can be bounded analogously. Let $n,m\in\bb N$. 
By definition of the transformations $W_m$ and using the fact that the vectors
\[
\slice{\Psi'_P}{n}{m},\quad W_m(\nabla \slice{E'_P}{n}{m})^*\slice{\phi_P}{n}{m}, \quad W_m(\nabla \slice{E'_P}{n}{m-1})^*\slice{\widetilde\phi_P}{n}{m}
\]
are parallel and their norm is less than one, we have
\begin{align*}
 &\left|\braket{\slice{\phi_P}{n}{m},\slice{\Pi_P}{n}{m} \slice{\phi_P}{n}{m}}\right|\\
 &\qquad\leq C\Big|\braket{\frac{\slice{\Psi_P'}{n}{m}}{\|\slice{\Psi'_P}{n}{m}\|},\left[P^f+\slice{B}{n}{0}+
 \slice{B^*}{n}{0}-C^{(k,n)}_{P,m}\right] \frac{\slice{\Psi_P'}{n}{m}}{\|\slice{\Psi'_P}{n}{m}\|}}\Big| \leq C[|P|+|\nabla\slice{E'_P}{n}{m}|+|C^{(k,n)}_{P,m}|]. 
\end{align*}
where the last inequality holds by \lem{lem:grad E}.
\end{proof}

\begin{lemma}\label{lem:leading term} Let $|g|,\zeta,\gamma$ be sufficiently small.  Furthermore, let $n\in\bb N$, $m\geq 2$ and $z\in\widehat\Delta_{m}$. Then
\begin{align}\nonumber
 &\left\|\left(\frac{1}{\slice{H^{W'}_P}{n}{m-1}-z}\right)^{1/2}\slice{\Delta \widehat H^{W'}_n}{m-1}{m}\left(\frac{1}{\slice{H^{W'}_P}{n}{m-1}-z}\right)^{1/2}\slice{\phi_P}{n}{m-1}\right\|\\
 &\qquad\leq |g|C\gamma^{\frac{m-2}{2}}\left(1+
 \sum_{i=1}^3\left|\braket{\slice{\Gamma^{(i)}_P}{n}{m-1}\slice{\phi_P}{n}{m-1},\frac{1}{\slice{H^{W'}_P}{n}{m-1}-z}\slice{\Gamma^{(i)}_P}{n}{m-1}\slice{\phi_P}{n}{m-1}}\right|^{\frac{1}{2}}\right)\label{eqn:delta H estimate}
\end{align}
holds true, where $\slice{\Delta \widehat H^{W'}_n}{m-1}{m}$ is defined in \eqn{eqn:tilde h}.
\end{lemma}

\begin{proof}
Recall the expression for $\slice{\Delta \widehat H_n^{W'}}{m-1}{m}$ given in \eqn{eqn:quad term}-\eqn{eqn:mixed term 2}
. With the usual estimates one can show that
\begin{align}\label{eqn:first couple of terms}
 &\left\|\left(\frac{1}{\slice{H^{W'}_P}{n}{m-1}-z}\right)^{1/2}\left(\eqn{eqn:quad term} + \eqn{eqn:quad term 2}\right)\left(\frac{1}{\slice{H^{W'}_P}{n}{m-1}-z}\right)^{1/2}\slice{\phi_P}{n}{m-1}\right\|\leq |g|C\gamma^{\frac{m-1}{2}}.
\end{align}

Next, we control the first term in \eqn{eqn:mixed term 2}. First, observe that
\begin{align}
 &\left\|\left(\frac{1}{\slice{H^{W'}_P}{n}{m-1}-z}\right)^{1/2}(\widetilde A_{P,m}^{(n)}-A_{P,m-1}^{(n)})\cdot\slice{\Gamma_P}{n}{m-1}\left(\frac{1}{\slice{H^{W'}_P}{n}{m-1}-z}\right)^{1/2}\slice{\phi_P}{n}{m-1}\right\|^2\nonumber\\
 &=\frac{1}{|\slice{E'_P}{n}{m-1}-z|}\braket{(\widetilde A_{P,m}^{(n)}-A_{P,m-1}^{(n)})\cdot\slice{\Gamma_P}{n}{m-1}\slice{\phi_P}{n}{m-1},\left|\frac{1}{\slice{H^{W'}_P}{n}{m-1}-z}\right|(\widetilde A_{P,m}^{(n)}-A_{P,m-1}^{(n)})\cdot\slice{\Gamma_P}{n}{m-1}\slice{\phi_P}{n}{m-1}}.\label{eqn:braket with norm}
\end{align}
Second, we recall that $\widetilde A_{P,m}^{(n)}-A_{P,m-1}^{(n)}$ contains boson creation operators restricted to the range $(\tau_{m},\tau_{m-1}]$ in momentum space. Therefore, 
\[
  \braket{\slice{\phi_P}{n}{m-1},\widetilde A_{P,m}^{(n)}-A_{P,m-1}^{(n)}\cdot\slice{\Gamma_P}{n}{m-1}\slice{\phi_P}{n}{m-1}}=0,
\]
which implies
\begin{align}
 &\eqn{eqn:braket with norm}\leq \frac{C}{\gamma^m}\left|\braket{(\widetilde A_{P,m}^{(n)}-A_{P,m-1}^{(n)})\cdot\slice{\Gamma_P}{n}{m-1}\slice{\phi_P}{n}{m-1},\frac{1}{\slice{H^{W'}_P}{n}{m-1}-z}(\widetilde A_{P,m}^{(n)}-A_{P,m-1}^{(n)})\cdot\slice{\Gamma_P}{n}{m-1}\slice{\phi_P}{n}{m-1}}\right|\label{eqn:before pull-through}
\end{align}
by using the spectral theorem and the gap estimate for $\slice{H^{W'}_P}{n}{m-1}\restrict{\slice{\cal F}{n}{m}}$.
Note further that
\begin{align*}
 (\widetilde A_{P,m}^{(n)}-A_{P,m-1}^{(n)})\cdot\slice{\Gamma_P}{n}{m-1}\slice{\phi_P}{n}{m-1}=\int dk\;(\alpha_m(\nabla \slice{E'_P}{n}{m-1})-\alpha_{m-1}(\nabla \slice{E'_P}{n}{m-1}))b^*(k)k\cdot\slice{\Gamma_P}{n}{m-1}\slice{\phi_P}{n}{m-1}.
\end{align*}
Using the pull-through formula we get
\begin{align*}
\frac{1}{\slice{H^{W'}_P}{n}{m-1}-z}b^*(k) = b^*(k)\frac{1}{\slice{H^{W'}_P}{n}{m-1}+\frac{1}{2}k^2+k\cdot\slice{\Gamma_P}{n}{m-1}+|k|-\nabla \slice{E'_P}{n}{m-1}\cdot k-z}
\end{align*}
so that we can rewrite the right-hand side of \eqn{eqn:before pull-through} as follows:
\begin{align}
\label{eqn:after pull-through}
  \eqn{eqn:before pull-through}&=\frac{C}{\gamma^m}\int dk\;\big[\alpha_m(\nabla \slice{E'_P}{n}{m-1})-\alpha_{m-1}(\nabla \slice{E'_P}{n}{m-1})\big]^2\times \nonumber\\
 &\quad\times\braket{k\cdot\slice{\Gamma_P}{n}{m-1}\slice{\phi_P}{n}{m-1},\frac{1}{\slice{H^{W'}_P}{n}{m-1}+\frac{1}{2}k^2+k\cdot\slice{\Gamma_P}{n}{m-1}+|k|-\nabla \slice{E'_P}{n}{m-1}\cdot k-z}k\cdot\slice{\Gamma_P}{n}{m-1}\slice{\phi_P}{n}{m-1}}.
\end{align}
In order to expand the resolvent in \eqn{eqn:after pull-through} in terms of $k\cdot\slice{\Gamma_P}{n}{m-1}$ we have to provide the bound
\begin{align}\label{eqn:exp cond}
 \left\|\left(\frac{1}{\slice{H^{W'}_P}{n}{m-1}+f_{P,m-1}(k)-z}\right)^{1/2}k\cdot\slice{\Gamma_P}{n}{m-1}\left(\frac{1}{\slice{H^{W'}_P}{n}{m-1}+f_{P,m-1}(k)-z}\right)^{1/2}\right\|_{\slice{\cal F}{n}{m}}<1
\end{align}
for $\tau_m<|k|\leq \tau_{m-1}$ and $z\in\widehat\Delta_m$, where we have defined
\begin{align*}
 f_{P,m-1}(k):=\frac{1}{2}k^2+|k|(1-\nabla \slice{E'_P}{n}{m-1}\cdot \widehat k).
\end{align*}
Recall that
\begin{align*}
 \slice{\Gamma_P}{n}{m-1}=P^f+A_{P,m-1}^{(n)}+\slice{B}{n}{0}+\slice{B^*}{n}{0}-\braket{\slice{\Pi_P}{n}{m-1}}_{\slice{\phi_{P}}{n}{m-1}}.
\end{align*}
The necessary estimates are:
\begin{enumerate}
 
 \item For $|g|$ sufficiently small,  the lower bound
\begin{align}
 f_{P,m-1}(k)-|\slice{E'_P}{n}{m-1}-z|>|k|\left(1-\nabla \slice{E'_P}{n}{m-1}\cdot \widehat k-\frac{1}{2}\zeta-g^2\gamma^{-1}C\right)>0\label{eqn:fest}
\end{align}
holds because $z$ belongs to the shifted contour $\widehat\Delta_m$ so that
\begin{align*}
 |\slice{E'_P}{n}{m-1}-z|\leq \frac{1}{2}\zeta\tau_m+g^2 C\tau_{m-1}.
\end{align*}
The inequality in \eqn{eqn:fest} implies
\begin{align*}
 \left\|\left(\frac{1}{\slice{H'_P}{n}{m-1}+f_{P,m-1}(k)-z}\right)^{1/2}\right\|^2_{\slice{\cal F}{n}{m-1}} &\leq \frac{1}{|k|\left(1-\nabla \slice{E'_P}{n}{m-1}\cdot \widehat k-\frac{1}{2}\zeta-g^2\gamma^{-1}C\right)}.
\end{align*}
  
 \item By the unitarity of $W_{m-1}(\nabla \slice{E'_P}{n}{m-1})$ and using $[\slice{B}{n}{0},W_{m-1}(\nabla \slice{E'_P}{n}{m-1})]=0$ as well as the standard inequalities \eqn{eqn:standard ineq}, we have
\begin{align*}
 &\left\|k\cdot \slice{B}{n}{0}\left(\frac{1}{\slice{H^{W'}_P}{n}{m-1}+f_{P,m-1}(k)-z}\right)^{1/2}\right\|_{\slice{\cal F}{n}{m-1}}\leq|g|\;|k|\;C\left\|H_{P,0}^{1/2}\left(\frac{1}{\slice{H'_P}{n}{m-1}+f_{P,m-1}(k)-z}\right)^{1/2}\right\|_{\slice{\cal F}{n}{m-1}}.
\end{align*}
 
\item By definition of the transformation $W_{m-1}(\nabla \slice{E'_P}{n}{m-1})$ and the transformation formulae \eqn{eqn:transformation formulae},
\begin{align*}
 W_{m-1}(\nabla \slice{E'_P}{n}{m-1})(P-P^f)W_{m-1}(\nabla \slice{E'_P}{n}{m-1})^*=P-P^f-A_{P,m-1}^{(n)}-C^{(k,n)}_{P,m-1}
\end{align*}
holds on $D(H_{P,0})$. Hence, we have the bound
\begin{align*}
&\left\|k\cdot (P^f+A_{P,m-1}^{(n)})\left(\frac{1}{\slice{H^{W'}_P}{n}{m-1}+f_{P,m-1}(k)-z}\right)^{1/2}\right\|_{\slice{\cal F}{n}{m-1}}\\
 &\qquad\qquad\qquad\leq|k|\;\left\|(P-P^f)\left(\frac{1}{\slice{H'_P}{n}{m-1}+f_{P,m-1}(k)-z}\right)^{1/2}\right\|_{\slice{\cal F}{n}{m-1}} \\
  &\qquad\qquad\qquad\quad +|k|(|P|+g^2 C)\;\left\|\left(\frac{1}{\slice{H'_P}{n}{m-1}+f_{P,m-1}(k)-z}\right)^{1/2}\right\|_{\slice{\cal F}{n}{m-1}}\\
&\qquad\qquad\qquad\leq|k|\sqrt 2\left\|H_{P,0}^{1/2}\left(\frac{1}{\slice{H'_P}{n}{m-1}+f_{P,m-1}(k)-z}\right)^{1/2}\right\|_{\slice{\cal F}{n}{m}}\\
 &\qquad\qquad\qquad\quad +|k|(|P|+g^2 C)\;\left\|\left(\frac{1}{\slice{H'_P}{n}{m-1}+f_{P,m-1}(k)-z}\right)^{1/2}\right\|_{\slice{\cal F}{n}{m-1}}. 
\end{align*}

 \item Using the a priori estimate \eqn{eqn:a priori} in \lem{lem:a priori} one derives 
\begin{multline*}
\left\|H_{P,0}^{1/2}\left(\frac{1}{\slice{H'_P}{n}{m-1}+f_{P,m-1}(k)-z}\right)^{1/2}\right\|_{\slice{\cal F}{n}{m-1}}\\
\leq \frac{1}{\sqrt{1-|g|\namer{a}}}\left(\left\|(\slice{H'_P}{n}{m-1})^{1/2}\;\left(\frac{1}{\slice{H'_{P}}{n}{m-1}+f_{P,m-1}(k)-z}\right)^{1/2}\right\|^2_{\slice{\cal F}{n}{m-1}} \right.\\
+\left.|g|\namer{b}\left\|\left(\frac{1}{\slice{H'_{P}}{n}{m-1}+f_{P,m-1}(k)-z}\right)^{1/2}\right\|^2_{\slice{\cal F}{n}{m-1}}\right)^{1/2}.
\end{multline*}
 \end{enumerate}
Collecting these estimates, we find:
\begin{align}\label{eqn:key term}
 &\left\|\left(\frac{1}{\slice{H^{W'}_P}{n}{m-1}+f_{P,m-1}(k)-z}\right)^{1/2}k\cdot\slice{\Gamma_P}{n}{m-1}\left(\frac{1}{\slice{H^{W'}_P}{n}{m-1}+f_{P,m-1}(k)-z}\right)^{1/2}\right\|_{\slice{\cal F}{n}{m-1}}\\
 &\leq|k|\;\left\|\left(\frac{1}{\slice{H^{W'}_P}{n}{m-1}+f_{P,m-1}(k)-z}\right)^{1/2}\right\|_{\slice{\cal F}{n}{m-1}}\times\\
 &\quad\times\left[\frac{\sqrt 2+|g|C}{\sqrt{1-|g|\namer{a}}}\;
 \left(\left\|(\slice{H'_P}{n}{m-1})^{1/2}\;\left(\frac{1}{\slice{H'_{P}}{n}{m-1}+f_{P,m-1}(k)-z}\right)^{1/2}\right\|^2_{\slice{\cal F}{n}{m-1}}+\nonumber \right.\right.\\
 &\left.\left.\quad+|g|\namer{b}\left\|\left(\frac{1}{\slice{H'_{P}}{n}{m-1}+f_{P,m-1}(k)-z}\right)^{1/2}\right\|^2_{\slice{\cal F}{n}{m-1}}\right)^{1/2}+(|P|+g^2 C)\;\left\|\left(\frac{1}{\slice{H_P}{n}{m-1}+f_{P,m-1}(k)-z}\right)^{1/2}\right\|_{\slice{\cal F}{n}{m-1}}\right]\nonumber.
\end{align}
Note that
\begin{align*}
 &\left\|(\slice{H'_P}{n}{m-1})^{1/2}\left(\frac{1}{\slice{H'_P}{n}{m-1}+f_{P,m-1}(k)-z}\right)^{1/2}\right\|_{\slice{\cal F}{n}{m-1}}\leq \left(1+\frac{|\slice{E'_P}{n}{m-1}|}{f_{P,m-1}(k)-|\slice{E'_P}{n}{m-1}-z|}\right)^{1/2}.
\end{align*}
Finally we obtain
\begin{align*}
 \eqn{eqn:key term}&\leq \frac{1}{\left(1-\nabla \slice{E'_P}{n}{m-1}\cdot \widehat k-\frac{1}{2}\zeta-g^{2}\gamma^{-1}C\right)}\times\\
 &\quad\times\Bigg[ 
   \frac{\sqrt 2+|g|C}{\sqrt{1-|g|\namer{a}}}\;\bigg(
     |\slice{E'_P}{n}{m-1}| +\tau_{m-1}\bigg(1-\nabla \slice{E'_P}{n}{m-1}\cdot \widehat k-\frac{1}{2}\zeta-Cg^2\gamma^{-1}\bigg)+g\namer{b}
   \bigg)^{1/2}+(|P|+g^2 C)
 \Bigg]
\end{align*}
so that 
\begin{align*}
 \lim\sup_{|g|,\gamma,\zeta\to 0}\eqn{eqn:key term}\leq \frac{2P_{\mathrm{max}}}{1-P_{\mathrm{max}}}<\frac{2}{3}
\end{align*}
for $P_{\mathrm{max}}<\frac{1}{4}$. By continuity, inequality \eqn{eqn:exp cond} holds for $g,\zeta,\gamma$ in a neighborhood of zero.\\

Going back to equation \eqn{eqn:after pull-through} we can proceed with the expansion (in $k\cdot\slice{\Gamma_P}{n}{m-1}$) of the resolvent:
\begin{align}
\eqn{eqn:before pull-through} &\leq Cg^{2}\gamma^{m-2}\sup_{\tau_m\leq |k|\leq \tau_{m-1}}\sum_{i,l=1}^3\Bigg\langle\left[\left(\frac{1}{\slice{H^{W'}_P}{n}{m-1}+f_{P,m-1}(k)-z}\right)^{1/2}\right]^*\slice{\Gamma_P^{(i)}}{n}{m-1}\slice{\phi_P}{n}{m-1},\\
 &\quad\sum_{j=0}^\infty\left[\left(\frac{1}{\slice{H^{W'}_P}{n}{m-1}+f_{P,m-1}(k)-z}\right)^{1/2}k\cdot\slice{\Gamma_P}{n}{m-1}
 \left(\frac{1}{\slice{H^{W'}_P}{n}{m-1}+f_{P,m-1}(k)-z}\right)^{1/2}\right]^j\times\nonumber\\
 &\quad\times
 \left(\frac{1}{\slice{H^{W'}_P}{n}{m-1}+f_{P,m-1}(k)-z}\right)^{1/2}\slice{\Gamma_P^{(l)}}{n}{m-1}\slice{\phi_P}{n}{m-1}\Bigg\rangle\nonumber\\
 &\leq Cg^{2}\gamma^{m-2}\sum_{i=1}^3\sup_{\tau_m\leq |k|\leq \tau_{m-1}}\left\|\left(\frac{1}{\slice{H^{W'}_P}{n}{m-1}+f_{P,m-1}(k)-z}\right)^{1/2}\slice{\Gamma_P^{(i)}}{n}{m-1}\slice{\phi_P}{n}{m-1}\right\|^2.\label{eqn:last pt}
\end{align}
Since $f_{P,m-1}(k)\geq 0$ and because of the property $\braket{\slice{\phi_P}{n}{m-1},\slice{\Gamma_P}{n}{m-1}\slice{\phi_P}{n}{m-1}}=0$ it follows that
\begin{equation*}
\left\|\left(\frac{1}{\slice{H^{W'}_P}{n}{m-1}+f_{P,m-1}(k)-z}\right)^{1/2}\slice{\Gamma_P^{(i)}}{n}{m-1}\slice{\phi_P}{n}{m-1}\right\|^2 \leq C\left|\braket{\slice{\Gamma_P^{(i)}}{n}{m-1}\slice{\phi_P}{n}{m-1},\frac{1}{\slice{H^{W'}_P}{n}{m-1}-z}\slice{\Gamma_P^{(i)}}{n}{m-1}\slice{\phi_P}{n}{m-1}}\right|.
\end{equation*}
Combining the estimates in \eqn{eqn:last pt} and \eqn{eqn:first couple of terms} yields the claim of the lemma.
\end{proof}

\begin{lemma}\label{lem:ir trafo diff}
 For all $n,m\in\bb N$ and $Q,Q'\in\bb R^3$ with $|Q|,|Q'|\leq 1$ the estimate
\begin{align*}
 \|[W_m(Q)-W_m(Q')]\;\slice{\Psi'_P}{n}{m}\|\leq |g| C|Q-Q'||\ln \tau_m|
\end{align*}
 holds.
\end{lemma}
\begin{proof}
 The Bogolyubov transformations $W_m$ defined in \eqn{eqn:W trafo} can be explicitly written as
\begin{equation*}
W_m(Q) = \exp\left(\int dk\, \alpha_m(Q,k)\big(b(k) - b^*(k)\big)\right)\,,
\end{equation*}
so that
\begin{equation}
\|[W_m(Q)-W_m(Q')]\;\slice{\Psi'_P}{n}{m}\|
\leq \left\|\int dk\;[\alpha_m(Q,k)-\alpha_m(Q',k)](b(k)-b^*(k))\;\slice{\Psi'_P}{n}{m}\right\|\label{eqn:ir trafo 1}
\end{equation}
In order to estimate this term we employ:
\begin{enumerate}
 \item The identity \eqn{eqn:froehlich} in  \cite[Equation (1.26)]{fraehlich_infrared_1973}
that relies on the bound $\slice{E'_{P-k}}{n}{m}-\slice{E'_{P}}{n}{m}\geq -\namer{c:cp}|k|$, $|P|\leq P_{\mathrm{max}}$, from \cor{cor:ground state energies}(iii).
\item By definition of $\alpha_m$ it holds
\begin{align*}
 \int dk \Big|\alpha_m(Q,k)-\alpha_m(Q',k)\Big|\frac{1}{|k|^{3/2}}&\leq |g|C|Q-Q'|\;|\ln\kappa-\ln\tau_m|.
\end{align*}
 \item $\|\slice{\Psi'_P}{n}{m}\|\leq 1$
\end{enumerate}
With these estimates, the claim is proven.
\end{proof}

\begin{lemma}\label{lem:grad E diff}
Let $|g|$ be sufficiently small. For $n,m\in\bb N$ the following estimate holds:
\begin{multline*}
|\nabla \slice{E'_P}{n}{m}-\nabla \slice{E'_P}{n}{m-1}|\\
\leq g^2C\tau_{m-1}^{1/2}+C\sup_{z\in\widehat\Delta_m}\left\|\left(\frac{1}{\slice{H^{W'}_P}{n}{m-1}-z}\right)^{1/2}\slice{\Delta \widehat H^{W'}_n}{m-1}{m}\left(\frac{1}{\slice{H^{W'}_P}{n}{m-1}-z}\right)^{1/2}\slice{\phi_P}{n}{m-1}\right\|
+C\frac{\|\slice{\phi_P}{n}{m-1}-\slice{\widetilde\phi_P}{n}{m}\|}{\|\slice{\phi_P}{n}{m-1}\|^2\|\slice{\widetilde\phi_P}{n}{m}\|^2}.
\end{multline*}
\end{lemma}

\begin{proof}
 Let $n,m\in\bb N$. Using \lem{lem:grad E} we have
\begin{align*}
 \nabla \slice{E'_P}{n}{m}-\nabla \slice{E'_P}{n}{m-1}=
 \braket{P^f+\slice{B}{n}{0}+\slice{B^*}{n}{0}}_{\slice{\Psi'_{P}}{n}{m-1}}-
 \braket{P^f+\slice{B}{n}{0}+\slice{B^*}{n}{0}}_{\slice{\Psi'_{P}}{n}{m}}
\end{align*}
which by unitarity of the transformation $W_{m-1}(\nabla\slice{E'_P}{n}{m-1})$ and $W_m(\nabla\slice{E'_P}{n}{m-1})$ can be rewritten as
\begin{align*}
 \nabla \slice{E'_P}{n}{m}-\nabla \slice{E'_P}{n}{m-1}=
 \braket{\slice{\Pi_P}{n}{m-1}}_{\slice{\phi_P}{n}{m-1}}-
 \braket{\slice{\widetilde\Pi_{P}}{n}{m}}_{\slice{\widetilde\phi_P}{n}{m}}+\widetilde C^{(k,n)}_{P,m}-C^{(k,n)}_{P,m-1}.
\end{align*}
We have already noted that $|\widetilde C^{(k,n)}_{P,m}-C^{(k,n)}_{P,m-1}|\leq g^2C\tau_{m-1}$. Moreover, we observe
\begin{align}
 &\left|\braket{\slice{\Pi_P}{n}{m-1}}_{\slice{\phi_P}{n}{m-1}}-
 \braket{{\widetilde\Pi}_{P,m}}_{\slice{\widetilde\phi_P}{n}{m}}\right|=
 \left|\frac{\braket{\slice{\phi_P}{n}{m-1},\slice{\Pi_P}{n}{m-1}\slice{\phi_P}{n}{m-1}}}{\|\slice{\phi_P}{n}{m-1}\|^2}-
 \frac{\braket{\slice{\widetilde\phi_P}{n}{m},\slice{\widetilde\Pi_{P}}{n}{m}\slice{\widetilde\phi_P}{n}{m}}}{\|\slice{\widetilde\phi_P}{n}{m}\|^2}\right|\nonumber\\
 &\leq \|\slice{\phi_P}{n}{m-1}\|^{-2} \left|\braket{\slice{\phi_P}{n}{m-1},\slice{\Pi_P}{n}{m-1}\slice{\phi_P}{n}{m-1}}-\braket{\slice{\widetilde\phi_P}{n}{m},\slice{\widetilde\Pi_{P}}{n}{m}\slice{\widetilde\phi_P}{n}{m}}
 \right|+\label{eqn:pi term 1}\\
 &\quad + \left|\braket{\slice{\widetilde\phi_P}{n}{m},\slice{\widetilde\Pi_{P}}{n}{m}\slice{\widetilde\phi_P}{n}{m}}\right| \left|\|\slice{\phi_P}{n}{m-1}\|^{-2}-\|\slice{\widetilde\phi_P}{n}{m}\|^{-2}\right|\label{eqn:pi term 2}.
\end{align}

We know that the norms $\|\slice{\phi_P}{n}{m-1}\|$ and $\|\slice{\widetilde\phi_P}{n}{m}\|$ are by construction smaller than one and non-zero. Using \lem{lemma:pi terms} we find
\begin{align*}
 \eqn{eqn:pi term 2}\leq C\frac{\|\slice{\phi_P}{n}{m-1}-\slice{\widetilde\phi_P}{n}{m}\|}{\|\slice{\phi_P}{n}{m-1}\|^2\|\slice{\widetilde\phi_P}{n}{m}\|^2}.
\end{align*}
In order to bound the term \eqn{eqn:pi term 1} we use
\begin{align}
\|\slice{\phi_P}{n}{m-1}\|^{2}\eqn{eqn:pi term 1} &=\Bigg|\braket{\left(\slice{\phi_P}{n}{m-1}-\slice{\widetilde\phi_P}{n}{m}\right),\slice{\Pi_P}{n}{m-1}\slice{\phi_P}{n}{m-1}}+\label{eqn:tele term 1}\\
 &\quad+\braket{\slice{\widetilde\phi_P}{n}{m},\left[\slice{\Pi_P}{n}{m-1}-\slice{\widetilde\Pi_{P}}{n}{m}\right]\slice{\phi_P}{n}{m-1}}+\label{eqn:tele term 2}\\
 &\quad+\braket{\slice{\widetilde\phi_P}{n}{m},\slice{\widetilde\Pi_{P}}{n}{m}\left(\slice{\phi_P}{n}{m-1}-\slice{\widetilde\phi_P}{n}{m}\right)}\Bigg|\label{eqn:tele term 3}.
\end{align}
The term \eqn{eqn:tele term 2} is bounded by
\begin{align*}
 |\eqn{eqn:tele term 2}|\leq\left|\braket{\slice{\widetilde\phi_P}{n}{m},\left[\widetilde A_{P,m}^{(n)}-A_{P,m-1}^{(n)}\right]\slice{\phi_P}{n}{m-1}}\right|\leq |g|C\tau_{m-1}^{1/2}
\end{align*}
because by the standard inequalities \eqn{eqn:standard ineq}
\begin{multline*}
\left\|\int dk\; k[\alpha_m(\nabla \slice{E'_P}{n}{m-1},k)-\alpha_{m-1}(\nabla \slice{E'_P}{n}{m-1},k)]b(k)\slice{\phi_P}{n}{m-1}\right\|\\
\leq C\left(\int dk\; \left|\frac{k[\alpha_m(\nabla \slice{E'_P}{n}{m-1},k)-\alpha_{m-1}(\nabla \slice{E'_P}{n}{m-1},k)}{|k|^{1/2}}\right|^2\right)^{1/2}\left\|(\slice{H^f}{m-1}{m})\left(\frac{1}{\slice{H^{W'}_P}{n}{m-1}-i}\right)^{1/2}\slice{\phi_P}{n}{m-1}\right\|
\leq |g|C\tau_{m-1}^{1/2}.
\end{multline*}

Terms \eqn{eqn:tele term 1} and \eqn{eqn:tele term 3} can be treated in the same way, and we only demonstrate the bound on the former. Using analytic perturbation theory we get
\begin{align}
 &\left|\braket{\left(\slice{\phi_P}{n}{m-1}-\slice{\widetilde\phi_P}{n}{m}\right),\slice{\Pi_P}{n}{m-1}\slice{\phi_P}{n}{m-1}}\right|\label{eqn:ir phi diff 1}\\
 &\qquad\leq C\tau_m \sup_{z\in\widehat\Delta_m}\sum_{j=1}^\infty \Bigg|\Bigg\langle\left[\left(\frac{1}{\slice{H^{W'}_P}{n}{m-1}-z}\right)^{1/2}\slice{\Delta \widehat H^{W'}_n}{m-1}{m}\left(\frac{1}{\slice{H^{W'}_P}{n}{m-1}-z}\right)^{1/2}\right]^j\slice{\phi_P}{n}{m-1},\nonumber\\
 &\qquad\hskip3cm,\left[\left(\frac{1}{\slice{H^{W'}_P}{n}{m-1}-z}\right)^{1/2}\right]^*\slice{\Pi_P}{n}{m-1}\left(\frac{1}{\slice{H^{W'}_P}{n}{m-1}-z}\right)^{1/2}\slice{\phi_P}{n}{m-1}\Bigg\rangle\Bigg|\nonumber\\
 &\qquad\leq C\tau_m \sup_{z\in\widehat\Delta_m}\left\|\left(\frac{1}{\slice{H^{W'}_P}{n}{m-1}-z}\right)^{1/2}\slice{\Delta \widehat H^{W'}_n}{m-1}{m}\left(\frac{1}{\slice{H^{W'}_P}{n}{m-1}-z}\right)^{1/2}\slice{\phi_P}{n}{m-1}\right\|\times\nonumber\\
 &\qquad\phantom{\leq C\tau_m \sup_{z\in\widehat\Delta_m}}\times\left\|\left[\left(\frac{1}{\slice{H^{W'}_P}{n}{m-1}-z}\right)^{1/2}\right]^*\slice{\Pi_P}{n}{m-1}\left(\frac{1}{\slice{H^{W'}_P}{n}{m-1}-z}\right)^{1/2}\right\|_{\slice{\cal F}{n}{m}}\label{eqn:ir phi diff 3}.
\end{align}
The term in  \eqn{eqn:ir phi diff 3} can be controlled  similarly to (\ref{gamma-term}) in the ultraviolet regime
so that we finally have
\begin{equation}
\left\|\left[\left(\frac{1}{\slice{H^{W'}_P}{n}{m-1}-z}\right)^{1/2}\right]^*\slice{\Pi_P}{n}{m-1}\left(\frac{1}{\slice{H^{W'}_P}{n}{m-1}-z}\right)^{1/2}\right\|_{\slice{\cal F}{n}{m}} \leq C\tau_m^{-1}.
\end{equation}
Combining these results, we obtain the estimate
\begin{align*}
 \left|\braket{\left(\slice{\phi_P}{n}{m-1}-\slice{\widetilde\phi_P}{n}{m}\right),\slice{\Pi_P}{n}{m-1}\slice{\phi_P}{n}{m-1}}\right| \leq C\sup_{z\in\widehat\Delta_m}\left\|\left(\frac{1}{\slice{H^{W'}_P}{n}{m-1}-z}\right)^{1/2}\slice{\Delta \widehat H^{W'}_n}{m-1}{m}\left(\frac{1}{\slice{H^{W'}_P}{n}{m-1}-z}\right)^{1/2}\slice{\phi_P}{n}{m-1}\right\|,
\end{align*}
which concludes the proof.
\end{proof}

\section{Ground States of the Transformed Hamiltonians $\slice{H^{W'}_P}{\infty}{\infty}$}\label{sec:Winfinity}
 In this section, we finally remove both the UV and the IR cut-off ($\sigma_n$ and $\tau_m$, respectively). In our study of the removal of the IR cut-off in \sct{sec:W ground states} we have proven that
\begin{equation*}
\|\slice{\phi_{P}}{n}{m}-\slice{\phi_{P}}{n}{m-1}\|\leq (m+1)\gamma^{\frac{m}{4}}
\end{equation*}
holds for any $n\in\bb N$. We shall now provide the analogous bound
\begin{equation}\label{eqn:phi_diff}
 \|\slice{\phi_P}{n}{m}-\slice{\phi_P}{n-1}{m}\|\leq C m K^{3m+1}|\ln\gamma|^{m+1}\left(\frac{n}{\beta^n\gamma^m}\right)^{1/2}
\end{equation}
as the UV cut-off is shifted from $\sigma_{n-1}$ to $\sigma_n$. The constant $K\geq 1$ will be introduced in \thm{thm:IRUV induction}. The latter bound, derived in \cor{cor:state conv}, holds for any IR cut-off $\tau_m$ and uses a particular scaling $\mathbb{N}\ni n:=n(m) > \alpha m$ for
\begin{equation}
\label{eq:alpha}
\alpha := \frac{-\ln|\gamma|}{\ln\beta}\geq 1.
\end{equation}
These two estimates will enable us to prove the second main result \thm{thm:ir} at the end of this section.

\begin{remark}
In this section  we implicitly assume  the constraints  $|P|<P_{\mathrm{max}}$ and $1<\kappa<2$. Furthermore, $g, \beta$, and  $\gamma$ are such that all  the results of Sections \ref{sec:main proof}, \ref{sec:ground states with ir cutoff}
, and \ref{sec:W ground states} hold true. 
 \end{remark}

 In order to control the norm difference $\|\slice{\phi_{P}}{n}{m}-\slice{\phi_{P}}{n-1}{m}\|$ we notice  that for $m\geq 1$ the vectors $\slice{\phi_P}{n}{m}$ can be rewritten in the following way
\begin{align*}
 \slice{\phi_P}{n}{m} = W_m(\nabla\slice{E'_P}{n}{m}) \; \slice{\cal Q'_P}{n}{m} \slice{W}{m-1}{m}(\nabla\slice{E'_P}{n}{m-1})^* \cdots
  \slice{\cal Q'_P}{n}{2} \slice{W}{1}{2}(\nabla\slice{E'_P}{n}{1})^* \slice{\cal Q'_P}{n}{1}\;\slice{W}{0}{1}(\nabla\slice{E'_P}{n}{0})^* \frac{\slice{\Psi'_P}{n}{0}}{\Vert \slice{\Psi'_P}{n}{0} \Vert},
\end{align*}
where $\slice{\cal Q'_{P}}{n}{m}$ is defined in \eqn{eqn:projection} and
\begin{equation*}
\slice{W}{m'}{m}(Q)^* := W_m(Q)^*W_{m'}(Q),\qquad \slice{W}{0}{1}(Q)^{*}=W_{1}(Q).
\end{equation*}
The following definition will be convenient.

\begin{definition}\label{def:etas}
For $n\in\bb N$ and $m\geq1$, we define
\begin{align}\label{eqn:eta_m}
 \slice{\eta_P}{n}{m}:=W_m(\nabla\slice{E'_P}{n}{m})^*\slice{\phi_P}{n}{m},
\end{align}
and $\slice{\eta_P}{n}{0}: = \slice{\phi_P}{n}{0} = \slice{\Psi_P}{n}{0} / \Vert   \slice{\Psi_P}{n}{0} \Vert\,$ in the case $m=0$.
\end{definition}
Note that by construction we have the identity
\begin{equation}\label{eqn:eta_mp1}
 \slice{\eta_P}{n}{m+1}=\slice{\cal Q'_P}{n}{m+1} \slice{W}{m}{m+1}(\nabla\slice{E'_P}{n}{m})^* \slice{\eta_P}{n}{m},
\end{equation}
and moreover, since the transformation $W_m$ is unitary and due to \thm{thm:key theorem}, the bounds
\begin{equation}\label{eqn:norm bounds}
 1\geq \|\slice{\phi_P}{n}{m}\|=\|\slice{\eta_P}{n}{m}\|\geq\frac{1}{2}
\end{equation} 
hold true for all $m,n\in\bb N$. First, we prove two a priori lemmas that can be combined to yield \thm{thm:IRUV induction}.
 
\begin{lemma} \label{lem: Delta eta}
For any $m \in \mathbb{N}$, let $\mathbb{N}\ni n > \alpha m\geq 1$. There exists a constant $K_1$ such  that for $|g|$ sufficiently small  the following estimates hold true:
\begin{equation} \label{eq: Delta eta}
\Vert \slice{\eta_P}{n}{m+1} -\slice{\eta_P}{n-1}{m+1} \Vert
\leq \Vert \slice{\eta_P}{n}{m} -\slice{\eta_P}{n-1}{m} \Vert + K_1\left[\left(\frac{n}{\beta^n \gamma^{m+1}}\right)^{1/2} + \left \vert \nabla\slice{E'_P}{n}{m} - \nabla\slice{E'_P}{n-1}{m} \right\vert \right].
\end{equation}
\end{lemma}
\begin{proof} By using \eqn{eqn:eta_m} and \eqn{eqn:eta_mp1} we get the bound
\begin{align}
\left\|\slice{\eta_P}{n}{m+1} -\slice{\eta_P}{n-1}{m+1}\right\|
&\leq \left\|\left(\slice{\cal Q'_P}{n}{m+1}-\slice{\cal Q'_P}{n-1}{m+1} \right) \slice{W}{m}{m+1}(\nabla\slice{E'_P}{n}{m})^* \slice{\eta_P}{n}{m}\right\| \label{eq:eta diff 1}\\
&\quad+ \left\|\slice{\cal Q'_P}{n-1}{m+1}\left( \slice{W}{m}{m+1}(\nabla\slice{E'_P}{n}{m})^*-\slice{W}{m}{m+1}(\nabla\slice{E'_P}{n-1}{m})^* \right) \slice{\eta_P}{n}{m} \right\|\label{eq:eta diff 2}\\
&\quad+ \left\|\slice{\cal Q'_P}{n-1}{m+1} \slice{W}{m}{m+1}(\nabla\slice{E'_P}{n-1}{m})^*\left(\slice{\eta_P}{n}{m} - \slice{\eta_P}{n-1}{m}\right)\right\|. \label{eq:eta diff 3}
\end{align}
Furthermore, the expansion
\begin{align}
\begin{split}
  \slice{\cal Q'_P}{n}{m+1}-\slice{\cal Q'_P}{n-1}{m+1}
 =
- \frac{1}{2\pi i}\oint_{\Delta_{m+1}} dz &\left\{ \left(\frac{1}{\slice{H'_P}{n-1}{m+1}-z}\right)^{1/2}\right.\times\\
 &\quad\times\sum_{j=1}^\infty
 \left[
   \left(\frac{1}{\slice{H'_P}{n-1}{m+1}-z}\right)^{1/2}\slice{\Delta H'}{n}{n-1}\left(\frac{1}{\slice{H'_P}{n-1}{m+1}-z}\right)^{1/2}
 \right]^j\times\\
 &\quad\left. \times\left(\frac{1}{\slice{H'_P}{n-1}{m+1}-z}\right)^{1/2}\right\},
\end{split}\label{eqn:uv ir res exp}
\end{align}
can be controlled by noting that
\begin{align}\label{eqn:new res est}
 \left\|\left(\frac{1}{\slice{H'_P}{n-1}{m+1}-z}\right)^{1/2}\right\|^2_{\slice{\cal F}{n}{m+1}} \leq \frac{2}{\zeta\tau_{m+1}}
\end{align}
(see \lem{lem:ir gap}), which yields
\begin{equation} \label{eq:YetAnotherRes}
 \sup_{z\in\Delta_{m+1}}
 \left\|\left(\frac{1}{\slice{H'_P}{n-1}{m+1}-z}\right)^{1/2}\slice{\Delta H'}{n}{n-1}\left(\frac{1}{\slice{H'_P}{n-1}{m+1}-z}\right)^{1/2}\right\|_{\slice{\cal F}{n}{m+1}} \leq C|g|\left(\frac{ n}{\beta^n \zeta \tau_{m+1}}\right)^{1/2}
\end{equation}
by a similar computation as for \eqn{eqn:uv int est}. Now, by the choice $n>\alpha m$ and $|g|$ sufficiently small, the right-hand side in \eqn{eq:YetAnotherRes} is strictly smaller than $1$. Hence, we get
\begin{equation*}
\Vert\slice{\cal Q'_P}{n}{m+1}-\slice{\cal Q'_P}{n-1}{m+1}\Vert\leq C|g|\left(\frac{ n}{\beta^n \zeta \tau_{m+1}}\right)^{1/2}.
\end{equation*}
Moreover, under the constraint in (\ref{eqn:constraint}) we get the bound
\begin{equation*}
\eqn{eq:eta diff 2} \leq C|g| |\ln\gamma|\left \vert \nabla\slice{E'_P}{n}{m} - \nabla\slice{E'_P}{n-1}{m} \right\vert\leq C\left \vert \nabla\slice{E'_P}{n}{m} - \nabla\slice{E'_P}{n-1}{m} \right\vert 
\end{equation*}
by a similar procedure as used in the proof of \lem{lem:ir trafo diff}.
The remaining term \eqn{eq:eta diff 3} can be estimated using the unitarity of $W_m$. This concludes the proof.
\end{proof}

\begin{lemma} \label{lem: Delta grad}
For any $m \in \mathbb{N}$, let $\mathbb{N}\ni n > \alpha m\geq 1$. There exists a constant $K_2$ such  that   for $|g|$ sufficiently small  the following estimate holds true:
\begin{equation} \label{eq: Delta grad}
\left \vert \nabla\slice{E'_P}{n}{m} - \nabla\slice{E'_P}{n-1}{m} \right\vert \leq K_2\left[\left(\frac{n}{\beta^n\gamma^m}\right)^{1/2}+\|\slice{\eta_P}{n}{m}-\slice{\eta_P}{n-1}{m}\| +  \left \vert \nabla\slice{E'_P}{n}{m-1} - \nabla\slice{E'_P}{n-1}{m-1} \right\vert\right].
 \end{equation}
\end{lemma}

\begin{proof}
Let us start with the  equality
\begin{align}\label{eqn:diff grad uv}
 \left|\nabla\slice{E'_P}{n}{m}-\nabla\slice{E'_P}{n-1}{m}\right|=\left|\braket{P^f+\slice{B}{n}{0}+\slice{B^*}{n}{0}}_{\slice{\Psi'_{P}}{n}{m}}-\braket{P^f+\slice{B}{n-1}{0}+\slice{B^*}{n-1}{0}}_{\slice{\Psi'_{P}}{n-1}{m}}\right|.
\end{align}
As $\slice{\Psi'_P}{n}{m}$ and $\slice{\eta_P}{n}{m}$ belong to the same ray in $\cal H_P$, we obtain
\begin{align*}
  \eqn{eqn:diff grad uv}&=\left|\braket{P^f+\slice{B}{n}{0}+\slice{B^*}{n}{0}}_{\slice{\eta_P}{n}{m}}-\braket{P^f+\slice{B}{n-1}{0}+\slice{B^*}{n-1}{0}}_{\slice{\eta_P}{n-1}{m}}\right|.
\end{align*}
In order to shorten the formulae we define
\begin{align*}
 V_n:=P^f+\slice{B}{n}{0}+\slice{B^*}{n}{0}
\end{align*}
so that
\begin{align}
   \eqn{eqn:diff grad uv}&\leq \frac{1}{\|\slice{\eta_P}{n-1}{m}\|^2}\left|\braket{\slice{\eta_P}{n}{m},V_n\slice{\eta_P}{n}{m}}-\braket{\slice{\eta_P}{n-1}{m},V_{n-1}\slice{\eta_P}{n-1}{m}}\right|\label{eqn:grad diff 2}\\
  &\quad+\left|\frac{1}{\|\slice{\eta_P}{n}{m}\|^2}-\frac{1}{\|\slice{\eta_P}{n-1}{m}\|^2}\right|\left|\braket{\slice{\eta_P}{n}{m},V_n\slice{\eta_P}{n}{m}}\right|\label{eqn:grad diff 1}.
\end{align}
Furthermore, by the definitions in \eqn{eqn:def A Ck}, \eqn{eqn:pi def} and \eqn{eqn:eta_m} we have
\begin{align}\label{eqn:est by grad e}
 \left|\braket{\slice{\eta_P}{n}{m}, V_n \slice{\eta_P}{n}{m}}\right|=\left|\braket{\slice{\phi_P}{n}{m}, \slice{\Pi_P}{n}{m} \slice{\phi_P}{n}{m}}+C^{(k,n)}_{P,m} \Vert \slice{\phi_P}{n}{m} \Vert^2 \right|\leq C\,,
\end{align}
where we used \lem{lemma:pi terms}. Hence, by \eqn{eqn:norm bounds} we get the estimate
\begin{align}
 \eqn{eqn:grad diff 1} \leq C\frac{\|\slice{\eta_P}{n}{m}-\slice{\eta_P}{n-1}{m}\|}{\|\slice{\eta_P}{n}{m}\|^2\|\slice{\eta_P}{n-1}{m}\|^2}\leq C\|\slice{\eta_P}{n}{m}-\slice{\eta_P}{n-1}{m}\|\label{eqn:gd4}.
\end{align}
Next, we proceed with
\begin{align}
 \eqn{eqn:grad diff 2}\leq C\Bigg[&
   \left|\braket{(\slice{\eta_P}{n}{m}-\slice{\eta_P}{n-1}{m}),V_n\slice{\eta_P}{n}{m}}\right|\label{eqn:gd 1}\\
   &+\left|\braket{\slice{\eta_P}{n-1}{m},(V_n-V_{n-1})\slice{\eta_P}{n}{m}}\right|\label{eqn:gd 2}\\
   &+\left|\braket{\slice{\eta_P}{n-1}{m},V_{n-1}(\slice{\eta_P}{n}{m}-\slice{\eta_P}{n-1}{m})}\right|\label{eqn:gd 3}
 \Bigg].
\end{align}
First, we observe that
\begin{equation*}
 \eqn{eqn:gd 2} \leq C\left|\braket{\slice{\eta_P}{n-1}{m},(\slice{B}{n}{n-1}+\slice{B^*}{n}{n-1})\slice{\eta_P}{n}{m}}\right|
 \leq C\left|\slice{E'_P}{n}{m}-i\right|^{1/2}\left|\braket{\slice{\eta_P}{n-1}{m} ,\slice{B}{n}{n-1}\left(\frac{1}{\slice{H'_P}{n}{m}-i}\right)^{1/2}\slice{\eta_P}{n}{m}}\right|
\end{equation*}
holds. Invoking the standard inequalities in \eqn{eqn:uv standard ineq} and the boundedness of 
\begin{align}
 \left\|H_{P,0}^{1/2}\left(\frac{1}{\slice{H'_P}{n}{m}-i}\right)^{1/2}\right\|\leq C\,,\label{eqn:hp0 const}
\end{align}
which holds by \lem{lem:a priori}, one has
\begin{align*}
 \left\|\slice{B}{n}{n-1}\left(\frac{1}{\slice{H'_P}{n}{m}-i}\right)^{1/2}\right\|_{\slice{\cal F}{n}{m}}\leq C|g|\left(\frac{1}{\beta^n}\right)^{1/2}.
\end{align*}
Hence, since the ground state energies are bounded from above and below by \cor{cor:ground state energies},
\begin{align}\label{eqn:middle term}
  \eqn{eqn:gd 2}&\leq C\left(\frac{1}{\beta^n}\right)^{1/2}
\end{align}
holds true.
Terms \eqn{eqn:gd 1} and \eqn{eqn:gd 3} can be treated similarly. By recalling the identity in \eqn{eqn:eta_mp1} we can write
\begin{align}
  \eqn{eqn:gd 1} &=\left|
    \braket{(\slice{\cal Q'_P}{n}{m} \slice{W}{m-1}{m}(\nabla\slice{E'_P}{n}{m-1})^*\slice{\eta_P}{n}{m-1}-\slice{\cal Q'_P}{n-1}{m}\slice{W}{m-1}{m}(\nabla\slice{E'_P}{n-1}{m-1})^*\slice{\eta_P}{n-1}{m-1}),V_n\slice{\eta_P}{n}{m}}
  \right|\nonumber\\
  &\leq
    \left|\braket{ (\slice{\cal Q'_P}{n}{m}-\slice{\cal Q'_P}{n-1}{m}) \slice{W}{m-1}{m}(\nabla\slice{E'_P}{n}{m-1})^* \slice{\eta_P}{n}{m-1},V_n\slice{\eta_P}{n}{m}}\right|\label{eqn:lastgd1} \\
    &\quad + \left|\braket{ \slice{\cal Q'_P}{n-1}{m} \left( \slice{W}{m-1}{m}(\nabla\slice{E'_P}{n}{m-1})^* - \slice{W}{m-1}{m}(\nabla\slice{E'_P}{n-1 }{m-1})^*\right) \slice{\eta_P}{n}{m-1} ,V_n\slice{\eta_P}{n}{m} }\right| \label{eqn:lastgd2} \\
    &\quad +
    \left|\braket{\slice{\cal Q'_P}{n-1}{m} \slice{W}{m-1}{m}(\nabla\slice{E'_P}{n-1}{m-1})^* \left(\slice{\eta_P}{n}{m-1} - \slice{\eta_P}{n-1}{m-1}\right) ,V_n\slice{\eta_P}{n}{m}}\right|. \label{eqn:lastgd3}
\end{align}
Observe that
\begin{multline*}
  \braket{\slice{\cal Q'_P}{n-1}{m} \slice{W}{m-1}{m}(\nabla\slice{E'_P}{n-1}{m-1})^* \left(\slice{\eta_P}{n}{m-1} - \slice{\eta_P}{n-1}{m-1}\right) ,V_n\slice{\eta_P}{n}{m}}\\
  = \braket{ \slice{W}{m-1}{m}(\nabla\slice{E'_P}{n-1}{m-1})^* \left(\slice{\eta_P}{n}{m-1} - \slice{\eta_P}{n-1}{m-1}\right) , \slice{\cal Q'_P}{n-1}{m}V_n\slice{\eta_P}{n}{m}}\\
  = \frac{1}{\|\slice{\eta_P}{n-1}{m}\|^{2}}\braket{ \slice{W}{m-1}{m}(\nabla\slice{E'_P}{n-1}{m-1})^* \left(\slice{\eta_P}{n}{m-1} - \slice{\eta_P}{n-1}{m-1}\right) , \slice{\eta_P}{n-1}{m}}\braket{\slice{\eta_P}{n-1}{m}, V_n\slice{\eta_P}{n}{m}}.
\end{multline*}
With 
\begin{align*}
 \left|\braket{\slice{\eta_P}{n-1}{m},V_n\slice{\eta_P}{n}{m}}\right| &\leq C \left|\slice{E'_P}{n}{m}-i\right|^{1/2}\left|\braket{\slice{\eta_P}{n-1}{m},H_{P,0}^{1/2}\left(\frac{1}{\slice{H'_P}{n}{m}-i}\right)^{1/2}\slice{\eta_P}{n}{m}} \right| \\
  &\quad+ C \left|\slice{E'_P}{n-1}{m}-i\right|^{1/2}\left|\braket{\slice{\eta_P}{n-1}{m},\left(\frac{1}{\slice{H'_P}{n-1}{m}-i}\right)^{1/2} H_{P,0}^{1/2}\slice{\eta_P}{n}{m}} \right|
\end{align*}
and \eqn{eqn:hp0 const}, we obtain the first estimate
\begin{align*}
 \eqn{eqn:lastgd3}\leq C \|\slice{\eta_P}{n}{m-1}-\slice{\eta_P}{n-1}{m-1}\| 
 \left|\braket{\slice{\eta_P}{n-1}{m},V_n\slice{\eta_P}{n}{m}}\right|\leq C\|\slice{\eta_P}{n}{m-1}-\slice{\eta_P}{n-1}{m-1}\|.
\end{align*}
Furthermore, \eqn{eqn:lastgd2} can be bounded by
\begin{align*}
 \eqn{eqn:lastgd2} &\leq C \left\Vert \left(\slice{W}{m-1}{m}(\nabla\slice{E'_P}{n}{m-1})^* - \slice{W}{m-1}{m}(\nabla\slice{E'_P}{n-1 }{m-1})^*\right) \slice{\eta_P}{n}{m-1} \right\Vert \left|\braket{\slice{\eta_P}{n-1}{m},V_n\slice{\eta_P}{n}{m}}\right| \\
 &\leq  C|g||\ln\gamma | \left \vert \nabla\slice{E'_P}{n}{m-1} - \nabla\slice{E'_P}{n-1}{m-1} \right\vert\leq C  \left \vert \nabla\slice{E'_P}{n}{m-1} - \nabla\slice{E'_P}{n-1}{m-1} \right\vert
\end{align*}
where the constraints (\ref{eqn:constraint}) has been used again.
Finally, using the resolvent expansion in \eqn{eqn:uv ir res exp} we get
\begin{align*}
 \eqn{eqn:lastgd1}\leq C \tau_m^{\frac{1}{2}}\left(\frac{ n}{\beta^n\gamma^m}\right)^{1/2}\left|\slice{E'_P}{n}{m}-i\right|^{1/2}  \sup_{z\in\Delta_m}
 \left\|\left[\left(\frac{1}{\slice{H'_P}{n-1}{m}-z}\right)^{1/2}\right]^*V_n \left(\frac{1}{\slice{H'_P}{n}{m}-i}\right)^{1/2}\slice{\eta_P}{n}{m}\right\|
\end{align*}
and the standard inequalities in \eqn{eqn:standard ineq} and \lem{lem:a priori} yield
\begin{align*}
 \eqn{eqn:lastgd1}\leq C\left(\frac{n}{\beta^n\gamma^m}\right)^{1/2}\,.
\end{align*}

Carrying out the same argument for term \eqn{eqn:gd 3} one obtains
\begin{align*}
 \eqn{eqn:gd 1}+\eqn{eqn:gd 3}\leq C\left[\left(\frac{n}{\beta^n\gamma^m}\right)^{1/2}+\|\slice{\eta_P}{n}{m}-\slice{\eta_P}{n-1}{m}\| +  \left \vert \nabla\slice{E'_P}{n}{m-1} - \nabla\slice{E'_P}{n-1}{m-1} \right\vert\right]
 \end{align*}
which, together with estimate \eqn{eqn:middle term}, proves the claim.
\end{proof}

\begin{theorem}\label{thm:IRUV induction}
There exist constants $K \geq \mathrm{max}(K_1,K_2,5)$, $g_*>0$ and $\frac{1}{2}>\gamma_*>0$ such that for $|g|\leq g_*$ and $\gamma \leq \gamma_*$ the following estimates hold true for all finite $n\in\bb N$ and $\bb N\ni m<n/\alpha$:
\begin{enumerate}[(i)]
\item
$\left \vert \nabla\slice{E'_P}{n}{m} - \nabla\slice{E'_P}{n-1}{m} \right\vert 
\leq K^{3m+1}\left(\frac{n}{\beta^n\gamma^m}\right)^{1/2}$.
\item
$
\Vert \slice{\eta_P}{n}{m} -\slice{\eta_P}{n-1}{m} \Vert
\leq K^{3m+1}\left(\frac{n}{\beta^n\gamma^m}\right)^{1/2}$.
\end{enumerate}
\end{theorem}

\begin{proof}
Let $n\in\bb N$ and fix $K \geq \mathrm{max}(K_1,K_2,5)$.
We prove the claim by induction in $m$ for $m<n/\alpha$.  Statements (i)-(ii) for $m$ will be referred to as assumptions A(i)-A(ii) while the same statements for $m+1$ are claims C(i)-C(ii).
 We recall that  $ \slice{\eta_P}{n}{0}\equiv \slice{\phi_P}{n}{0}\equiv \slice{\Psi'_P}{n}{0}/\|\slice{\Psi'_P}{n}{0}\|$ so that C(i) and C(ii) for $m=0$ are consequence of  (\ref{diff-grad-uv}) and (\ref{diff-uv}) for $|g|$ sufficiently small. The induction step $m \Rightarrow (m+1)$ for $(m+1)<\frac{n}{\alpha}$ is a straightforward consequence of inequalities \eqn{eq: Delta grad} and \eqn{eq: Delta eta}: For C(i) we estimate
\begin{align*}
 \left|\nabla\slice{E'_P}{n}{m+1}-\nabla\slice{E'_P}{n-1}{m+1}\right|
 &\leq  K_2\left[\left(\frac{n}{\beta^n\gamma^{m+1}}\right)^{1/2}+\|\slice{\eta_P}{n}{m+1}-\slice{\eta_P}{n-1}{m+1}\| +  \left \vert \nabla\slice{E'_P}{n}{m} - \nabla\slice{E'_P}{n-1}{m} \right\vert\right]\\
  &\leq  K_2\Bigg[\left(\frac{n}{\beta^n\gamma^{m+1}}\right)^{1/2}+\left \vert \nabla\slice{E'_P}{n}{m} - \nabla\slice{E'_P}{n-1}{m} \right\vert\\
  &\qquad+\Vert \slice{\eta_P}{n}{m} -\slice{\eta_P}{n-1}{m} \Vert + K_1\left[\left(\frac{n}{\beta^n \gamma^{m+1}}\right)^{1/2} + \left \vert \nabla\slice{E'_P}{n}{m} - \nabla\slice{E'_P}{n-1}{m} \right\vert \right]\Bigg]\\
  &\leq K(K+1)\left(\frac{n}{\beta^n\gamma^{m+1}}\right)^{1/2}+K(K+1)\left \vert \nabla\slice{E'_P}{n}{m} - \nabla\slice{E'_P}{n-1}{m} \right\vert\\
  &\quad + K \Vert \slice{\eta_P}{n}{m} -\slice{\eta_P}{n-1}{m} \Vert.
\end{align*}
Hence, A(i) and A(ii) and $\gamma<\frac{1}{2}$ imply
\begin{align*}
  \left|\nabla\slice{E'_P}{n}{m+1}-\nabla\slice{E'_P}{n-1}{m+1}\right|&\leq K^{3(m+1)+1}\left(\frac{n}{\beta^n\gamma^{m+1}}\right)^{1/2}\left[\left(\frac{1}{K^2}+\frac{1}{K^3}\right)+\left(\frac{1}{K}+\frac{1}{K^2}\right)+\frac{1}{K^2}\right],
\end{align*}
which by the assumption on $K$ proves C(i). For C(ii), using \eqn{eq: Delta eta} again, we get
\begin{align*}
\Vert \slice{\eta_P}{n}{m+1} -\slice{\eta_P}{n-1}{m+1} \Vert
&\leq \Vert \slice{\eta_P}{n}{m} -\slice{\eta_P}{n-1}{m} \Vert + K_1\left[\left(\frac{n}{\beta^n \gamma^{m+1}}\right)^{1/2} + \left \vert \nabla\slice{E'_P}{n}{m} - \nabla\slice{E'_P}{n-1}{m} \right\vert \right]\\
&\leq K^{3(m+1)+1}\left(\frac{n}{\beta^n\gamma^{m+1}}\right)^{1/2}\left[\frac{1}{K^3}+\frac{1}{K^3}+\frac{1}{K^2}\right],
\end{align*}
which by the assumption on $K$ and $\gamma<\frac{1}{2}$ proves C(ii) and concludes the proof.
\end{proof}

\begin{corollary}\label{cor:state conv}
 Let $n > \alpha m\geq 1$. For $|g|$ and $\gamma$ as in \thm{thm:IRUV induction}  the estimate
\begin{align*}
 \|\slice{\phi_P}{n}{m}-\slice{\phi_P}{n-1}{m}\|\leq  C m K^{3m+1}\left(\frac{n}{\beta^n\gamma^m}\right)^{1/2}
\end{align*}
holds true.
\end{corollary}

\begin{proof}
By \dfn{def:etas} and the unitarity of the transformations $W_m$ we have that
\begin{align}\label{eqn:the important term}
 \|\slice{\phi_P}{n}{m}-\slice{\phi_P}{n-1}{m}\|\leq \|[W_m(\nabla\slice{E'_P}{n}{m})-W_m(\nabla\slice{E'_P}{n-1}{m})]\slice{\eta_P}{n}{m}\| + \|\slice{\eta_P}{n}{m}-\slice{\eta_P}{n-1}{m}\|.
\end{align}
The lower bound on the norm of $\slice{\eta_P}{n}{m}$ in \eqn{eqn:norm bounds} together with \lem{lem:ir trafo diff} and the constraints (\ref{eqn:constraint}) yield the estimate
\begin{align*}
 \|[W_m(\nabla\slice{E'_P}{n}{m})-W_m(\nabla\slice{E'_P}{n-1}{m})]\slice{\eta_P}{n}{m}\|\leq C m\left|\nabla\slice{E'_P}{n}{m}-\nabla\slice{E'_P}{n-1}{m}\right|.
\end{align*}
The claim then follows from a direct application of \thm{thm:IRUV induction}.
\end{proof}

Before we can prove the second main result, we must show the convergence of the fiber Hamiltonians under the simultaneous removal of the UV and IR cut-off, $\slice{H^{W'}_P}{n(m)}{m} \to \slice{H^{W'}_P}{\infty}{\infty}$. For this, we need a slightly faster scaling $n(m)$.
\begin{lemma}\label{lem:res conv-2}
 Under the same assumptions of \thm{thm:IRUV induction}, there exist $\bar\alpha\geq \alpha$ such that for any $\bb N\ni \alpha'> \bar\alpha$ and $n(m)=\alpha'm$, the Hamiltonians $(\slice{H^{W'}_P}{n(m)}{m})_{m\in\bb N}$ converge in the norm resolvent sense as $m\to \infty$.
\end{lemma}

\begin{proof}
The convergence of the resolvent of $\slice{H^{W'}_P}{n(m)}{m}$ consists of direct applications of results of \sct{sec:main proof}, \sct{sec:W ground states} and the present section. Let $z=i\lambda$ with $|\lambda|>1$. First, we observe that for all $m\in\bb N$ the range of $(\slice{H^{W'}_P}{n(m)}{m}-z)^{-1}$ equals $D(H_{P,0})$ which is dense in $\cal F$.
By the Trotter-Kato Theorem \cite[Theorem VIII.22]{reed_methods_1981} it suffices to prove that the family of resolvents $([\slice{H^{W'}_P}{n(m)}{m}-z]^{-1})_{m\in\bb N}$ is convergent. We begin with
\begin{multline*}
\left\Vert \frac{1}{\slice{H^{W'}_P}{l}{m}-z} - \frac{1}{\slice{H^{W'}_P}{l-1}{m}-z} \right\Vert 
\leq \left\Vert \frac{1}{\slice{H_P'}{l}{m} -z} - \frac{1}{\slice{H_P'}{l-1}{m}-z} \right\Vert \\
+ \left\Vert \frac{1}{W_m(\nabla\slice{E'_P}{l}{m})^*\slice{H_P'}{l-1}{m}W_m(\nabla\slice{E'_P}{l}{m}) -z} - \frac{1}{W_m(\nabla\slice{E'_P}{l-1}{m})^*\slice{H_P'}{l-1}{m}W_m(\nabla\slice{E'_P}{l-1}{m})-z} \right\Vert
\end{multline*}
where we used unitarity of $W_{m}$ in the first line. Mimicking \cor{lem:resolvent diff}, the first term is bounded above by
\begin{equation*}
C|g| \left(\frac{l}{\beta^l}\right)^{1/2}.
\end{equation*}
With the standard inequalities, the second term is bounded by
\begin{equation*}
C|g|\left\Vert\frac{1}{(\slice{H_P}{l-1}{m}-z)^{1/2}} \right\Vert \cdot 
\left\Vert(H^f)^{1/2}\frac{1}{(\slice{H_P}{l-1}{m}-z)^{1/2}} \right\Vert\cdot 
\frac{1}{\tau_m^{1/2}}\left \vert \nabla\slice{E'_P}{l}{m} - \nabla\slice{E'_P}{l-1}{m} \right\vert,
\end{equation*}
which can be further bounded by
\begin{equation*}
C|g|\frac{1}{|\Im z|}\gamma^{-m/2} K^{3m+1}\left(\frac{l}{\beta^l\gamma^m}\right)^{1/2}
\end{equation*}
with the help of \lem{lem:a priori} and \thm{thm:IRUV induction}.
Hence, it holds
\begin{equation}
\left\Vert \frac{1}{\slice{H^{W'}_P}{n(m)}{m}-z} - \frac{1}{\slice{H^{W'}_P}{n(m-1)}{m}-z} \right\Vert 
\leq \alpha' C|g| K \left(\alpha' m\right)^{1/2} \left(\frac{K^3}{\gamma \beta^{\alpha'/2}}\right)^m\label{eqn:res1}
\end{equation}
where 
\begin{equation*}
  \frac{K^3}{\gamma \beta^{\alpha'/2}}<1
\end{equation*}
for $\alpha'\geq\bar\alpha$ and $\bar \alpha$ sufficiently large.\\

Moreover, using the explicit expressions \eqn{eqn:Hw}, \eqn{eqn:tilde H}, \lem{lem:a priori}, the bound
\begin{equation}\label{eqn:fin_e_diff}
\left\vert \nabla\slice{E'_P}{n}{m} - \nabla\slice{E'_P}{n}{m-1} \right \vert \leq C\gamma^{m/4}
\end{equation}
at fixed $n$ from \lem{lem:grad E diff}, and a resolvent expansion one can show that
\begin{multline}
\left\Vert \frac{1}{\slice{H^{W'}_P}{n(m-1)}{m}-z} - \frac{1}{\slice{\widetilde{H}^{W'}_P}{n(m-1)}{m}-z} \right\Vert \\
\leq \frac{C}{|\Im z|} \left\Vert \left(\frac{1}{\slice{H^{W'}_P}{n(m-1)}{m}-z}\right)^{1/2} \left[ \slice{\widetilde{H}^{W'}_P}{n(m-1)}{m}-\slice{H^{W'}_P}{n(m-1)}{m}\right] \left(\frac{1}{\slice{H^{W'}_P}{n(m-1)}{m}-z}\right)^{1/2}\right\Vert\label{eqn:res2}
\end{multline}
where the right-hand side in \eqn{eqn:res2} can be controlled in terms of \eqn{eqn:fin_e_diff}.\\

Furthermore, we observe that
\begin{equation}
\left\Vert \frac{1}{\slice{\widetilde{H}^{W'}_P}{n(m-1)}{m}-z} - \frac{1}{\slice{H^{W'}_P}{n(m-1)}{m-1}-z} \right\Vert\leq \frac{C|g|\gamma^{(m-1)/2}}{|\Im z|}.\label{eqn:res3}
\end{equation}
by operator estimates similar to those used to control (\ref{eqn:ir diff 2}).\\
 
 Finally, for $\alpha'\geq\bar\alpha$ and $\bar \alpha$ sufficiently large, the estimates in \eqn{eqn:res1},\eqn{eqn:res2} and \eqn{eqn:res3} imply that the family of resolvents $([\slice{H^{W'}_P}{n(m)}{m}-z]^{-1})_{m\in\bb N}$ is a Cauchy sequence in the norm topology, which concludes the proof.
\end{proof}

We can now prove the second main result, namely the convergence of the ground state vectors $\slice{\phi_P}{n}{m}$ as $n,m\to\infty$ with $n\equiv n(m)$.

\begin{proof}[Proof of \thm{thm:ir} in \sct{sec:main results}]\mbox{}
\begin{enumerate}[(i)]
 \item  Define
\begin{align}\label{eqn:alpha}
 \alpha_{\mathrm{min}} :=\max\left\{\left|\frac{6\ln K -\ln|\gamma|}{\ln\beta}\right|,\bar\alpha\right\}.
\end{align}
For any $\bb N\ni \alpha' > \alpha_\mathrm{min}$, let $n(m)=\alpha' m$. By \thm{thm:key theorem} and \cor{cor:state conv} we can estimate
\begin{align*}
 \|\slice{\phi_P}{n(m)}{m}-\slice{\phi_P}{n(m-1)}{m-1}\|&\leq \|\slice{\phi_P}{n(m-1)}{m-1}-\slice{\phi_P}{n(m-1)}{m}\|+\sum_{l=\alpha'(m-1)}^{\alpha'm}\|\slice{\phi_P}{l}{m}-\slice{\phi_P}{l-1}{m}\| \nonumber\\
 &\leq m\gamma^{\frac{m-1}{4}}
 +\alpha' \left[ CmK^{3m+1} \left(\frac{\alpha'm}{\beta^{\alpha'(m-1)\gamma^m}}\right)^{1/2} \right]\nonumber\\
 &\leq m\gamma^{\frac{m-1}{4}} + m^{3/2} {\alpha'}^{3/2}C K \beta^{\alpha'/2} \left( \frac{K^3}{(\beta^{\alpha'}\gamma)^{1/2}} \right)^m 
\end{align*}
Due to \eqn{eqn:alpha} the term $\frac{K^3}{(\beta^{\alpha'}\gamma)^{1/2}}<1$ so that $(\slice{\phi_P}{n(m)}{m})_{m\in\bb N}$ is a Cauchy sequence. We denote its limit by $\slice{\phi_{P}}{\infty}{\infty}$.  Finally \thm{thm:key theorem} ensures that the vector  $\slice{\phi_{P}}{\infty}{\infty}$ has norm larger than $\frac{1}{2}$.
 \item Let $\slice{E'_P}{\infty}{\infty}:=\lim_{m\to\infty}\slice{E'}{\infty}{m}$ which exists by \cor{cor:ground state energies}. By \lem{lem:res conv-2} and (i), $\slice{E'_P}{\infty}{\infty}$ is the eigenvalue corresponding to the eigenvector $\slice{\phi_P}{\infty}{\infty}$ of $\slice{H^{W'}_P}{\infty}{\infty}$. Furthermore,
\begin{align*}
 \spec{\slice{H^{W'}_P}{n}{m}}=\spec{\slice{H'_P}{n}{m}}\subseteq[\slice{E'_P}{n}{m},\infty).
\end{align*}
By the nonexpansion property of the norm resolvent convergence for self-adjoint operators \cite[Theorem VIII.24]{reed_methods_1981}, this implies that $\slice{\phi_P}{\infty}{\infty}$  is ground state of $\slice{H^{W'}_P}{\infty}{\infty}$ and $\slice{E'_P}{\infty}{\infty}$ is the ground state energy.
\end{enumerate}
\end{proof}

\appendix

\section{Proofs of  \lem{lem:a priori} and \cor{cor:ground state energies}}\label{sec:Apriori}

\begin{proof}[Proof of \lem{lem:a priori}]
Let $\psi\in D(H_{P,0}^{1/2})$. We start with the identity
\begin{align}\label{eqn:h0}
  \braket{\psi,H_{P,0}\psi}=\braket{\psi,\slice{H'_P}{n}{m}\psi}-\braket{\psi,\slice{\Delta H_P'}{n}{0}\psi}-\braket{\psi,g\slice{\Phi}{0}{m}\psi}
\end{align}
where
\begin{align*}
 \braket{\psi,\slice{\Delta H_P'}{n}{0}\psi}&=\braket{\psi,\left[\frac{1}{2}\left((\slice{B}{n}{0})^2+({{\slice{B^*}{n}{0}}})^2\right)+{\slice{B^*}{n}{0}}\cdot\slice{B}{n}{0}-(P-P^f)\cdot\slice{B}{n}{0}-\slice{B^*}{n}{0}\cdot(P-P^f)\right]\psi}\\
 &=\Re\left[\braket{\psi,(\slice{B}{n}{0})^2\psi}+\braket{\slice{B}{n}{0}\psi,\slice{B}{n}{0}\psi}-2\braket{(P-P^f)\psi,\slice{B}{n}{0}\psi}\right].
\end{align*}
We denote the number operator of bosons in the momentum range $[\kappa,\sigma_n)$ by
\[\slice{N}{n}{0}:=\int_{\cal B_{\sigma_n}\setminus\cal B_\kappa} dk\;b(k)^*b(k)\]
and express the vector $\psi\in\cal F$ as a sequence $(\psi^j)_{j\geq 0}$ of $j$-particle wave functions $\psi^j\in L^2(\bb R^{3j},\bb C)$, $j\geq 1$, and $\psi^0\in\bb C$.
Following \cite[Proof of Lemma 5]{nelson_interaction_1964} it is convenient to consider an estimate of the following type
\begin{align}
 \Re\braket{\psi,(\slice{B}{n}{0})^2\psi}&=\Re\braket{(\slice{N}{n}{0}+3)^{1/2}\psi,(\slice{N}{n}{0}+3)^{-1/2}(\slice{B}{n}{0})^2\psi}\nonumber\\
 &\leq \left\|(\slice{N}{n}{0}+3)^{1/2}\psi\right\| \left\|(\slice{N}{n}{0}+3)^{-1/2}(\slice{B}{n}{0})^2\psi\right\|.\label{eqn:h0_1}
\end{align}
We consider the two norms in \eqn{eqn:h0_1} separately. For $I\subset\bb R^{+}_{0}$ let $\charf{I}(k)\equiv\charf{I}(|k|)$ denote the characteristic function of $I$. Schwarz's inequality gives
\begin{multline}
  \left\Vert \left(N|_{0}^{n}+3\right)^{-1/2}\left(B|_{0}^{n}\right)^{2}\psi\right\Vert ^{2} \\
\leq  c_1g^{4}\sum_{j=0}^{\infty}\int dk_{1}\ldots\int dk_{j+2}\frac{(j+1)(j+2)\omega(k_{j+1})^{1/2}\charf{[\kappa,\infty)}(k_{j+1})\omega(k_{j+2})^{1/2}\charf{[\kappa,\infty)}(k_{j+2})}{\sum_{i=1}^{j}\charf{[\kappa,\infty)}(k_{i})+3} \times \\
\times\left|\psi^{(j+2)}(k_{1}\ldots k_{j+2}) \right|^{2} \\
=  c_1g^{4}\sum_{j=0}^{\infty}\int dk_{1}\ldots\int dk_{j+2}\frac{(j+1)(j+2)\omega(k_{j+1})^{1/2}\omega(k_{j+2})^{1/2}\charf{[\kappa,\infty)}(k_{j+1})\charf{[\kappa,\infty)}(k_{j+2})}{\sum_{i=1}^{j+2}\charf{[\kappa,\infty)}(k_{i})+1} \times \\
\times \left|\psi^{(j+2)}(k_{1}\ldots k_{j+2})\right|^{2} \\
\leq  c_1g^{4}\sum_{j=0}^{\infty}\int dk_{1}\ldots\int dk_{j+2}(j+1)(j+2)\frac{1}{2}\left[\omega(k_{j+1})\charf{[\kappa,\infty)}(k_{j+2})+\omega(k_{j+2})\charf{[\kappa,\infty)}(k_{j+1})\right] \times \\
\times\frac{\left|\psi^{(j+2)}(k_{1}\ldots k_{j+2}) \right|^{2}}{\sum_{i=1}^{j+2}\charf{[\kappa,\infty)}(k_{i})+1}\label{eqn:Nest}.
\end{multline}
for an $n$-independent and finite constant
\[c_1:=\left(\int dk\;\left|k\frac{\rho(k)}{\frac{|k|^2}{2}+\omega(k)}\frac{\charf{[\kappa,\infty)}(k)}{\omega(k)^{1/4}}\right|^2\right)^{1/2}.\]
Using the symmetry we get
\begin{align*}
\eqn{eqn:Nest} = & g^{4}c_1\sum_{j=0}^{\infty}\int dk_{1}\ldots\int dk_{j+2}\sum_{l=1}^{j+2}\sum_{m\neq l}\omega(k_{l})\charf{[\kappa,\infty)}(k_{m})\frac{\left|\psi^{(j+2)}(k_{1}\ldots k_{j+2}) \right|^{2}}{\sum_{i=1}^{j+2}\charf{[\kappa,\infty)}(k_{i})+1}\\
\leq & g^{4}c_1\sum_{j=0}^{\infty}\int dk_{1}\ldots\int dk_{j+2}\left[\sum_{l=1}^{j+2}\omega(k_{l})\right]\frac{\sum_{m=1}^{j+2}\charf{[\kappa,\infty)}(k_{m})}{\sum_{i=1}^{j+2}\charf{[\kappa,\infty)}(k_{i})+1}\left|\psi^{(j+2)}(k_{1}\ldots k_{j+2}) \right|^{2}\\
\leq & g^{4}c_1\left\Vert \left(H^{f}\right)^{1/2}\psi\right\Vert ^{2}.
\end{align*}
For the remaining term in \eqn{eqn:h0_1} we compute
\begin{align}
 \braket{\psi,(\slice{N}{n}{0}+3)\psi}\leq \frac{1}{\kappa}\braket{\psi,H^f\psi}+3\braket{\psi,\psi}.\label{eqn:h0_3}
\end{align}
Moreover, we estimate
\begin{align}
 \left|\braket{\psi,(P-P^f)\slice{B}{n}{0}\psi}\right|&\leq \|(P-P^f)\psi\|\;\|\slice{B}{n}{0}\psi\|\leq \sqrt{2}\|H_{P,0}^{1/2}\psi\|\;\|\slice{B}{n}{0}\psi\| \label{eqn:h0_4}
\end{align}
where by the standard inequalitites in \eqn{eqn:uv standard ineq}
\begin{align}
 \|\slice{B}{n}{0}\psi\|&\leq |g| c_2\|(H^f)^{1/2}\psi\|\label{eqn:h0_5}
\end{align}
holds true for an $n$-independent and finite constant
\[c_2:=\left(\int dk\; \left|k\frac{\rho(k)}{\frac{|k|^2}{2}+\omega(k)}\frac{\charf{[\kappa,\infty)}(k)}{\omega(k)^{1/2}}\right|^2\right)^{1/2}.\]
Finally, using the standard inequalities in \eqn{eqn:standard ineq} again, we find
\begin{align}
 \left|\braket{\psi,g\slice{\Phi}{0}{m}\psi}\right|\leq 2 |g| c_3\;\|\psi\|\;\|(H^f)^{1/2}\psi\|\leq |g|c_3\left(\braket{\psi,H_{P,0}\psi}+\braket{\psi,\psi}\right)\label{eqn:h0_6}
\end{align}
for an $m$-independent and finite constant
\begin{align*}
 c_3:=\left(\int dk\;\left|\frac{\rho(k)\charf{[0,\kappa)}(k)}{\omega(k)^{1/2}}\right|\right)^{1/2}
\end{align*}

Hence, for $|g|\leq 1$ the identity \eqn{eqn:h0} and the estimates \eqn{eqn:h0_1}-\eqn{eqn:h0_6} yield the bound
\begin{align}
 \left|\braket{\psi,\slice{\Delta H_P'}{n}{0}\psi}\right|+\left|\braket{\psi,g\slice{\Phi}{0}{m}\psi}\right|&\leq |g|\left[\namer{a}\braket{\psi,H_{P,0}\psi}+\namer{b}\braket{\psi,\psi}\right]\label{eqn:h0ineq}
\end{align}
for $m$ and $n$-independent positive constants $\namer{a}$ and $\namer{b}$.
For $|g|<\frac{1}{\namer{a}}$ inequality \eqn{eqn:h0ineq} proves the claim.
\end{proof}

\begin{proof}[Proof of \cor{cor:ground state energies}]\mbox{}
\begin{enumerate}[(i)]
 \item We note that $\slice{E'_P}{n}{m}\leq \braket{\Omega,\slice{H'_P}{n}{m}\Omega}=\frac{P^2}{2}$ and, furthermore, by applying \lem{lem:a priori} we observe that for any $\phi \in D(H_{P,0}^{1/2})$, $\|\phi\|=1$,
\begin{align*}
 0\leq (1-|g|\namer{a})\braket{\phi,H_{P,0}\phi}\leq \braket{\phi,\slice{H'_P}{n}{m}\phi}+|g|\namer{b}.
\end{align*}
 \item First we study  the case $|k|<1$ where we follow a strategy similar to \cite[Section VI]{chen_infraparticle_2009}: 
\begin{align*}
\slice{E'_{P-k}}{n}{m}-\slice{E'_{P}}{n}{m} &= \inf_{\|\varphi\|=1}\left[\braket{\varphi,(\slice{H'_{P-k}}{n}{m}-\slice{H'_{P}}{n}{m})\varphi}+\braket{\varphi,\slice{H'_{P}}{n}{m}\varphi}-\slice{E'_{P}}{n}{m}\right]\\
 &\geq \inf_{\|\varphi\|=1}\left[\frac{k^2}{2}-|k|\;|\braket{\varphi,(P-P^f+\slice{B}{n}{0}+\slice{B^*}{n}{0})\varphi}|+\braket{\varphi,\slice{H'_{P}}{n}{m}\varphi}-\slice{E'_P}{n}{m}\right]
\end{align*}
 where the infimum is meant to be taken over $\varphi\in D(H_{P,0}^{1/2})\cap\slice{\cal F}{n}{m}$ only. By the standard estimates \eqn{eqn:uv standard ineq} we get
\begin{align}\label{eqn:grad E estimate}
  |\braket{\varphi,(P-P^f+\slice{B}{n}{0}+\slice{B^*}{n}{0})\varphi}|\leq (\sqrt 2+2|g|C)\|H^{1/2}_{P,0}\varphi\|
\end{align}
 where $C$ does not depend on $n$ since $\slice{B^*}{n}{0}$ can be seen to act to the left as $\slice{B}{n}{0}$ and the integral in \eqn{eqn:uv standard ineq} converges for any $n\in\bb N\cup\{\infty\}$. Using \lem{lem:a priori} it turns out that  $\slice{E'_{P-k}}{n}{m}-\slice{E'_{P}}{n}{m}$ is bounded from below by
\begin{multline*}
\inf_{\|\varphi\|=1}\left[\frac{k^2}{2}-|k|\;\frac{\sqrt 2+2C|g|}{\sqrt{1-|g|\namer{a}}}\sqrt{\braket{\varphi,\slice{H'_{P}}{n}{m}\varphi}+|g|\namer{b}}+\braket{\varphi,\slice{H'_{P}}{n}{m}\varphi}-\slice{E'_{P}}{n}{m}\right]\\
\geq \inf_{\lambda\geq 0}\left[\frac{k^2}{2}-|k|\;\frac{\sqrt 2+2C|g|}{\sqrt{1-|g|\namer{a}}}\sqrt{\lambda+\slice{E'_{P}}{n}{m}+|g|\namer{b}}+\lambda\right]=:\inf_{\lambda\geq 0}f(\lambda)
\end{multline*}
where 
\begin{equation}
f(\lambda):=\frac{k^2}{2}-|k|\;\frac{\sqrt 2+2C|g|}{\sqrt{1-|g|\namer{a}}}\sqrt{\lambda+\slice{E'_{P}}{n}{m}+|g|\namer{b}}+\lambda
\end{equation}
The infimum can be attained either at $\lambda^*=0$ or at $\lambda^*$ such that $f'(\lambda^*)=0$, i.e.
\begin{equation}\label{expression}
\lambda^*=\frac{|k|^2}{4}\frac{(\sqrt 2+2C|g|)^2}{{1-|g|\namer{a}}}-(\slice{E'_{P}}{n}{m}+|g|\namer{b})
\end{equation}

\begin{enumerate}
 \item[Case $\lambda^*=0$:] Since
\begin{align*}
 f(0)\geq-|k|\frac{\sqrt 2+2C|g|}{\sqrt{1-|g|\namer{a}}}\sqrt{\slice{E'_{P}}{n}{m}+|g|\namer{b}}
\end{align*}
and, by claim (ii),
\begin{align*}
0\leq \slice{E'_{P}}{n}{m}+|g|\namer{b}\leq \frac{P^2}{2}+|g|\namer{b}\leq \frac{P^2_{\mathrm{max}}}{2}+|g|\namer{b}\,,
\end{align*}
we obtain the lower bound
\begin{align}\label{eqn:b1}
 f(0)\geq -|k|\frac{\sqrt 2+2C|g|}{\sqrt{1-|g|\namer{a}}}\left(\frac{P_{\mathrm{max}}}{\sqrt 2}+\cal O(|g|)\right)=-|k|P_{\mathrm{max}}\left(1+\cal O(|g|)\right).
\end{align}
 \item[Case $\lambda^*>0$:] To evaluate
\begin{align*}
 f(\lambda^*)=\frac{k^2}{2}\left(1-\frac{1}{2}\frac{(\sqrt 2+2C|g|)^2}{1-|g|\namer{a}}\right)-(\slice{E'_{P}}{n}{m}+|g|\namer{b})
\end{align*}
we consider that $\lambda^* $ given in (\ref{expression}) is assumed to be larger than zero. This implies that
\begin{align}\label{eqn:b2}
 f(\lambda^*)>\frac{k^2}{2}\left(1-\frac{(\sqrt 2+2C|g|)^2}{1-|g|\namer{a}}\right)=-k^2\left(\frac{1}{2}+\cal O(g)\right) > -|k|\left(\frac{1}{2}+\cal O(g)\right)
\end{align}
where we have used that $|k|<1$.
\end{enumerate}
Recall that $P_{\mathrm{max}}=\frac{1}{4}$. Therefore, taking the minimum of both lower bounds \eqn{eqn:b1} and \eqn{eqn:b2} for $|g|$ sufficiently small proves that, for all $|k|<1$,
\begin{align}\label{main}
 \slice{E'_{P-k}}{n}{m}-\slice{E'_{P}}{n}{m}\geq -c|k|\,,
\end{align}
for any $c>\frac{1}{2}$, and in particular for $c=\namer{c:cp} := \frac{3}{4}$.
 
For the case $|k|\geq 1$ \thm{thm:ground state energies} implies:
\begin{align}
 \slice{E'_{P-k}}{n}{m}-\slice{E'_{P}}{n}{m}&=(\slice{E'_{P-k}}{n}{m}-\slice{E'_{0}}{n}{m})+(\slice{E'_{0}}{n}{m}-\slice{E'_{P}}{n}{m}) \geq \slice{E'_{0}}{n}{m}-\slice{E'_{P}}{n}{m}\label{eqn:case k big 1}\\
 &\geq -\namer{c:cp}|P_{\mathrm{max}}|\geq -\namer{c:cp}|k|\label{eqn:case k big 2},
\end{align}
where the step from \eqn{eqn:case k big 1} to \eqn{eqn:case k big 2} is justified by invoking the result in the case $|k|<1$, i.e., by replacing   $k=P$ in (\ref{main}) .
 \item Let $\slice{\Psi'_P}{n}{m}$ be the eigenvector corresponding to $ \slice{E'_P}{n}{m}$, then we get
\begin{align*}
 \slice{E'_P}{n+1}{m}\leq \braket{\frac{\slice{\Psi'_P}{n}{m}}{\|\slice{\Psi'_P}{n}{m}\|}\otimes\Omega,[\slice{H'_{P}}{n}{0}+\slice{\Delta H_P'}{n+1}{n}+g\slice{\Phi}{0}{m}]\frac{\slice{\Psi'_P}{n}{m}}{\|\slice{\Psi'_P}{n}{m}\|}\otimes\Omega}=\braket{\frac{\slice{\Psi'_P}{n}{m}}{\|\slice{\Psi'_P}{n}{m}\|},\slice{H'_{P}}{n}{0}\frac{\slice{\Psi'_P}{n}{m}}{\|\slice{\Psi'_P}{n}{m}\|}}=\slice{E'_P}{n}{m}
\end{align*}
as well as
\begin{align*}
 \slice{E'_P}{n}{m+1}\leq \braket{\frac{\slice{\Psi'_P}{n}{m}}{\|\slice{\Psi'_P}{n}{m}\|}\otimes\Omega,[\slice{H'_{P}}{n}{m}+g\slice{\Phi}{m}{m+1}]\frac{\slice{\Psi'_P}{n}{m}}{\|\slice{\Psi'_P}{n}{m}\|}\otimes\Omega}=\braket{\frac{\slice{\Psi'_P}{n}{m}}{\|\slice{\Psi'_P}{n}{m}\|},\slice{H'_{P}}{n}{m}\frac{\slice{\Psi'_P}{n}{m}}{\|\slice{\Psi'_P}{n}{m}\|}}=\slice{E'_P}{n}{m}.
\end{align*}

\end{enumerate}
\end{proof}

\section{Transformed Hamiltonians: derivation of  identities  \eqn{eqn:Hw}, \eqn{eqn:tilde H} and \eqn{eqn:gamma diff} }\label{sec:formal}

\begin{proof}[Derivation of identity \eqn{eqn:Hw}] 
 Let $n,m\in\bb N$. Recalling \eqn{eqn:G trafo Hamiltonain} we can start with the expression
\begin{align*}
 \slice{H'_P}{n}{m}&=\frac{1}{2}\left(P-P^f\right)^2 + H^f+\frac{1}{2}[(\slice{B}{n}{0})^2+(\slice{B^*}{n}{0})^2]+\slice{B^*}{n}{0}\cdot \slice{B}{n}{0}\\
 &\quad-(P-P^f)\cdot \slice{B}{n}{0}-\slice{B^*}{n}{0}\cdot (P-P^f)+ g\slice{\Phi}{0}{m}.
\end{align*}
This Hamiltonian can be written in the form
\begin{align*}
 \slice{H'_P}{n}{m}&=\frac{1}{2}\left(P-P^f-\slice{B}{n}{0}-\slice{B^*}{n}{0}\right)^2 + H^f + g\slice{\Phi}{0}{m} + S_{P,n}
\end{align*}
where we collected terms acting in the ultraviolet region in
\begin{align*}
 S_{P,n} := -\frac{1}{2}\left([\slice{B}{n}{0},P-P^f]+[P-P^f,\slice{B^*}{n}{0}]+[\slice{B}{n}{0},\slice{B^*}{n}{0}]\right).
\end{align*}
The conjugation by $W_m(\nabla \slice{E'_P}{n}{m})$ on these various terms reads
\begin{align}\label{eqn:transformation formulae}
\begin{split}
  W_m(\nabla \slice{E'_P}{n}{m})\; P^f \;W_m(\nabla \slice{E'_P}{n}{m})^*&=P^f+A_{P,m}^{(n)}+C^{(k,n)}_{P,m} \\
 W_m(\nabla \slice{E'_P}{n}{m}) \;H^f\; W_m(\nabla \slice{E'_P}{n}{m})^*&=H^f+L_{P,m}^{(n)}+C^{(\omega,n)}_{P,m} \\
 W_m(\nabla \slice{E'_P}{n}{m}) \;\slice{\Phi}{0}{m}\; W_m(\nabla \slice{E'_P}{n}{m})^*&=\slice{\Phi}{0}{m}+C^{(\rho,n)}_{P,m}\\
 W_m(\nabla \slice{E'_P}{n}{m}) \;S_{P,n}\; W_m(\nabla \slice{E'_P}{n}{m})^*&=S_{P,n}
\end{split}
\end{align}
for
\begin{align*}
 L_{P,m}^{(n)} &:= \int dk\;\omega(k)\;\alpha_m(\nabla\slice{E'_P}{n}{m},k)[b(k)+b^*(k)].
\end{align*}
and $A_{P,m}^{(n)},C^{(k,n)}_{P,m},C^{(\omega,n)}_{P,m},C^{(\rho,n)}_{P,m}$ given in equations \eqn{eqn:def A Ck}.

Using these formulae we find
\begin{align*}
 W_m(\nabla \slice{E'_P}{n}{m})\;\slice{H'_P}{n}{m}\;W_m(\nabla \slice{E'_P}{n}{m})^*&=\frac{1}{2}\Big(P-P^f-A_{P,m}^{(n)}-\slice{B}{n}{0}-\slice{B^*}{n}{0}-C^{(k,n)}_{P,m}\Big)^2 \\
 & \quad +\Big(H^f +L_{P,m}^{(n)}+C^{(\omega,n)}_{P,m}\Big) + \Big(g\slice{\Phi}{0}{m}+C^{(\rho)}_m\Big) + S_{P,n}.
\end{align*}
Applying the identity \eqn{eqn:grad e formuar} we further have
\begin{align}
 P&=\nabla \slice{E'_P}{n}{m}+\braket{[P^f+\slice{B}{n}{0}+\slice{B^*}{n}{0}]}_{\slice{\Psi'_{P}}{n}{m}} = \nabla \slice{E'_P}{n}{m}+\braket{P^f+A_{P,m}^{(n)}+\slice{B}{n}{0}+\slice{B^*}{n}{0}}_{W_m\slice{\Psi'_{P}}{n}{m}}+C^{(k,n)}_{P,m}\nonumber\\
 &=\nabla \slice{E'_P}{n}{m}+\braket{\slice{\Pi_P}{n}{m}}_{\slice{\phi_P}{n}{m}}+C^{(k,n)}_{P,m}\,,\label{eqn:exp pi}
\end{align}
so that we obtain
\begin{align*}
 &W_m(\nabla \slice{E'_P}{n}{m})\;\slice{H'_P}{n}{m}\;W_m(\nabla \slice{E'_P}{n}{m})^*\\
 &\qquad=\frac{1}{2}\left(\nabla \slice{E'_P}{n}{m}+\braket{P^f+A_{P,m}^{(n)}+\slice{B}{n}{0}+\slice{B^*}{n}{0}}_{\slice{\phi_P}{n}{m}}-\left(P^f+A_{P,m}^{(n)}+\slice{B}{n}{0}+\slice{B^*}{n}{0}\right)\right)^2 \\
 &\qquad \quad + H^f +L_{P,m}^{(n)}+C^{(\omega,n)}_{P,m} + g\slice{\Phi}{\kappa}{\tau_m} + C^{(\rho)}_m + S_{P,n}\\
 &\qquad=\frac{1}{2}\slice{\Gamma_P}{n}{m}^2 + \frac{1}{2}\nabla \slice{E'_P}{n}{m}^2\\
 &\qquad\quad+\nabla \slice{E'_P}{n}{m} \cdot\left(\braket{P^f+A_{P,m}^{(n)}+\slice{B}{n}{0}+\slice{B^*}{n}{0}}_{\slice{\phi_P}{n}{m}}-\left(P^f+A_{P,m}^{(n)}+\slice{B}{n}{0}+\slice{B^*}{n}{0}\right)\right)\\
 &\qquad \quad + H^f +L_{P,m}^{(n)}+C^{(\omega,n)}_{P,m} + g\slice{\Phi}{\kappa}{\tau_m} + C^{(\rho,n)}_{P,m} + S_{P,n}.
\end{align*}
The transformation $W_m$ was designed to yield the following cancellation
\begin{equation}\label{eqn:cancellation}
 -\nabla \slice{E'_P}{n}{m}\cdot A_{P,m}^{(n)}+L_m+g\slice{\Phi}{\kappa}{\tau_m}=0.
\end{equation}
Hence, using the abbreviations introduced in the beginning of \sct{sec:W ground states}, we finally arrive at the form
\begin{equation}\label{eqn:Hw2}
  \slice{H^{W'}_P}{n}{m}:=W_m(\nabla \slice{E'_P}{n}{m})\;\slice{H'_P}{n}{m}\;W_m(\nabla \slice{E'_P}{n}{m})^* = \frac{1}{2}\slice{\Gamma_P}{n}{m}^2+ H^f-\nabla \slice{E'_P}{n}{m} \cdot P^f + C_{P,m}^{(n)} + \slice{R_P}{n}{m}.
\end{equation}
By analogous methods as in \cite{nelson_interaction_1964} for the ultraviolet region it can then we verified that this equality actually holds on $D(H_{P,0})$.
\end{proof}

\begin{proof}[Derivation of Identity \eqn{eqn:tilde H}]
From the definition of $\slice{\widetilde H_P^{W'}}{n}{m}$, we can write
\begin{equation*}
\slice{\widetilde H_P^{W'}}{n}{m}=W_{m}(\nabla \slice{E'_P}{n}{m-1})\;W_{m-1}(\nabla \slice{E'_P}{n}{m-1})^*\;[\slice{ H_P^{W'}}{n}{m-1}+g\slice{\Phi}{m-1}{m}]\;W_{m-1}(\nabla \slice{E'_P}{n}{m-1})\;W_{m}(\nabla \slice{E'_P}{n}{m-1})^*
\end{equation*} 
which by virtue of the formulae \eqn{eqn:transformation formulae} as well as identity \eqn{eqn:Hw2} gives
\begin{align*}
 \slice{\widetilde H_P^{W'}}{n}{m}&=\frac{1}{2}\left(\slice{\Gamma_P}{n}{m-1}+\widetilde A_{P,m}^{(n)}-A_{P,m-1}^{(n)}+\widetilde C^{(k,n)}_{P,m}-C^{(k,n)}_{P,m-1}\right)^2\\
 &\quad+ H^f + \widetilde L_{P,m}^{(n)}-L_{P,m-1}+\widetilde C^{(\omega,n)}_{P,m}-
C^{(\omega,n)}_{P,m-1}\\
 &\quad- \nabla\slice{E'_P}{n}{m-1}\cdot\left(P^f + \widetilde A_{P,m}^{(n)}-A_{P,m-1}^{(n)}+\widetilde C^{(k,n)}_{P,m}-C^{(k,n)}_{P,m-1}\right)\\
 &\quad+ g\slice{\Phi}{0}{m}+\widetilde C^{(\rho,n)}_{P,m}-C^{(\rho,n)}_{P,m-1} + C_{P,m-1}^{(n)} + \slice{R_P}{n}{m-1}
\end{align*}
for
\begin{align*}
 \widetilde L_{P,m}^{(n)} &:= \int dk\;\omega(k)\;\alpha_m(\nabla\slice{E'_P}{n}{m-1},k)[b(k)+b^*(k)].
\end{align*}
Due to the cancellation \eqn{eqn:cancellation} and
\begin{align*}
 \widetilde C^{(k,n)}_{P,m}=C^{(k,n)}_{P,m-1}-\nabla\slice{E'_P}{n}{m-1}\cdot\left(\widetilde C^{(k,n)}_{P,m}-C^{(k,n)}_{P,m-1}\right)+\widetilde C^{(\omega,n)}_{P,m}-
C^{(\omega,n)}_{P,m-1}+\widetilde C^{(\rho,n)}_{P,m}-
C^{(\rho,n)}_{P,m-1}
\end{align*}
we finally obtain
\begin{equation*} 
  \slice{\widetilde H^{W'}_P}{n}{m}=\frac{1}{2}(\slice{\Gamma_P}{n}{m-1}+\widetilde A_{P,m}^{(n)}-A_{P,m-1}^{(n)}+\widetilde C^{(k,n)}_{P,m}-C^{(k,n)}_{P,m-1})^2 + H^f-\nabla \slice{E'_P}{n}{m-1} \cdot P^f+\widetilde C_{P,m}^{(n)}+\slice{R_P}{n}{m-1}.
\end{equation*}
One can verify that this identity holds on $D(H_{P,0})$.
\end{proof}

\begin{proof}[Derivation of Identity \eqn{eqn:gamma diff}]
By definitions \eqn{eqn:def Gamma} and \eqn{eqn:def tilde Gamma},
\begin{equation*}
  \slice{\widetilde \Gamma_P}{n}{m}-\slice{\Gamma_P}{n}{m-1}= \braket{\slice{\Pi_P}{n}{m}}_{\slice{\phi_P}{n}{m-1}}-\braket{\slice{\widetilde\Pi_P}{n}{m}}_{\slice{\widetilde \phi_P}{n}{m}} +\slice{\widetilde\Pi_P}{n}{m}-\slice{\Pi_P}{n}{m}.
\end{equation*}
so that \eqn{eqn:exp pi} yields
\begin{equation*}
  \slice{\widetilde \Gamma_P}{n}{m}-\slice{\Gamma_P}{n}{m-1}= \nabla \slice{E'_P}{n}{m}-\nabla \slice{E'_P}{n}{m-1} + \widetilde A_{P,m}^{(n)}-A_{P,m-1}^{(n)} + \widetilde C^{(k,n)}_{P,m}-C^{(k,n)}_{P,m-1}.
\end{equation*}
One can verify that this identity holds on $D(H_{P,0})$.
\end{proof}

%

\bibliographystyle{alpha}

\end{document}